\documentclass[a4paper,UKenglish,cleveref,autoref,thm-restate]{lipics-v2021}
\title{On the role of connectivity in Linear Logic proofs}

\author{Raffaele {Di Donna}}{Université Paris Cité, France \and Università Roma Tre, Italy \and \url{https://www.irif.fr/~didonna/} }{didonna@irif.fr}{}{}%mandatory, please use full name; only 1 author per \author macro; first two parameters are mandatory, other parameters can be empty. Please provide at least the name of the affiliation and the country. The full address is optional. Use additional curly braces to indicate the correct name splitting when the last name consists of multiple name parts.

\author{Lorenzo {Tortora de Falco}}{Università Roma Tre, Italy \and \url{http://logica.uniroma3.it/~tortora/}}{tortora@uniroma3.it}{}{}

\authorrunning{R. Di Donna and L. Tortora de Falco} %mandatory. First: Use abbreviated first/middle names. Second (only in severe cases): Use first author plus 'et al.'

\Copyright{Raffaele Di Donna and Lorenzo Tortora de Falco} %mandatory, please use full first names. LIPIcs license is "CC-BY";  http://creativecommons.org/licenses/by/3.0/

 %mandatory; please add comma-separated list of keywords

\category{} %optional, e.g. invited paper

\relatedversion{} %optional, e.g. full version hosted on arXiv, HAL, or other respository/website
\acknowledgements{We thank Olivier Laurent and Willem~Heijltjes~for~discussions~and~suggestions.}%optional

\nolinenumbers %uncomment to disable line numbering

\usepackage{amsmath}
\usepackage{amsthm}
\usepackage{amssymb}
\usepackage{mathtools}
\usepackage{anyfontsize}
\usepackage{stmaryrd}
\SetSymbolFont{stmry}{bold}{U}{stmry}{m}{n}
\usepackage{cmll}
\usepackage{graphicx}
\usepackage[table]{xcolor}
\definecolor{headcolor}{RGB}{234,236,240}
\definecolor{rowcolor}{RGB}{248,249,250}
\usepackage{tikz}
\usetikzlibrary{calc}
\usepackage{tikzit}
\usepackage{bussproofs}
\usepackage{bussproofs-extra}
\def\infcstrut{\vrule width 0pt height 8.125pt depth 1.875pt}
\definecolor{myred}{RGB}{179,0,0}
\definecolor{myblue}{RGB}{0,0,179}
\definecolor{mygreen}{RGB}{0,128,0}
\definecolor{mypurple}{RGB}{222,0,174}
\definecolor{myyellow}{RGB}{252,222,30}
\definecolor{myorange}{RGB}{255,149,0}
\usepackage{hyperref}
\hypersetup{
	colorlinks,
	linkcolor={myred},
	citecolor={mygreen},
	urlcolor={myblue}
}
\usepackage{array,booktabs}
\usepackage{xspace}

\input{style.tikzdefs}
\tikzstyle{cell}=[fill=black, draw=black, shape=circle]
\tikzstyle{lowLeftBoxCorner}=[yshift=-0.8425ex]
\tikzstyle{lowRightBoxCorner}=[yshift=-0.885ex]
\tikzstyle{squary}=[anchor=north, rounded corners=2pt, inner sep=1.2pt, fill=black, tikzit fill={rgb,255: red,11; green,96; blue,255}, shape=rectangle]
\tikzstyle{linkDOWN}=[anchor=north, trapezium, trapezium angle=110, rounded corners=2pt, inner sep=1.2pt, fill=black, tikzit fill={rgb,255: red,11; green,96; blue,255}]
\tikzstyle{linkUP}=[anchor=north, trapezium, rotate=180, trapezium angle=110, rounded corners=2pt, inner sep=1.2pt, fill=black, tikzit fill={rgb,255: red,255; green,73; blue,76}]
\tikzstyle{agent}=[inner sep=1.5pt, tikzit fill={rgb,255: red,255; green,224; blue,46}]
\tikzstyle{tee}=[line width=0.3pt]
\tikzstyle{tensorLink}=[anchor=north, trapezium, trapezium angle=110, rounded corners=2pt, inner sep=1.2pt, fill=black, label={[text=white]:$\otimes$}, tikzit fill={rgb,255: red,12; green,255; blue,0}]
\tikzstyle{parrLink}=[anchor=north, trapezium, trapezium angle=110, rounded corners=2pt, inner sep=1.2pt, fill=black, label={[text=white]:$\parr$}, tikzit fill={rgb,255: red,0; green,255; blue,213}]
\tikzstyle{cutLink}=[anchor=north, trapezium, trapezium angle=110, rounded corners=2pt, inner sep=1.2pt, fill=black, label={[text=white]:$\scriptscriptstyle\mathtt{cut}$}, tikzit fill={rgb,255: red,11; green,96; blue,255}, shape=rectangle]
\tikzstyle{daimonLink}=[anchor=north, trapezium, rotate=180, trapezium angle=110, rounded corners=2pt, inner sep=1.2pt, fill=black, label={[text=white]:$\maltese$}, tikzit fill={rgb,255: red,255; green,73; blue,76}, shape=rectangle]
\tikzstyle{wlinkDOWN}=[draw=black, anchor=north, trapezium, trapezium angle=110, rounded corners=2pt, inner sep=1.2pt, fill=white, tikzit fill={rgb,255: red,100; green,120; blue,255}]
\tikzstyle{wlinkUP}=[anchor=north, trapezium, rotate=180, trapezium angle=110, rounded corners=2pt, inner sep=1.2pt, fill=white, tikzit fill={rgb,255: red,100; green,120; blue,255}]
\tikzstyle{wlinkROUND}=[draw=black, shape=circle, inner sep=1.2pt, fill=white, tikzit fill={rgb,255: red,180; green,120; blue,255}]
\tikzstyle{tensorwLink}=[draw=black, anchor=north, trapezium, trapezium angle=110, rounded corners=2pt, inner sep=1.2pt, fill=white, font={$\scriptstyle\varotimes$}, tikzit fill={rgb,255: red,12; green,255; blue,0}]
\tikzstyle{parrwLink}=[draw=black, anchor=north, trapezium, trapezium angle=110, rounded corners=2pt, inner sep=1.2pt, fill=white, font={$\scriptstyle\parr$}, tikzit fill={rgb,255: red,0; green,255; blue,213}]
\tikzstyle{cutwLink}=[draw=black, anchor=north, trapezium, trapezium angle=110, rounded corners=2pt, inner sep=1.2pt, fill=white, font={$\scriptstyle\mathtt{cut}$}, tikzit fill={rgb,255: red,11; green,96; blue,255}, tikzit shape=rectangle]
\tikzstyle{daimonwLink}=[draw=black, anchor=north, trapezium, rotate=180, trapezium angle=110, rounded corners=2pt, inner sep=1.2pt, fill=white, font={\scalebox{0.8}{$\scriptscriptstyle\maltese$}}, tikzit fill={rgb,255: red,255; green,73; blue,76}, tikzit shape=rectangle]
\tikzstyle{axwLink}=[draw=black, anchor=north, trapezium, trapezium angle=110, rounded corners=2pt, inner sep=1.2pt, fill=white, font={$\scriptstyle\mathtt{ax}$}, tikzit fill={rgb,255: red,255; green,73; blue,76}, tikzit shape=rectangle]
\tikzstyle{auxwLink}=[draw=black, anchor=north, trapezium, trapezium angle=110, rounded corners=2pt, inner sep=1.2pt, fill=white, font={$\scriptstyle\mathtt{p}$}, tikzit fill={rgb,255: red,11; green,96; blue,255}, tikzit shape=rectangle]
\tikzstyle{wnwLink}=[draw=black, anchor=north, trapezium, trapezium angle=110, rounded corners=2pt, inner sep=1.2pt, fill=white, font={$\scriptstyle?$}, tikzit fill={rgb,255: red,200; green,255; blue,90}]
\tikzstyle{ocwLink}=[draw=black, anchor=north, trapezium, trapezium angle=110, rounded corners=2pt, inner sep=1.2pt, fill=white, font={$\scriptstyle!$}, tikzit fill={rgb,255: red,0; green,255; blue,213}]
\tikzstyle{boxLink}=[draw=black, anchor=north, trapezium, trapezium angle=110, rounded corners=2pt, inner sep=1.2pt, fill=white, font={$\scriptstyle\mathtt{b}$}, tikzit fill={rgb,255: red,0; green,255; blue,213}]
\tikzstyle{botwLink}=[draw=black, anchor=north, trapezium, trapezium angle=110, rounded corners=2pt, inner sep=1.2pt, fill=white, font={$\scriptstyle\bot$}, tikzit fill={rgb,255: red,0; green,255; blue,213}]
\tikzstyle{onewLink}=[draw=black, anchor=north, trapezium, trapezium angle=110, rounded corners=2pt, inner sep=1.2pt, fill=white, font={$\scriptstyle{\mathtt{1}}$}, tikzit fill={rgb,255: red,0; green,255; blue,213}]
\tikzstyle{cellORANGE}=[fill={rgb,255: red,100; green,159; blue,255}, shape=circle]
\tikzstyle{cellYELLOW}=[fill={rgb,255: red,137; green,252; blue,195}, shape=circle]
\tikzstyle{cellVIOLET}=[fill={rgb,255: red,148; green,137; blue,252}, shape=circle]
\tikzstyle{cellROSE}=[fill={rgb,255: red,252; green,137; blue,228}, shape=circle]
\tikzstyle{DOTagent}=[font={$\bullet$}, tikzit fill={rgb,255: red,255; green,224; blue,46}]
\tikzstyle{contrLink}=[draw=black, anchor=north, trapezium, trapezium angle=110, rounded corners=2pt, inner sep=1.2pt, fill=white, font={$\scriptstyle\mathtt{?c}$}, tikzit fill={rgb,67: red,109; green,90; blue,90}]
\tikzstyle{derLink}=[draw=black, anchor=north, trapezium, trapezium angle=110, rounded corners=2pt, inner sep=1.2pt, fill=white, font={$\scriptstyle{\mathtt{?d}}$}, tikzit fill={rgb,67: red,109; green,90; blue,90}]
\tikzstyle{weakLink}=[draw=black, anchor=north, trapezium, trapezium angle=110, rounded corners=2pt, inner sep=1.2pt, fill=white, font={$\scriptstyle{\mathtt{?w}}$}, tikzit fill={rgb,67: red,109; green,90; blue,90}]
\tikzstyle{genwLink}=[draw=black, anchor=north, trapezium, trapezium angle=110, rounded corners=2pt, inner sep=1.2pt, fill=white, font={$\scriptstyle l$}, tikzit fill={rgb,255: red,0; green,255; blue,213}]
\tikzstyle{genwBlue}=[draw=blue!95!black, anchor=north, trapezium, trapezium angle=110, rounded corners=2pt, inner sep=1.2pt, fill=white, font={$\scriptstyle\textcolor{blue!95!black}{\ell}$}, tikzit fill={rgb,255: red,0; green,0; blue,255}]
\tikzstyle{tenRed}=[draw=red!80!black, anchor=north, trapezium, trapezium angle=110, rounded corners=2pt, inner sep=1.2pt, fill=white, font={$\scriptstyle\textcolor{red!80!black}{\varotimes}$}, tikzit fill={rgb,255: red,255; green,0; blue,0}]
\tikzstyle{parBlue}=[draw=blue!95!black, anchor=north, trapezium, trapezium angle=110, rounded corners=2pt, inner sep=1.2pt, fill=white, font={$\scriptstyle\textcolor{blue!95!black}{\parr}$}, tikzit fill={rgb,255: red,0; green,0; blue,255}]
\tikzstyle{parOrange}=[draw=orange, anchor=north, trapezium, trapezium angle=110, rounded corners=2pt, inner sep=1.2pt, fill=white, font={$\scriptstyle\textcolor{orange}{\parr}$}, tikzit fill={rgb,255: red,155; green,100; blue,0}]
\tikzstyle{cutRed}=[draw=red!80!black, anchor=north, trapezium, trapezium angle=110, rounded corners=2pt, inner sep=1.2pt, fill=white, font={$\scriptstyle\textcolor{red!80!black}{\mathtt{cut}}$}, tikzit fill={rgb,255: red,255; green,0; blue,0}, tikzit shape=rectangle]
\tikzstyle{axGreen}=[draw=green!60!black, anchor=north, trapezium, trapezium angle=110, rounded corners=2pt, inner sep=1.2pt, fill=white, font={$\scriptstyle\textcolor{green!60!black}{\mathtt{ax}}$}, tikzit fill={rgb,255: red,0; green,255; blue,0}, tikzit shape=rectangle]
\tikzstyle{boxGreen}=[draw=green!60!black, anchor=north, trapezium, trapezium angle=110, rounded corners=2pt, inner sep=1.2pt, fill=white, font={$\scriptstyle\textcolor{green!60!black}{\mathtt{b}}$}, tikzit fill={rgb,255: red,0; green,255; blue,0}]
\tikzstyle{botOrange}=[draw=orange, anchor=north, trapezium, trapezium angle=110, rounded corners=2pt, inner sep=1.2pt, fill=white, font={$\scriptstyle\textcolor{orange}{\bot}$}, tikzit fill={rgb,255: red,155; green,100; blue,0}]
\tikzstyle{oneGreen}=[draw=green!60!black, anchor=north, trapezium, trapezium angle=110, rounded corners=2pt, inner sep=1.2pt, fill=white, font={$\scriptstyle\textcolor{green!60!black}{\mathtt{1}}$}, tikzit fill={rgb,255: red,0; green,255; blue,0}]
\tikzstyle{contrBlue}=[draw=blue!95!black, anchor=north, trapezium, trapezium angle=110, rounded corners=2pt, inner sep=1.2pt, fill=white, font={$\scriptstyle\textcolor{blue!95!black}{\mathtt{?c}}$}, tikzit fill={rgb,255: red,0; green,0; blue,255}]
\tikzstyle{derOrange}=[draw=orange, anchor=north, trapezium, trapezium angle=110, rounded corners=2pt, inner sep=1.2pt, fill=white, font={$\scriptstyle\textcolor{orange}{\mathtt{?d}}$}, tikzit fill={rgb,67: red,155; green,100; blue,0}]
\tikzstyle{weakOrange}=[draw=orange, anchor=north, trapezium, trapezium angle=110, rounded corners=2pt, inner sep=1.2pt, fill=white, font={$\scriptstyle\textcolor{orange}{\mathtt{?w}}$}, tikzit fill={rgb,255: red,155; green,100; blue,0}]
\tikzstyle{DOTorange}=[font={$\textcolor{orange}{\bullet}$}, tikzit fill={rgb,255: red,155; green,100; blue,0}]
\tikzstyle{parrwLinkHighlighted}=[draw=black, anchor=north, trapezium, trapezium angle=110, rounded corners=2pt, inner sep=1.2pt, fill=pink, font={$\scriptstyle\parr$}, tikzit fill={rgb,255: red,0; green,255; blue,213}]
\tikzstyle{botwLinkHighlighted}=[draw=black, anchor=north, trapezium, trapezium angle=110, rounded corners=2pt, inner sep=1.2pt, fill=pink, font={$\scriptstyle\bot$}, tikzit fill={rgb,255: red,0; green,255; blue,213}]
\tikzstyle{whiteTest}=[text=white]

\tikzstyle{linkInput}=[->, line width=0.2pt, draw opacity=0, latex reversed-, postaction={draw, opacity=1, -, line width=1pt, shorten <=1pt}, tikzit draw={rgb,255: red,245; green,75; blue,200}]
\tikzstyle{linkOutput}=[->, draw opacity=0, -latex, line width=0.2pt, postaction={draw, opacity=1, -, line width=1pt, shorten >=3pt}]
\tikzstyle{filler}=[-, fill={rgb,255: red,238; green,238; blue,238}, draw=none, fill opacity=0.5, tikzit draw=black]
\tikzstyle{fillerTransp}=[-, fill={rgb,255: red,238; green,238; blue,238}, draw=black, fill opacity=0, tikzit draw=black]
\tikzstyle{fillerRed}=[-, fill={rgb,255: red,100; green,159; blue,255}, draw=none, fill opacity=0.5, tikzit draw=black]
\tikzstyle{fillerDARK}=[-, fill={rgb,255: red,100; green,255; blue,255}, draw=none, fill opacity=0.5, tikzit draw=black]
\tikzstyle{fillerGreen}=[-, fill={rgb,255: red,252; green,137; blue,228}, draw=none, fill opacity=0.5, tikzit draw=black]
\tikzstyle{fillerYellow}=[-, fill={rgb,255: red,137; green,252; blue,195}, draw=none, fill opacity=0.5, tikzit draw=black]
\tikzstyle{dottedEdge}=[-, dotted, draw=black]
\tikzstyle{dottedArrow}=[draw=black, ->, dashed]
\tikzstyle{Dsimple}=[draw=black, ->, draw opacity=0, latex reversed-, line width=0.2pt, postaction={draw, opacity=1, -, line width=1pt, shorten >=-5pt}, tikzit draw={rgb,255: red,100; green,75; blue,200}]
\tikzstyle{DsimpleB}=[draw=black, ->, draw opacity=0, latex reversed-, line width=0.2pt, postaction={draw, opacity=1, -, line width=1pt}, tikzit draw={rgb,255: red,100; green,75; blue,200}]
\tikzstyle{simple}=[draw=black, -, draw opacity=0, -latex, line width=0.2pt, postaction={draw, opacity=1, -, line width=1pt, shorten >=3pt, shorten <=-1pt}, tikzit draw={rgb,255: red,100; green,75; blue,200}]
\tikzstyle{simpleB}=[draw=black, -, draw opacity=0, -latex, line width=0.2pt, postaction={draw, opacity=1, -, line width=1pt, shorten >=-1pt, shorten <=-1pt}, shorten >=-4pt, tikzit draw={rgb,255: red,100; green,75; blue,200}]
\tikzstyle{simpleC}=[draw=black, -, draw opacity=0, -latex, line width=0.2pt, postaction={draw, opacity=1, -, line width=1pt, shorten >=3pt, shorten <=1pt}, tikzit draw={rgb,255: red,100; green,75; blue,200}]
\tikzstyle{simpleD}=[draw=black, -, draw opacity=0, -latex, line width=0.2pt, postaction={draw, opacity=1, -, line width=1pt, shorten >=3pt}, tikzit draw={rgb,255: red,100; green,75; blue,200}]
\tikzstyle{reduction}=[->]
\tikzstyle{simpleOrange}=[->, color=orange, tikzit draw={rgb,255: red,255; green,93; blue,0}]
\tikzstyle{simpleBorange}=[draw=orange, -, draw opacity=0, -latex, line width=0.2pt, postaction={draw, opacity=1, -, line width=1pt, shorten >=-1pt, shorten <=-1pt}, shorten >=-4pt, tikzit draw={rgb,255: red,255; green,128; blue,0}]
\tikzstyle{implicarrow}=[->, draw=black, double]
\tikzstyle{thickedge}=[line width=2pt, draw=black, -]
\tikzstyle{DoutLEFT}=[out=180, in=90, looseness=0.85, ->, draw opacity=0, -latex, line width=0.2pt, postaction={draw, opacity=1, -, line width=1pt, shorten <=-1pt, shorten >=3pt}, tikzit draw={rgb,255: red,100; green,75; blue,200}]
\tikzstyle{DoutRIGHT}=[out=00, in=90, looseness=0.85, ->, draw opacity=0, -latex, line width=0.2pt, postaction={draw, opacity=1, -, line width=1pt, shorten <=-1pt, shorten >=3pt}, tikzit draw={rgb,255: red,245; green,200; blue,100}]
\tikzstyle{DinLEFT}=[out=-90, in=180, looseness=0.85, ->, draw opacity=0, latex reversed-, line width=0.2pt, postaction={draw, opacity=1, -, line width=1pt, shorten <=1pt, shorten >=-1pt}, tikzit draw={rgb,255: red,100; green,75; blue,200}]
\tikzstyle{DinRIGHT}=[out=-90, in=00, looseness=0.85, ->, draw opacity=0, latex reversed-, line width=0.2pt, postaction={draw, opacity=1, -, line width=1pt, shorten <=1pt, shorten >=-1pt}, tikzit draw={rgb,255: red,245; green,200; blue,100}]
\tikzstyle{Dout1LEFT}=[out=-90, in=180, looseness=0.85, ->, draw opacity=0, -latex, line width=0.2pt, postaction={draw, opacity=1, -, line width=1pt, shorten >=3pt}, tikzit draw={rgb,255: red,100; green,75; blue,200}]
\tikzstyle{Dout1RIGHT}=[out=-90, in=00, looseness=0.85, ->, draw opacity=0, -latex, line width=0.2pt, postaction={draw, opacity=1, -, line width=1pt, shorten >=3pt}, tikzit draw={rgb,255: red,245; green,200; blue,100}]
\tikzstyle{outLEFT}=[out=180, in=90, looseness=0.85, -, draw opacity=0, -latex, line width=0.2pt, postaction={draw, opacity=1, -, line width=1pt, shorten >=3pt, shorten <=-1pt}, tikzit draw={rgb,255: red,100; green,75; blue,200}]
\tikzstyle{outRIGHT}=[out=00, in=90, looseness=0.85, -, draw opacity=0, -latex, line width=0.2pt, postaction={draw, opacity=1, -, line width=1pt, shorten >=3pt, shorten <=-1pt}, tikzit draw={rgb,255: red,245; green,200; blue,100}]
\tikzstyle{EoutLEFT}=[out=180, in=90, looseness=0.85, -, draw opacity=0, -latex, line width=0.2pt, postaction={draw, opacity=1, -, line width=1pt, shorten >=-1pt, shorten <=-1pt}, shorten >=-4pt, tikzit draw={rgb,255: red,100; green,75; blue,200}]
\tikzstyle{EoutRIGHT}=[out=00, in=90, looseness=0.85, -, draw opacity=0, -latex, line width=0.2pt, postaction={draw, opacity=1, -, line width=1pt, shorten >=-1pt, shorten <=-1pt}, shorten >=-4pt, tikzit draw={rgb,255: red,245; green,200; blue,100}]
\tikzstyle{DoACLEFT}=[out=180, in=0, looseness=0.85, ->, draw opacity=0, -latex, line width=0.2pt, postaction={draw, opacity=1, -, line width=1pt, shorten <=-1pt, shorten >=3pt}, tikzit draw={rgb,255: red,100; green,75; blue,200}]
\tikzstyle{DoACRIGHT}=[out=0, in=180, looseness=0.85, ->, draw opacity=0, -latex, line width=0.2pt, postaction={draw, opacity=1, -, line width=1pt, shorten <=-1pt, shorten >=3pt}, tikzit draw={rgb,255: red,245; green,200; blue,100}]
\tikzstyle{DoAC1LEFT}=[out=180, in=180, looseness=0.85, ->, draw opacity=0, -latex, line width=0.2pt, postaction={draw, opacity=1, -, line width=1pt, shorten <=-1pt, shorten >=3pt}, tikzit draw={rgb,255: red,100; green,75; blue,200}]
\tikzstyle{DoAC1RIGHT}=[out=0, in=0, looseness=0.85, ->, draw opacity=0, -latex, line width=0.2pt, postaction={draw, opacity=1, -, line width=1pt, shorten <=-1pt, shorten >=3pt}, tikzit draw={rgb,255: red,245; green,200; blue,100}]
\tikzstyle{inLEFT}=[out=-90, in=180, looseness=0.85, -, draw opacity=0, latex reversed-, line width=0.2pt, postaction={draw, opacity=1, -, line width=1pt, shorten <=1pt, shorten >=-1pt}, tikzit draw={rgb,255: red,100; green,75; blue,200}]
\tikzstyle{inRIGHT}=[out=-90, in=00, looseness=0.85, -, draw opacity=0, latex reversed-, line width=0.2pt, postaction={draw, opacity=1, -, line width=1pt, shorten <=1pt, shorten >=-1pt}, tikzit draw={rgb,255: red,245; green,200; blue,100}]
\tikzstyle{out1LEFT}=[out=-90, in=180, looseness=0.85, -, draw opacity=0, -latex, line width=0.2pt, postaction={draw, opacity=1, -, line width=1pt, shorten >=3pt}, tikzit draw={rgb,255: red,100; green,75; blue,200}]
\tikzstyle{out1LEFTorange}=[draw=orange, out=-90, in=180, looseness=0.85, -, draw opacity=0, -latex, line width=0.2pt, postaction={draw, opacity=1, -, line width=1pt, shorten >=3pt}, tikzit draw={rgb,255: red,255; green,128; blue,0}]
\tikzstyle{out1RIGHT}=[out=-90, in=00, looseness=0.85, -, draw opacity=0, -latex, line width=0.2pt, postaction={draw, opacity=1, -, line width=1pt, shorten >=3pt}, tikzit draw={rgb,255: red,245; green,200; blue,100}]
\tikzstyle{out1LEFThighlight}=[draw=red, out=-90, in=180, looseness=0.85, -, draw opacity=0, -latex, line width=0.3pt, postaction={draw, opacity=1, -, line width=1.5pt, shorten >=3pt}, tikzit draw={rgb,255: red,255; green,0; blue,0}]
\tikzstyle{out1RIGHThighlight}=[draw=red, out=-90, in=00, looseness=0.85, -, draw opacity=0, -latex, line width=0.3pt, postaction={draw, opacity=1, -, line width=1.5pt, shorten >=3pt}, tikzit draw={rgb,255: red,245; green,10; blue,0}]
\tikzstyle{up1RIGHThighlight}=[draw=red, out=0, in=270, looseness=0.85, -, draw opacity=0, -latex, line width=0.3pt, postaction={draw, opacity=1, -, line width=1.5pt, shorten >=3pt, shorten <=-0.3pt}, tikzit draw={rgb,255: red,225; green,10; blue,20}]
\tikzstyle{simpleORANGE}=[draw={rgb,255: red,100; green,159; blue,255}, ->, -latex]
\tikzstyle{simpleYELLOW}=[draw={rgb,255: red,137; green,252; blue,195}, ->, -latex]
\tikzstyle{simpleVIOLET}=[draw={rgb,255: red,148; green,137; blue,252}, ->, -latex]
\tikzstyle{simpleROSE}=[draw={rgb,255: red,252; green,137; blue,228}, ->, -latex]
\tikzstyle{dottedORANGE}=[draw={rgb,255: red,100; green,159; blue,255}, ->, -latex, dashed]
\tikzstyle{dottedYELLOW}=[draw={rgb,255: red,137; green,252; blue,195}, ->, -latex, dashed]
\tikzstyle{dottedVIOLET}=[draw={rgb,255: red,148; green,137; blue,252}, ->, -latex, dashed]
\tikzstyle{dottedROSE}=[draw={rgb,255: red,252; green,137; blue,228}, ->, -latex, dashed]
\tikzstyle{lsimple}=[draw=black, ->, -latex]
\tikzstyle{ldottedArrow}=[draw=black, ->, -latex, dashed]
\tikzstyle{roundedCornerBlackFill}=[draw=black, line width=1pt, -, fill={rgb,255: red,0; green,0; blue,0}, rounded corners, tikzit draw=black]

\makeatletter
\tikzset{
	dot diameter/.store in=\dot@diameter,
	dot diameter=0.25ex,
	dot spacing/.store in=\dot@spacing,
	dot spacing=1ex,
	dots/.style={
		line width=\dot@diameter,
		line cap=round,
		dash pattern=on 0pt off \dot@spacing
	}
}
\makeatother

\tikzset{
	deduceLine/.style = {line width=1.1pt, dots}
}

\declaretheorem[sibling=theorem]{notation}

\makeatletter
\def\@textbottom{\vskip \z@ \@plus 2pt}
\let\@texttop\relax
\makeatother

\makeatletter
\g@addto@macro{\UrlBreaks}{%
	\do\/\do\a\do\b\do\c\do\d\do\e\do\f%
	\do\g\do\h\do\i\do\j\do\k\do\l\do\m%
	\do\n\do\o\do\p\do\q\do\r\do\s\do\t%
	\do\u\do\v\do\w\do\x\do\y\do\z%
	\do\A\do\B\do\C\do\D\do\E\do\F\do\G%
	\do\H\do\I\do\J\do\K\do\L\do\M\do\N%
	\do\O\do\P\do\Q\do\R\do\S\do\T\do\U%
	\do\V\do\W\do\X\do\Y\do\Z}
\makeatother

\makeatletter
\g@addto@macro\bfseries{\boldmath}
\def\thmhead@plain#1#2#3{%
	\thmname{#1}\thmnumber{\@ifnotempty{#1}{ }\@upn{#2}}%
	\thmnote{ {\the\thm@notefont\unboldmath(#3)}}}
\let\thmhead\thmhead@plain
\makeatother

\makeatletter
\let\save@mathaccent\mathaccent
\newcommand*\if@single[3]{%
	\setbox0\hbox{${\mathaccent"0362{#1}}^H$}%
	\setbox2\hbox{${\mathaccent"0362{\kern0pt#1}}^H$}%
	\ifdim\ht0=\ht2 #3\else #2\fi
}
\newcommand*\rel@kern[1]{\kern#1\dimexpr\macc@kerna}
\newcommand*\widebar[1]{\@ifnextchar^{{\wide@bar{#1}{0}}}{\wide@bar{#1}{1}}}
\newcommand*\wide@bar[2]{\if@single{#1}{\wide@bar@{#1}{#2}{1}}{\wide@bar@{#1}{#2}{2}}}
\newcommand*\wide@bar@[3]{%
	\begingroup
	\def\mathaccent##1##2{%
		\let\mathaccent\save@mathaccent
		\if#32 \let\macc@nucleus\first@char \fi
		\setbox\z@\hbox{$\macc@style{\macc@nucleus}_{}$}%
		\setbox\tw@\hbox{$\macc@style{\macc@nucleus}{}_{}$}%
		\dimen@\wd\tw@
		\advance\dimen@-\wd\z@
		\divide\dimen@ 3
		\@tempdima\wd\tw@
		\advance\@tempdima-\scriptspace
		\divide\@tempdima 10
		\advance\dimen@-\@tempdima
		\ifdim\dimen@>\z@ \dimen@0pt\fi
		\rel@kern{0.6}\kern-\dimen@
		\if#31
		\overline{\rel@kern{-0.6}\kern\dimen@\macc@nucleus\rel@kern{0.4}\kern\dimen@}%
		\advance\dimen@0.4\dimexpr\macc@kerna
		\let\final@kern#2%
		\ifdim\dimen@<\z@ \let\final@kern1\fi
		\if\final@kern1 \kern-\dimen@\fi
		\else
		\overline{\rel@kern{-0.6}\kern\dimen@#1}%
		\fi
	}%
	\macc@depth\@ne
	\let\math@bgroup\@empty \let\math@egroup\macc@set@skewchar
	\mathsurround\z@ \frozen@everymath{\mathgroup\macc@group\relax}%
	\macc@set@skewchar\relax
	\let\mathaccentV\macc@nested@a
	\if#31
	\macc@nested@a\relax111{#1}%
	\else
	\def\gobble@till@marker##1\endmarker{}%
	\futurelet\first@char\gobble@till@marker#1\endmarker
	\ifcat\noexpand\first@char A\else
	\def\first@char{}%
	\fi
	\macc@nested@a\relax111{\first@char}%
	\fi
	\endgroup
}
\makeatother

\DeclareFontFamily{U}{mathx}{}
\DeclareFontShape{U}{mathx}{m}{n}{<-> mathx10}{}
\DeclareSymbolFont{mathx}{U}{mathx}{m}{n}
\DeclareMathAccent{\widehat}{0}{mathx}{"70}
\DeclareMathAccent{\widecheck}{0}{mathx}{"71}

\renewcommand{\epsilon}{\varepsilon}

\theoremstyle{definition}
\newtheorem*{outline}{Outline}
\newtheorem*{notations}{Notations}
\newtheorem*{convention}{Convention}
\newtheorem*{fw}{Future work}

\EventEditors{John Q. Open and Joan R. Access}
\EventNoEds{2}
\EventLongTitle{42nd Conference on Very Important Topics (CVIT 2016)}
\EventShortTitle{CVIT 2016}
\EventAcronym{CVIT}
\EventYear{2016}
\EventDate{December 24--27, 2016}
\EventLocation{Little Whinging, United Kingdom}
\EventLogo{}
\SeriesVolume{42}
\ArticleNo{23}
\begin{document}
	
 \maketitle
 
 % Proof theory
\newcommand{\cutelim}{cut elimination\xspace}
\newcommand{\Cutelim}{Cut elimination\xspace}
\newcommand{\cf}{cut-free\xspace}
\newcommand{\Cf}{Cut-free\xspace}
\newcommand{\one}{\boldsymbol{1}}

% Fragments
\newcommand{\linl}{\mathit{LL}}
\newcommand{\mell}{\mathit{MELL}}
\newcommand{\mllu}{\mathit{MLL}}
\newcommand{\mll}{\mllu^-}
\newcommand{\comll}{\mathit{COMLL}}
\newcommand{\imell}{\mathit{IMELL}}
\newcommand{\imll}{\mathit{IMLL}}
\newcommand{\icomll}{\mathit{ICOMLL}}
\newcommand{\wnten}{(\neg ? \otimes)}
\newcommand{\wntenll}{\wnten\linl}
\newcommand{\bten}{(\neg \bot \otimes)}
\newcommand{\btenll}{\bten\linl}
\newcommand{\btenlls}{\btenll^*}
\newcommand{\llpol}{\linl_\mathit{pol}}
\newcommand{\lamlinl}{\lambda\linl}

% Sequent calculus
\newcommand{\seqc}{sequent calculus\xspace}
\newcommand{\Seqc}{Sequent calculus\xspace}
\newcommand{\mix}{$\mathit{mix}$\xspace}
\newcommand{\axsc}{$\mathit{ax}$\xspace}
\newcommand{\cutsc}{$\mathit{cut}$\xspace}
\newcommand{\drsc}{$?\mathit{d}$\xspace}
\newcommand{\wksc}{$?\mathit{w}$\xspace}
\newcommand{\ctsc}{$?\mathit{c}$\xspace}
\newcommand{\exsc}{$\mathit{ex}$\xspace}
\newcommand{\equivseqc}{\sim}
\newcommand{\ol}{$\out$-light\xspace}

% Proof-nets
\newcommand{\ps}{proof-structure\xspace}
\newcommand{\Ps}{Proof-structure\xspace}
\newcommand{\pss}{proof-structures\xspace}
\newcommand{\Pss}{Proof-structures\xspace}
\newcommand{\pn}{proof-net\xspace}
\newcommand{\Pn}{Proof-net\xspace}
\newcommand{\pns}{proof-nets\xspace}
\newcommand{\Pns}{Proof-nets\xspace}
\newcommand{\cc}{\mathsf{cc}}
\newcommand{\w}{\mathsf{w}}
\newcommand{\AC}{\mathsf{AC}}
\newcommand{\C}{\mathsf{C}}
\newcommand{\Cw}{\C_{\sharp \w}}
\newcommand{\Prp}{\mathsf{P}}
\newcommand{\ACC}{\AC\C}
\newcommand{\ACCw}{\AC\Cw}
\newcommand{\sg}{switching graph\xspace}
\newcommand{\Sg}{Switching graph\xspace}
\newcommand{\sgs}{switching graphs\xspace}
\newcommand{\Sgs}{Switching graphs\xspace}
\newcommand{\wten}{(\neg\w\otimes)}
\newcommand{\llg}{linear logical graph\xspace}
\newcommand{\Llg}{Linear logical graph\xspace}
\newcommand{\llgs}{linear logical graphs\xspace}
\newcommand{\Llgs}{Linear logical graphs\xspace}
\newcommand{\seq}{sequential\xspace}
\newcommand{\Seq}{Sequential\xspace}
\newcommand{\er}{erasing\xspace}
\newcommand{\Er}{Erasing\xspace}
\newcommand{\ve}{\mathsf{n}}
\newcommand{\ar}{\mathsf{a}}
\newcommand{\dpt}{\mathsf{d}}
\newcommand{\des}{\circ}
\newcommand{\cwf}{\C_\w^\forall}
\newcommand{\axpn}{\mathtt{ax}}
\newcommand{\cutpn}{\mathtt{cut}}
\newcommand{\onepn}{\mathtt{1}}
\newcommand{\drpn}{?\mathtt{d}}
\newcommand{\wkpn}{?\mathtt{w}}
\newcommand{\ctpn}{?\mathtt{c}}
\newcommand{\dpath}{descent path\xspace}
\newcommand{\Dpath}{Descent path\xspace}
\newcommand{\dpaths}{descent paths\xspace}
\newcommand{\Dpaths}{Descent paths\xspace}
\newcommand{\axcut}{axiom cut\xspace}
\newcommand{\Axcut}{Axiom cut\xspace}
\newcommand{\axcuts}{axiom cuts\xspace}
\newcommand{\Axcuts}{Axiom cuts\xspace}
\newcommand{\ucut}{unit cut\xspace}
\newcommand{\Ucut}{Unit cut\xspace}
\newcommand{\ucuts}{unit cuts\xspace}
\newcommand{\Ucuts}{Unit cuts\xspace}
\newcommand{\mcut}{multiplicative cut\xspace}
\newcommand{\Mcut}{Multiplicative cut\xspace}
\newcommand{\mcuts}{multiplicative cuts\xspace}
\newcommand{\Mcuts}{Multiplicative cuts\xspace}
\newcommand{\spath}{switching path\xspace}
\newcommand{\Spath}{Switching path\xspace}
\newcommand{\spaths}{switching paths\xspace}
\newcommand{\Spaths}{Switching paths\xspace}
\newcommand{\deseq}{\circ}
\newcommand{\deseqjump}{\Mapsto}
\newcommand{\tocut}{\to}
\newcommand{\ws}{$\w$-switching\xspace}
\newcommand{\wsp}{\ws path\xspace}
\newcommand{\wsps}{\ws paths\xspace}
\newcommand{\wc}{$\w$-compatible\xspace}
\newcommand{\wcs}{\wc switching\xspace}
\newcommand{\wcss}{\wc switchings\xspace}
\newcommand{\thr}{erasing thread\xspace}
\newcommand{\Cwf}{\mathsf{C}_\w^\forall}
\newcommand{\jf}{jump-free\xspace}
\newcommand{\Jf}{Jump-free\xspace}
\newcommand{\jt}{jump-total\xspace}
\newcommand{\Jt}{Jump-total\xspace}
\newcommand{\jc}{jump-correct\xspace}
\newcommand{\Jc}{Jump-correct\xspace}
\newcommand{\ujf}[1]{\langle #1 \rangle}
\newcommand{\wparr}{\w_\parr}
\newcommand{\we}{\w_{\mathsf{e}}}
\newcommand{\inp}{\mathsf{i}}
\newcommand{\wi}{\w_{\mkern2mu\inp}}
\newcommand{\can}[1]{\ujf{#1}_\w}
\newcommand{\canext}[2]{\ujf{#1}_\w^#2}
\newcommand{\rewpns}{\overset{1}{\approx}}
\newcommand{\equivpns}{\approx}
\newcommand{\out}{\mathsf{o}}
\newcommand{\rob}{{\onepn/\bot}}
\newcommand{\rtp}{{\otimes/\parr}}
\newcommand{\intu}{intuitionistic\xspace}
\newcommand{\Intu}{Intuitionistic\xspace}
\newcommand{\is}{\intu switching\xspace}
\newcommand{\iss}{\intu switchings\xspace}
\newcommand{\Is}{\Intu switching\xspace}
\newcommand{\Iss}{\Intu switchings\xspace}
\newcommand{\nof}[1]{n_{#1}}
\newcommand{\ta}[1]{#1_\otimes}
\newcommand{\up}[3]{\theta_{#1,#2}^#3}
\newcommand{\less}[1]{<_#1}
\newcommand{\pa}{\mathsf{p}}

% Others
\newcommand{\lc}{$\lambda$-calculus\xspace}
\newcommand{\dom}{\mathsf{dom}}

% Fragments
\newcommand{\vimell}{$\mkern-2mu\mathit{VIMELL}$\xspace}
 
 \begin{abstract}
  We investigate a property that extends the Danos-Regnier correctness criterion for linear logic \pss. The property applies to the correctness graphs of a \ps: it states that any such graph is acyclic and the number of its connected components is exactly one more than the number of nodes bottom or weakening. This is known to be necessary but not sufficient in multiplicative exponential linear logic to recover a \seqc proof from a \ps. We present a geometric condition on untyped \pss allowing us to turn this necessary property into a sufficient one: we can thus isolate fragments of linear logic for which this property is indeed a correctness criterion. In a suitable fragment of multiplicative linear logic with units, the criterion yields a characterization of the equivalence induced by permutations of rules in \seqc. In intuitionistic linear logic, the property is equivalent to the familiar requirement of~having~exactly~one output conclusion, and it is sufficient for sequentialization in the axiom-free setting.
 \end{abstract}
 
 \section{Introduction}
  % %TODO
% Note:
% \begin{itemize}
%  \item
%   Bisogna giustificare la rilevanza di $(\neg \bot \otimes)$$\linl$, e dire che in generale, in $\mllu$, il quoziente indotto dalle reti è più grande di quello indotto dalle permutazioni di regole. Di conseguenza, è facile decidere l'equivalenza di due dimostrazioni in calcolo dei sequenti in $(\neg \bot \otimes)$$\linl$, contrariamente a quanto accade in $\mllu$;
%  \item
%   Bisogna parlare della non stabilità per eliminazione del taglio;
%  \item
%   Come presentare $\ACCw$: abbiamo trovato una condizione naturale $\wten$~sotto la quale è sufficiente. Successivamente, abbiamo indagato $\ACCw$ indipendentemente da $\wten$. Il primo passo è stato $\imell$, in cui, per le strutture soddisfacenti $\AC$, $\Cw$ equivale all'unicità della conclusione output.
% \end{itemize}
% 
% ------------------------------------

 It is widely acknowledged that one of the main contributions of Linear Logic ($\linl$~\cite{girard1987linear}) to the field of proof theory is the shift from a term-like or tree-like representation of proofs to a more general \emph{graph}-like one. Graphs representing proofs in $\linl$ are usually called \emph{\pns}: their syntax is richer and more expressive than the one of the \lc, or the one of \seqc, or natural deduction. Contrary to the mentioned formalisms, \pns are special inhabitants of the wider land of \emph{\pss}. A \ps is any ``graph'' that can be built in the language of \pns, and it need not represent a proof. \Pss are well-behaved for performing computations: right from the start, it appeared clearly that the procedure of \cutelim can be applied to \pss, and that they can be interpreted in denotational models thanks to the notion of experiment (\cite{girard1987linear}).

 Two of the main research lines in this area, to which nowadays one often refers as ``the theory of \pns'', are the study of their rewriting properties, like confluence and the normalization of the \cutelim procedure (\cite{girard1987linear,danos1990logique,regnier1992lambda,TdF2000,TdFAdditives03,Pag-TdF2010,GueManTdFVaux,accattoli:LIPIcs.RTA.2013.39}, ...) and the study of \emph{correctness criteria} for \pss: it is obviously an interesting theoretical question to find abstract properties allowing us to~isolate,~among~\pss,~exactly~those~that~are indeed proofs. The present work is a contribution~to~this~second~line~of~research.
 % in the theory of \pns.

 %After the inception in the landscape of this new representation of proofs, 
 A strong divergence appeared: on one hand the ideal (but restricted) purely multiplicative fragment $\mll$, with many correctness criteria  (\cite{girard1987linear,danos1989structure,rbpn,pnehrhard,Metayer1994,Bechet98}, ...), and on the other hand the wild ``rest of the world'' for which no really convincing correctness criterion has ever been found. Some rare notable exceptions are the multiplicative fragment with quantifiers (\cite{quantif2}) and the classical polarized fragment (\cite{laurent2002etude}); we can also quote the multiplicative~additive~fragment (\cite{pnJYGadditives,mallpnlong}), provided however one is not too picky.

 The property $\ACC$ (a.k.a.~Danos-Regnier criterion) states that every ``correctness graph'' of the \ps (called \sg in this paper, \Cref{def:switching}) is acyclic and connected: it is one of the most famous correctness criteria for $\mll$ (\cite{danos1989structure}). In the multiplicative exponential fragment $\mell$ (the fragment of $\linl$ in which it is natural to represent and study the majority of logical and computational phenomena), one often considers \pss satisfying the property $\AC$, obtained from $\ACC$ by dropping connectivity. Like $\ACC$ for \pss of $\mll$, $\AC$ is stable with respect to \cutelim of \pss of $\mell$; but $\AC$ allows us to isolate, among \pss of $\mell$, the ones that are proofs only in an extended sense: essentially, the \seqc proofs in $\mell$ with (some variants of) the \mix rule (\cite{Fleury1994TheMR,llhandbook}). The several attempts to add some other condition(s) to $\AC$ in order to obtain a correctness criterion for $\mell$ (without \mix) were never really successful, and for very good reasons: a straightforward consequence of the $\mathit{NP}$-hardness of provability in the ``constant-only multiplicative fragment'' $\comll$ proven in~\cite{LincolnW94} is that no feasible correctness criterion can ever be found as soon as one adds the multiplicative units to $\mll$ (the fragment thus obtained is denoted in this paper by $\mllu$, \Cref{sec:ps}). Indeed, a formula $A$ of $\comll$ yields a unique \ps ($A$'s tree) and were we able to ``easily''~determine~whether~or~not~the \ps is~correct~we~could~``easily''~decide~the~provability of $A$.
 
 % The property $\ACC$ (also called Danos-Regnier criterion) states that every ``correctness graph'' of the structure is acyclic and connected: it is one of the most famous correctness criteria for $\mll$ (\cite{danos1989structure}). In the multiplicative and exponential fragment $\mell$ (the fragment of $\linl$ in which it is natural to represent and study the majority of logical and computational phenomena), one often considers \pss satisfying the property $\AC$, obtained from $\ACC$ by dropping connectivity. Like $\ACC$ for \pss of $\mll$, $\AC$ is stable wrt \cutelim of \pss of $\mell$; but $\AC$ allows us to isolate, among \pss of $\mell$, the ones that are proofs only in an extended sense: essentially, the proofs of $\mell$ \seqc with (some variants) of the \mix rule (for example~\cite{Fleury1994TheMR,llhandbook}). The several attempts to add some other condition(s) to $\AC$ in order to obtain a correctness criterion for $\mell$ (without \mix) were never really successful, and for very good reasons: it has been eventually proven in~\cite{heijltjes2014no} that no~feasible~correctness~criterion~can~ever~be~found~as~soon~as~one~adds~just the multiplicative units to $\mll$ (this fragment is denoted by $\mllu$). 

 From a conceptual point of view, the inception of \pns and correctness criteria in proof-theory is a real breakthrough. Most notably, the Danos-Regnier correctness criterion relates a purely logical property (to be a logical proof) and a purely geometric one (\sgs are trees). Since connectivity is undeniably a simple and solid geometric property, one can try to reverse the point of view and address the question of the logical meaning of connectivity: that's our main motivation for this work. For example, in~\cite{guerrieri_et_al:LIPIcs.FSCD.2016.20} a remarkable property (that does not hold for general \pss of $\mell$) is proven for \emph{connected} \pss of $\mell$: the element of order $2$ of the Taylor expansion of a \ps is enough to entirely rebuild the \ps. In the same spirit, in this paper we revisit an old property which is in between $\AC$ and $\ACC$. Connectivity fails in $\mell$ for the presence of the weakening rule of logic and of the rule for the multiplicative constant bottom: it is then natural to relate the number of connected components ($\cc$) and the nodes weakening or bottom ($\w$) through the equality $\smash{\cc = \w + 1}$. We thus obtain the property denoted by $\ACCw$ in this paper: to our knowledge it was first considered in~\cite{regnier1992lambda}, where the author shows that it is a necessary property for a \ps to be a real proof, but it is not sufficient due to the counterexample in \Cref{subfig:counterexample-mllu}. It is known in the community that $\ACCw$ is stable with respect to \cutelim of \pss of $\mell$ (\Cref{thm:stability-cut-elimination}: in the absence of a precise reference, we give a full proof in \Cref{subsec:cut-elimination-steps}). It is common practice (\cite{TdF2000,lamarcheIMELL,llhandbook}, ...) to add ``jumps'' to \pss in order to recover the property $\ACC$, and well-known are also the drawbacks of such a technique (there is no canonical position for jumps in general, their behaviour through \cutelim is a mess, ...). Thus, following the idea that simple and solid geometric properties are likely to have a logical meaning, we wondered whether it was possible to identify interesting~subsets~of~\pss~of~$\mell$~for~which~$\ACCw$~is~a~correctness criterion.

 We found a geometric condition on \pss of $\mell$, which we called $\wten$,\footnote{The notation $\wten$ suggests that no premise of a $\otimes$ node is the conclusion of a weakening (\er~node).} such that every $\wten$-\ps satisfies $\ACCw$ if and only if it is a proof: this is our \Cref{thm:untyped-sequentiality}, which is obtained in an untyped framework. Indeed, in recent years it has become quite common in the theory of \pns to drop types when they are not needed: correctness criteria allow us to present \pss in a sequential way (\Cref{sec:ps}), which, in presence of types, yields the traditional ``sequentialization theorems'' describing how a proof of \seqc can be built from a \ps satisfying some correctness criterion; but the sequentialization process itself has nothing to do with types (this was known, but maybe clearly expressed for the first time in~\cite{llhandbook}). The condition $\wten$ should be thought as a \emph{local} condition in presence of which $\ACCw$ is enough to sequentialize a \ps of $\mell$. One easily checks that $\wten$ is not stable with respect to \cutelim (contrary to $\ACCw$), and we use types in \Cref{sec:btenll,sec:mell} to present some interesting fragments of $\mell$ the \cf \pss of which are $\wten$-\pss, so that our theorem applies. We also argue (\Cref{prop:sequentiality-cut-elimination}) why it is interesting to have a correctness criterion for \cf \pss (by the way, this is not new in the theory of \pns: for example~\cite{DBLP:journals/apal/AbrusciR99} in the non-commutative case and~\cite{LTdF04} in presence of the additives). To simplify the presentation, we decided to give the definitions and to state and prove all the results without exponentials, that is in the fragment $\mllu$: any reader acquainted with \pns knows that there is not the slightest technical difficulty in extending the results thus obtained to $\mell$, as we do in \Cref{sec:mell}. In \Cref{sec:btenll} we refine the sequentialization process of the proof of \Cref{thm:untyped-sequentiality} by putting jumps in a ``canonical'' position (\Cref{thm:untyped-sequentialityBottomTens}) and we prove that the restriction to $\wten$-\pss has an interesting consequence for the fragment $\mllu$: we know from~\cite{heijltjes2014no} that deciding the equivalence by permutations of rules of two \cf \seqc proofs of $\mllu$ is $\mathit{PSPACE}$-complete. On the contrary, we prove that, for the $\wten$ fragment of $\mllu$,\footnote{A very similar fragment had already been considered in Appendix B of~\cite{regnier1992lambda}: a sequentialization result is proven with a technique different from the one used in this paper (\Cref{remark:AppendixRegnier} for more details).} this decision problem is much easier (\Cref{cor:equivalenceFeasible}). Since our interest in $\ACCw$ has no reason to be limited to $\wten$-\pss, we start exploring, in \Cref{sec:imell}, the intuitionistic fragment $\imell$ (where $\wten$ does not hold). We expect that the correctness criterion~for~$\imell$~of~\cite{lamarcheIMELL}~can~be~expressed~in~terms~of~connected~components~(in the~style~of~$\ACCw$).
 
 \begin{outline}
  In \Cref{sec:ps}, we recall the basic definitions of formula and \seqc proof, we define \pss and describe the desequentialization process (\Cref{subsec:desequentialization-main}) that builds a \ps from a \seqc proof. \Cref{def:sequentialps} generalizes the notion of sequentializable \ps to the untyped setting: intuitively, a \seq \ps represents a \seqc proof with no formula. We then express for \pss the properties $\AC$, $\ACC$ and $\ACCw$: since several correctness criteria are considered in the paper, we stick to the term ``\ps'' (satisfying some property) rather than using the expression ``\pn'', which usually refers to \pss satisfying a particular correctness criterion or to \pss that are the desequentialization of some \seqc proof. 
  %The \cutelim procedure for \pss and its properties conclude the section, in particular the stability of $\ACCw$ (\Cref{thm:stability-cut-elimination}). 
  In \Cref{sec:untyped-sequentiality-theorem}, we introduce the condition $\wten$, and we give a characterization of it on \pss in terms of a property of graphs (\Cref{lemma:cw-forall-tensor}): this is the key tool to prove \Cref{thm:untyped-sequentiality}. In \Cref{sec:btenll}, we refine the results of \Cref{sec:untyped-sequentiality-theorem} and we characterize, for the $\wten$ fragment of $\mllu$, the equivalence induced by permutations of rules in \seqc as the equality of \pss satisfying $\ACCw$ (\Cref{cor:equivalenceFeasible}).
  % PENSO CHE SI POSSA TOGLIERE
  %: a \ps satisfying $\ACCw$ contains no jump and the equality is between \pss without jumps, 
  %(jumps are precisely what the property $\ACCw$ allows to avoid) but in this section we use the procedure of rewiring jumps (\cite{heijltjes2014no}) to prove this characterization. 
  In \Cref{sec:mell}, we present the results obtained in the more general framework of $\mell$, thus recovering the polarized fragment (\Cref{cor:sequentialization-polarized}). \Cref{sec:imell} is a first step to study $\ACCw$ beyond $\wten$-\pss: by studying paths in \pss of $\imell$ we prove that, for a \ps of $\imell$ satisfying $\AC$, having a unique output conclusion is equivalent to satisfying $\ACCw$~(\Cref{prop:accw-output-conclusion}). We also show that, in the intuitionistic fragment of $\comll$ (\cite{LincolnW94}), denoted~by~$\icomll$,~one~can~prove,~like~we~do~in~\Cref{sec:btenll},~a~sequentialization~theorem by adding jumps in a ``canonical'' position (\Cref{thm:sequentialization-icomll}).
 \end{outline}

 \begin{notations}
  For ease of reference, we make a list of common notations used in the paper. \medskip
  
  \begin{minipage}{0.4675\textwidth}
   \begin{itemize}
    \item
     $\ve(G)$, $\w(G)$, $\ar(G)$, $\cc(G)$ (\Cref{def:nodes});
    \item
     $\prec$ (\Cref{def:erasing-order});
    \item
     $R^\varphi$ (\Cref{def:switching});
    \item
     $\sharp S$ (\Cref{not:cardinality});
    \item
     $\pi^\des$ (\Cref{subsec:desequentialization-main} and \Cref{fig:desequentialization});
   \end{itemize}
  \end{minipage}
  \hfill
  \begin{minipage}{0.4675\textwidth}
   \begin{itemize}
    \item
     $R \models \AC$, $\C$, $\Cw$, $\ACCw$
          
     (\Cref{def:acyclicity-connectivity,def:acyclicity-connectivity-global});
    \item
     $R \models \cwf$ (\Cref{def:cw-forall});
    \item
     $\wten$-ps (\Cref{def:wten});
    \item
     $J_R$, $\ujf{R}$, $R_n^m$, $\pi \deseqjump R$ (\Cref{sec:btenll,sec:imell}).
   \end{itemize}
  \end{minipage}
  \hfill\null
 \end{notations}
 
 \section{Formulae, \texorpdfstring{\seqc}{sequent calculus} and \texorpdfstring{\pss}{proof-structures}}
  \label{sec:ps}
  \subsection{Definitions}

 %If $R$ is a binary relation on a set $X$, we write $R^+$ for its transitive closure, that is the smallest transitive relation on $X$ containing $R$.
 %For every set $S$, a \emph{partition} $P$ of $S$ is a set of subsets of $S$ such that any two elements~of~$P$ are disjoint and the union of all the elements of $P$ is $S$.
 We consider countably infinite \emph{atomic formulae}, which we denote by $X$.
 %We consider a partition~of~$\mathcal{X}$:
 %\[
 % \mathit{NOT} \vcentcolon = \{\{X,Y\} \colon X, Y \in \mathcal{X}, X \neq Y\}
 %\]
 The~set~of~formulae~of \emph{multiplicative linear logic with units} ($\mllu$) is defined by the following grammar:
 \[
  A \Coloneqq X \mid \one \mid \bot \mid A \otimes A \mid A \parr A
 \]
 $\mll$ is obtained from $\mllu$ by forgetting $\one$ and $\bot$. \emph{Linear negation} on atomic formulae~is~an involution $(\cdot)^\bot$ without fixed points, and is extended to~formulae~of~$\mllu$~inductively~as~follows:
 \[
  \one^\bot \coloneq \bot \qquad \bot^\bot \coloneq \one \qquad (A \otimes B)^\bot \coloneq A^\bot \parr B^\bot \qquad (A \parr B)^\bot \coloneq A^\bot \otimes B^\bot
 \]
% Sequents ($\Gamma,\Delta,\Sigma, \dots$) are finite sequences of formulae of $\mllu$. The rules are in \Cref{fig:sequent-calculus-rules}.
 %A sequent is a finite sequence of formulae of $\mell$. Sequents range over $\Gamma, \Delta, \Sigma, \dots$ The \seqc rules of $\mell$ are:

 \begin{definition}
  A \emph{fragment $\mathcal{F}$ of $\mllu$} is a set of formulae of $\mllu$ which is closed under sub-formulae, meaning that, for every $F \in \mathcal{F}$ and for every sub-formula $G$ of $F$, $G \in \mathcal{F}$. If $\mathcal{F}$ is a fragment of $\mllu$, we say that $\mathcal{F}$ is \emph{closed under linear negation} when, for every $F \in \mathcal{F}$, $F^\perp \in \mathcal{F}$. If $\mathcal{F}$ and $\mathcal{G}$ are fragments of $\mllu$, we say that $\mathcal{F}$ is a \emph{fragment of $\mathcal{G}$} when $\mathcal{F} \subseteq \mathcal{G}$.
 \end{definition}

 \begin{definition}
  A \emph{sequent} $\Gamma$ of a fragment $\mathcal{F}$ of $\mllu$ is a finite sequence of formulae of $\mathcal{F}$. A \emph{\seqc proof} in $\mathcal{F}$ with conclusion $\Gamma$ is a finite tree of sequents~of~$\mathcal{F}$~obtained~by~using the rules of Fig.~\ref{fig:sequent-calculus-rules}. $\Gamma$ is \emph{provable in $\mathcal{F}$} if there is a \seqc~proof~in~$\mathcal{F}$~with~conclusion~$\Gamma$.
  %A \emph{sequent} $\Gamma$ of a fragment $\mathcal{F}$ of $\mllu$ is a finite sequence of formulae of $\mathcal{F}$. A \emph{\seqc proof} in $\mathcal{F}$ with conclusion $\Gamma$ is a finite tree of sequents~of~$\mathcal{F}$~obtained~by~using the rules of \Cref{fig:sequent-calculus-rules}. We say that $\Gamma$ is \emph{provable in $\mathcal{F}$} if there exists a \seqc~proof~in~$\mathcal{F}$ with conclusion $\Gamma$.
  
  \begin{figure}
   \centering
   \AxiomC{\infcstrut}
   \RightLabel{\scriptsize\axsc}
   \UnaryInfC{\infcstrut$\vdash A,A^\perp$}
   \DisplayProof
   \!\quad
   \AxiomC{\infcstrut$\vdash \Gamma, A$}
   \AxiomC{\infcstrut$\vdash A^\perp, \Delta$}
   \RightLabel{\scriptsize\cutsc}
   \BinaryInfC{\infcstrut$\vdash \Gamma, \Delta$}
   \DisplayProof
   \!\quad
   \AxiomC{\infcstrut$\vdash \Gamma, A, B, \Delta$}
   \RightLabel{\scriptsize\exsc}
   \UnaryInfC{\infcstrut$\vdash \Gamma, B, A, \Delta$}
   \DisplayProof \\[1.5em]
   \AxiomC{\infcstrut$\vdash \Gamma, A$}
   \AxiomC{\infcstrut$\vdash B, \Delta$}
   \RightLabel{\scriptsize$\otimes$}
   \BinaryInfC{\infcstrut$\vdash \Gamma, A \otimes B, \Delta$}
   \DisplayProof
   \!\quad
   \AxiomC{\infcstrut}
   \RightLabel{\scriptsize$\one$}
   \UnaryInfC{\infcstrut$\vdash \boldsymbol{1}$}
   \DisplayProof
   \!\quad
   \AxiomC{\infcstrut$\vdash \Gamma, A, B$}
   \RightLabel{\scriptsize$\parr$}
   \UnaryInfC{\infcstrut$\vdash \Gamma, A \parr B$}
   \DisplayProof
   \!\quad
   \AxiomC{\infcstrut$\vdash \Gamma$}
   \RightLabel{\scriptsize$\bot$}
   \UnaryInfC{\infcstrut$\vdash \Gamma, \bot$}
   \DisplayProof
   \caption{\label{fig:sequent-calculus-rules}\Seqc rules.}
  \end{figure}
 \end{definition}

% \begin{remark}
    %  A \seqc proof $\pi$ in $\mell$ with conclusion a sequent $\Gamma$ of $\mathcal{F}$ is in $\mathcal{F}$ if and only if, for every instance of the \cutsc rule in $\pi$, its cut formulae are formulae of $\mathcal{F}$.
    % \end{remark}

% We now define two well-known fragments of $\mell$ which are closed under linear negation.
% 
%% \begin{definition}
    %%  \label{def:mll}
    %%  \emph{Multiplicative linear logic with units} is the fragment $\mllu$ of $\mell$~given~by $\smash{A \Coloneqq X \mid \one \mid \bot \mid A \otimes A \mid A \parr A}$. $\mll$ is obtained from $\mllu$ by forgetting the units $\one$~and~$\bot$.
    %% \end{definition}

% \begin{definition}
    %  \emph{Multiplicative linear logic} is the fragment $\mathit{MLL}$ of $\mathit{MLL}_u$ given by:
    %  \[
    %   A \Coloneqq X \mid A \otimes A \mid A \parr A
    %  \]
    % \end{definition}
% 
% \begin{remark}
    %  The fragments $\mathit{MLL}$ and $\mathit{MLL}_u$ are closed under linear negation.
    % \end{remark}

 In our setting, a directed graph $G$ is an ordered triple consisting of a set $N$ of \emph{nodes}, a set $A$, disjoint from $N$, of \emph{arcs}, and an \emph{incidence} function $I$ which associates with every arc of $G$ an ordered pair of distinct nodes of $G$. The \emph{empty graph} is the graph without nodes. If $a$ is an arc of $G$ and $I(a) = (m, n)$, then $m$ and $n$ are called the \emph{ends} of $a$. More specifically, $m$ is called the \emph{tail} of $a$, and $n$ is the \emph{head} of $a$. We~also~say~that~$a$~is~a~\emph{conclusion}~of~$m$~and~a~\emph{premise}~of~$n$.
 
% The \emph{underlying undirected graph} of $G$, written $U(G)$, is the undirected graph obtained from $G$ by keeping the same set of nodes and by replacing each arc by an edge with~the~same~ends.
 
 A \emph{path} of $G$ is a finite sequence $\gamma = n_0 a_1 n_1 a_2 \dots n_{k-1} a_k n_k$ of alternating nodes and arcs of $G$ such that $a_1 \neq a_k$ (when $k \geq 2$), for every $i \in \{1, \dots, k\}$ the ends of $a_i$ are $n_{i-1}$ and $n_i$, and, for every $i, j \in \{0, \dots, k\}$ with $i \neq j$, if $n_i = n_j$, then $\{i, j\} = \{0, k\}$. We say that $\gamma$ is a path \emph{from $n_0$ to $n_k$} and that the non-negative integer $k$ is the \emph{length} of $\gamma$. For every $i \in \{1, \dots, k\}$, we may write $a_i \in \gamma$ ($\gamma$ \emph{contains} $a_i$). We say that $\gamma$ is: \emph{directed} if, for every $i \in \{1, \dots, k\}$, the tail of $a_i$ is $n_{i-1}$ and the head of $a_i$ is $n_i$; \emph{empty} if $k = 0$; a \emph{cycle} if $k > 0$ and $n_0 = n_k$. We say that $G$ is: \emph{acyclic} (resp.~a \emph{dag}) if $G$ has no cycle (resp.~directed cycle); \emph{connected} if, for every $m, n$ nodes of $G$, there exists a path of $G$ from $m$ to $n$. Every~graph~$G$~can~be~expressed~as~a disjoint~union~of~non-empty~connected~graphs,~called~the~\emph{connected~components} of $G$.\footnote{The empty graph is connected, and it is the only graph with zero connected components. Therefore,~it~is~a \ps $R_\epsilon$ (see~\Cref{def:proof-structure}) such that neither $R_\epsilon \models \C$ or $R_\epsilon \models \Cw$ holds (see~\Cref{def:acyclicity-connectivity-global}).}
 
 In figures depicting graphs, every arc can be represented with its tail above its head. We then speak of \emph{\dpaths} rather than directed paths because, with this convention, every directed path runs downward in the plane.
 
% \begin{figure}
%  \centering
%  \begin{subfigure}{0.090625\textwidth}
%   \centering
%%   \scalebox{\figscale}{
%   	\input{Figures/axiom.tikz}
%%   }
%   \caption{$l \hspace{-0.3ex} = \hspace{-0.3ex} \axpn$}
%  \end{subfigure}
%  \begin{subfigure}{0.1\textwidth}
%   \centering
%%   \scalebox{\figscale}{
%   	\input{Figures/cut.tikz}
%%   }
%   \caption{$l \hspace{-0.3ex} = \hspace{-0.3ex} \cutpn$}
%  \end{subfigure}
%  \begin{subfigure}{0.1625\textwidth}
%   \centering
%%   \scalebox{\figscale}{
%   	\input{Figures/one-gen.tikz}
%%   }
%   \caption{$l \in \{\onepn, \bot, \wkpn\}$}
%  \end{subfigure}
%  \begin{subfigure}{0.16875\textwidth}
%   \centering
%%   \scalebox{\figscale}{
%   	\input{Figures/tensor-gen.tikz}
%%   }
%   \caption{$l \in \{\otimes, \parr, \ctpn\}$}
%  \end{subfigure}
%  \begin{subfigure}{0.084375\textwidth}
%   \centering
%%   \scalebox{\figscale}{
%   	\input{Figures/box.tikz}
%%   }
%   \caption{$l \hspace{-0.3ex} = \hspace{-0.3ex} \bpn$}
%  \end{subfigure}
%  \begin{subfigure}{0.20625\textwidth}
%   \centering
%%   \scalebox{\figscale}{
%   	\input{Figures/dereliction-gen.tikz}
%%   }
%   \caption{$l \in \{\parr, !, \auxpn, \drpn, \ctpn\}$}
%  \end{subfigure}
%  \begin{subfigure}{0.14\textwidth}
%   \centering
%%   \scalebox{\figscale}{
%   	\input{Figures/conclusion-gen.tikz}
%%   }
%   \caption{$l \in \{!, \auxpn, \bullet\}$}
%  \end{subfigure}
%  \caption{\label{fig:links}Local constraints for a node labeled by $l$ and its incident arcs.}
% \end{figure}
 
 \begin{definition}
  \label{def:graph}
  A \emph{\llg} is a dag $G$ whose nodes are labelled by exactly~one~symbol among $\axpn$, $\cutpn$, $\onepn$, $\bot$, $\otimes$, $\parr$, $\bullet$, such that:
  \begin{itemize}
   \item
    Every $\axpn$ node has exactly two conclusions and no premise;
   \item
    Every $\cutpn$ node has exactly two premises and no conclusion;
   \item
    Every $\onepn$ node or $\bot$ node has exactly one conclusion and no premise;
   \item
    Every $\otimes$ node has exactly two premises and one conclusion;
   \item
    Every $\parr$ node has one or two premises (in which case we call it \emph{binary})~and~one~conclusion;
   \item
    Every $\bullet$ node has exactly one premise and no conclusion.
  \end{itemize}
 \end{definition}
 
 From now on, except for \Cref{sec:btenll}, we refer to \llgs simply as graphs.
 
 \begin{definition}
  \label{def:nodes}
  Let $G$ be a graph. We denote by $\ve(G)$, $\w(G)$, $\ar(G)$, $\cc(G)$ the sets of nodes, $\bot$ nodes, arcs, connected components of $G$ respectively. A \emph{conclusion of $G$} is a premise~of~a~$\bullet$ node of $G$. A non-$\bullet$ node is \emph{terminal} if all of its conclusions are conclusions of $G$.
  %Finally, every false, par and dereliction node of $G$ is called a \emph{negative} node.
  %Finally, a \emph{\dpath} of $G$ is a directed path $\gamma$ of $G$ such that no arc of $\gamma$ is a jump.
 \end{definition}
 
 We introduce an order relation on the nodes of a graph, which is used in~\Cref{def:erasing}.
 
 \begin{definition}
  \label{def:erasing-order}
  Let $G$ be a graph. We define a binary relation $\prec$ on $\ve(G)$ in the following~way: we have $n \prec n'$ if and only if there exists a non-empty \dpath of $R$ from $n$ to $n'$.%We write $n \preceq n'$ if $n \prec n'$ or $n = n'$, and $n \succ n'$ (resp.~$n \succeq n'$) if $n' \prec n$ (resp.~$n' \preceq n$).
 \end{definition}
 
 \begin{remark}
  \label{rmk:prec}
  A graph $G$ is a dag (\Cref{def:graph}), hence $\prec$ is a strict order relation. Since $\prec$ is finite, $\prec$ is well-founded and converse well-founded: a non-empty subset~of~$\ve(G)$~has~a~minimal and a maximal element. For every node $n$ of $G$, $\prec$ is a total order on $\{m : n \prec m\}$.
 \end{remark}
 
 We adapt to graphs some standard notions from the theory of linear logic \pss.
 
 \begin{definition}
  \label{def:proof-structure}
  A \emph{\ps}
  \renewcommand{\ps}{ps\xspace}%
  \renewcommand{\pss}{pss\xspace}%
  (\ps for short) is a graph $R$ s.t.~every $\parr$ node is binary,~with:
  \begin{itemize}
   \item
    An ordering of the premises of each $\otimes$ node and of each $\parr$ node;
   \item
    A total ordering of the conclusions of $R$.
  \end{itemize}
  For every $\otimes$ or $\parr$ node $n$ of $R$, the first premise is the \emph{left premise} of $n$, and the second one is the \emph{right premise}. We say that $R$ is: \emph{\cf} if $R$ has no $\cutpn$ node; a \ps of $\mllu$ if it comes with a labelling $\ell$ of its arcs by formulae of $\mllu$ (\emph{types}), s.t., for every $a, a'$ conclusions (resp.~premises) of an $\axpn$ (resp.~$\cutpn$) node, $\ell(a) = \ell(a')^\bot$, and for every arc $a$ conclusion of:
  \begin{itemize}
   \item
    A $\onepn$ (resp.~$\bot$) node, $\ell(a) = \one$ (resp.~$\ell(a) = \bot$);
   \item
    A $\otimes$ node with left premise $a_1$ and right premise $a_2$, $\ell(a) = \ell(a_1) \otimes \ell(a_2)$;
   \item
    A $\parr$ node with left premise $a_1$ and right premise $a_2$, $\ell(a) = \ell(a_1) \parr \ell(a_2)$.
  \end{itemize}
  A \emph{\ps of $\mathcal{F}$}, with $\mathcal{F}$ fragment of $\mllu$, is a \ps of $\mllu$ with all its arcs typed by formulae~of~$\mathcal{F}$.
 \end{definition}
 
 \renewcommand{\ps}{ps\xspace}%
 \renewcommand{\pss}{pss\xspace}%
 
 \begin{definition}
  \label{def:switching}
  A \emph{switching} $\varphi$ of a \ps $R$ is a function which maps every $\parr$ node of $R$ to one of its premises. The \emph{\sg} $R^\varphi$ of $R$ induced by $\varphi$ is obtained by adding a fresh $\bullet$ node $n_0$ for each $\parr$ node $n$, and by setting the head of the premise $a$ of $n$ such~that~$a \neq \varphi(n)$~to~be~$n_0$.\footnote{Switching graphs are \llgs, because $\parr$ nodes can have one or two premises (\Cref{def:graph}).}
 \end{definition}
 
 \begin{notation}
  \label{not:cardinality}
  Let $S$ be a set. We denote by $\sharp S$ the cardinality of $S$.
 \end{notation}
 
 \begin{remark}
  \label{rmk:same-vertices-arcs}
%  Let $R$ be a \ps, and let $\varphi_1$ and $\varphi_2$ be switchings of $R$. Then:
%  \[
%   \sharp \ve(R^{\varphi_1}) = \sharp \ve(R^{\varphi_2}) \qquad \sharp \ar(R^{\varphi_1}) = \sharp \ar(R^{\varphi_2})
%  \]
  For $R$ \ps, $\varphi_1$ and $\varphi_2$ switchings of $R$: $\sharp \ve(R^{\varphi_1}) = \sharp \ve(R^{\varphi_2})$,~$\sharp \ar(R^{\varphi_1}) = \sharp \ar(R^{\varphi_2})$.
 \end{remark}
 
 A \spath cannot ``bounce'' on a pair of premises $a_1$ and $a_2$ of the same $\parr$ node. It is also forbidden for a \spath to contain $a_1$ and $a_2$ as its first and last arcs.
 
 \begin{definition}
 A path is \emph{switching} if it does not contain two premises of~the~same~$\parr$ node.
 %A path $\gamma$ is \emph{switching} if $\gamma$ does not contain two premises of the same $\parr$ or $\ctpn$ node.
  %Let $G$ be a graph. A path $\gamma$ of $G$ is a \emph{\spath} if, for every $\parr$ or $\ctpn$ node $n$ of $\gamma$, we have that $\gamma$ does not contain two premises of $n$.
 \end{definition}
 
 \begin{remark}
  A path of a \ps $R$ is switching if and only if it is a path of $R^\varphi$ for some~switching~$\varphi$.
 \end{remark}
 
% \begin{proof}
%  We prove that $<$ is irreflexive and transitive.
%  \begin{itemize}
%   \item
% 	\emph{Irreflexivity.} Let $n$ be a node of $G$. We know that $G$ is a directed acyclic graph by Definition~\ref{def:graph}. Hence, we have $n \not < n$;
%   \item
% 	\emph{Transitivity.} Let $n$, $n'$ and $n''$ be nodes of $G$ and assume $n < n'$ and $n' < n''$. Then there exist a non-trivial descent path $\gamma$ from $n$ to $n'$ and a non-trivial descent path $\gamma'$ from $n'$ to $n''$ in $G$. If there existed a node of $\gamma'$ which is also crossed by $\gamma$ and is different from $n'$, then we could consider the first such node $m$. Let $\sigma$ be the suffix of $\gamma$ starting from $m$ and let $\rho'$ be the prefix of $\gamma'$ ending in $m$. By minimality of $m$, $\sigma\rho'$ would be a non-trivial directed cycle of $G$, which is a contradiction, because $G$ is a directed acyclic graph by Definition~{def:graph}. Then the only common node of $\gamma$ and $\gamma'$ is $n'$. Therefore, $\gamma\gamma'$ is a non-trivial descent path from $n$ to $n''$. Hence, we have $n < n''$. \qedhere
%  \end{itemize}
% \end{proof}

\subsection{Desequentialization}
 
 \label{subsec:desequentialization-main}
 The well-known process of desequentialization maps a \seqc proof $\pi$ in $\mllu$ with conclusion $A_1, \dots, A_k$ to a \ps $\pi^\deseq$ of $\mllu$ with conclusions of types $A_1, \dots, A_k$. The function $\pi \mapsto \pi^\deseq$ is defined by induction on $\pi$ (see \Cref{fig:desequentialization}). Full details~are~in~\Cref{subsec:cut-elimination-steps}~and~in~\cite{llhandbook}.

 \begin{figure}
  \centering \hfill
  \begin{subfigure}{0.1125\textwidth}
   \centering
 		%  	\scalebox{\deseqscale}{
   \begin{tikzpicture}
	\begin{pgfonlayer}{nodelayer}
		\node [style=agent] (1) at (-0.5, -0.5) {\scriptsize$A\vphantom{A^\perp}$};
		\node [style=agent] (2) at (0.5, -0.5) {\scriptsize$A^\perp$};
		\node [style=axwLink] (6) at (0, 0.25) {};
		\node [style=none] (7) at (0.3, 0.325) {\tiny$n$};
		\node [style=DOTagent] (8) at (-0.5, -1.075) {};
		\node [style=DOTagent] (9) at (0.5, -1.075) {};
	\end{pgfonlayer}
	\begin{pgfonlayer}{edgelayer}
		\draw [style=DoutLEFT] (6) to (1);
		\draw [style=DoutRIGHT] (6) to (2);
		\draw [style=Dsimple] (1) to (8);
		\draw [style=Dsimple] (2) to (9);
	\end{pgfonlayer}
\end{tikzpicture}
 			%  	}
   \caption{\label{subfig:desequentialization-axiom}\axsc rule.}
  \end{subfigure}
  \hfill
  \begin{subfigure}{0.175\textwidth}
   \centering
 		%  	\scalebox{\deseqscale}{
   \begin{tikzpicture}
	\begin{pgfonlayer}{nodelayer}
		\node [style=none] (0) at (-1.1, 1.75) {};
		\node [style=none] (1) at (-0.25, 1.75) {};
		\node [style=none] (2) at (-1.1, 1.25) {};
		\node [style=none] (3) at (-0.25, 1.25) {};
		\node [style=none] (4) at (-1.1, 1.5) {};
		\node [style=none] (5) at (-0.25, 1.5) {};
		\node [style=none] (6) at (-0.25, 1) {};
		\node [style=none] (7) at (-1.1, 1) {};
		\node [style=none] (10) at (-0.5, 1) {};
		\node [style=whiteTest] (15) at (-0.675, 1.375) {\scriptsize$R_1$};
		\node [style=none] (16) at (0.25, 1.75) {};
		\node [style=none] (17) at (1.1, 1.75) {};
		\node [style=none] (18) at (0.25, 1.25) {};
		\node [style=none] (19) at (1.1, 1.25) {};
		\node [style=none] (20) at (0.25, 1.5) {};
		\node [style=none] (21) at (1.1, 1.5) {};
		\node [style=none] (22) at (1.1, 1) {};
		\node [style=none] (23) at (0.25, 1) {};
		\node [style=none] (26) at (0.5, 1) {};
		\node [style=whiteTest] (31) at (0.675, 1.375) {\scriptsize$R_2$};
		\node [style=cutwLink] (32) at (0, 0) {};
		\node [style=none] (33) at (0.3, -0.325) {\tiny$n$};
		\node [style=agent] (34) at (-0.5, 0.5) {\scriptsize$\smash{A}\vphantom{X}$};
		\node [style=agent] (35) at (0.5, 0.5) {\scriptsize$\smash{A^\perp}\vphantom{X}$};
		\node [style=none] (40) at (-0.875, 0.35) {};
		\node [style=none] (41) at (-0.875, 0.125) {};
		\node [style=none] (42) at (-0.825, 0.125) {};
		\node [style=none] (43) at (-0.825, 0.35) {};
		\node [style=DOTagent] (44) at (-0.85, -0.075) {};
		\node [style=none] (45) at (-0.875, 1) {};
		\node [style=none] (46) at (-0.875, 0.775) {};
		\node [style=none] (47) at (-0.825, 0.775) {};
		\node [style=none] (48) at (-0.825, 1) {};
		\node [style=agent] (49) at (-0.85, 0.5) {\scriptsize$\smash{\Gamma}\vphantom{X}$};
		\node [style=none] (50) at (0.825, 0.35) {};
		\node [style=none] (51) at (0.825, 0.125) {};
		\node [style=none] (52) at (0.875, 0.125) {};
		\node [style=none] (53) at (0.875, 0.35) {};
		\node [style=DOTagent] (54) at (0.85, -0.075) {};
		\node [style=none] (55) at (0.825, 1) {};
		\node [style=none] (56) at (0.825, 0.775) {};
		\node [style=none] (57) at (0.875, 0.775) {};
		\node [style=none] (58) at (0.875, 1) {};
		\node [style=agent] (59) at (0.85, 0.5) {\scriptsize$\smash{\Delta}\vphantom{X}$};
	\end{pgfonlayer}
	\begin{pgfonlayer}{edgelayer}
		\draw [style=roundedCornerBlackFill] (2.center)
			 to (0.center)
			 to (1.center)
			 to (3.center);
		\draw [style=roundedCornerBlackFill] (4.center)
			 to (7.center)
			 to (6.center)
			 to (5.center);
		\draw [style=roundedCornerBlackFill] (18.center)
			 to (16.center)
			 to (17.center)
			 to (19.center);
		\draw [style=roundedCornerBlackFill] (20.center)
			 to (23.center)
			 to (22.center)
			 to (21.center);
		\draw [style=DinLEFT] (34) to (32);
		\draw [style=DinRIGHT] (35) to (32);
		\draw [style=simple] (10.center) to (34);
		\draw [style=simple] (26.center) to (35);
		\draw [style=simpleB] (45.center) to (46.center);
		\draw [style=simpleB] (48.center) to (47.center);
		\draw [style=Dsimple] (43.center) to (42.center);
		\draw [style=Dsimple] (40.center) to (41.center);
		\draw [style=simpleB] (55.center) to (56.center);
		\draw [style=simpleB] (58.center) to (57.center);
		\draw [style=Dsimple] (53.center) to (52.center);
		\draw [style=Dsimple] (50.center) to (51.center);
	\end{pgfonlayer}
\end{tikzpicture}
 			%  	}
   \caption{\label{subfig:desequentialization-cut}\cutsc rule.}
  \end{subfigure}
  \hfill
  \begin{subfigure}{0.1\textwidth}
   \centering
 		%  	\scalebox{\deseqscale}{
   \begin{tikzpicture}
	\begin{pgfonlayer}{nodelayer}
		\node [style=agent] (1) at (0, -0.75) {\scriptsize$\boldsymbol{1}$};
		\node [style=onewLink] (7) at (0, 0) {};
		\node [style=none] (8) at (0.25, 0.075) {\tiny $n$};
		\node [style=DOTagent] (9) at (0, -1.3125) {};
	\end{pgfonlayer}
	\begin{pgfonlayer}{edgelayer}
		\draw [style=simple] (7) to (1);
		\draw [style=Dsimple] (1) to (9);
	\end{pgfonlayer}
\end{tikzpicture}
 			%  	}
   \caption{\label{subfig:desequentialization-one}$\one$ rule.}
  \end{subfigure}
  \hfill
  \begin{subfigure}{0.125\textwidth}
   \centering
 		%  	\scalebox{\deseqscale}{
   \begin{tikzpicture}
	\begin{pgfonlayer}{nodelayer}
		\node [style=none] (0) at (-1, 1.25) {};
		\node [style=none] (1) at (-0.5, 1.25) {};
		\node [style=none] (2) at (-1, 0.75) {};
		\node [style=none] (3) at (-0.5, 0.75) {};
		\node [style=none] (4) at (-1, 1) {};
		\node [style=none] (5) at (-0.5, 1) {};
		\node [style=none] (6) at (-0.5, 0.5) {};
		\node [style=none] (7) at (-1, 0.5) {};
		\node [style=whiteTest] (15) at (-0.75, 0.875) {\scriptsize$R_1$};
		\node [style=botwLink] (16) at (-0.125, 0.75) {};
		\node [style=agent] (17) at (-0.125, 0) {\scriptsize$\bot\vphantom{X}$};
		\node [style=none] (18) at (0.175, 0.825) {\tiny $n$};
		\node [style=DOTagent] (21) at (-0.125, -0.575) {};
		\node [style=none] (22) at (-0.775, -0.15) {};
		\node [style=none] (23) at (-0.775, -0.375) {};
		\node [style=none] (24) at (-0.725, -0.375) {};
		\node [style=none] (25) at (-0.725, -0.15) {};
		\node [style=DOTagent] (26) at (-0.75, -0.575) {};
		\node [style=none] (27) at (-0.775, 0.5) {};
		\node [style=none] (28) at (-0.775, 0.275) {};
		\node [style=none] (29) at (-0.725, 0.275) {};
		\node [style=none] (30) at (-0.725, 0.5) {};
		\node [style=agent] (31) at (-0.75, 0) {\scriptsize$\smash{\Gamma}\vphantom{X}$};
	\end{pgfonlayer}
	\begin{pgfonlayer}{edgelayer}
		\draw [style=roundedCornerBlackFill] (2.center)
			 to (0.center)
			 to (1.center)
			 to (3.center);
		\draw [style=roundedCornerBlackFill] (4.center)
			 to (7.center)
			 to (6.center)
			 to (5.center);
		\draw [style=simple] (16) to (17);
		\draw [style=Dsimple] (17) to (21);
		\draw [style=simpleB] (27.center) to (28.center);
		\draw [style=simpleB] (30.center) to (29.center);
		\draw [style=Dsimple] (25.center) to (24.center);
		\draw [style=Dsimple] (22.center) to (23.center);
	\end{pgfonlayer}
\end{tikzpicture}
 			%  	}
   \caption{\label{subfig:desequentialization-bottom}$\bot$ rule.}
  \end{subfigure}
  \hfill
  \begin{subfigure}{0.175\textwidth}
   \centering
 		%  	\scalebox{\deseqscale}{
   \begin{tikzpicture}
	\begin{pgfonlayer}{nodelayer}
		\node [style=none] (0) at (-1.1, 1.75) {};
		\node [style=none] (1) at (-0.25, 1.75) {};
		\node [style=none] (2) at (-1.1, 1.25) {};
		\node [style=none] (3) at (-0.25, 1.25) {};
		\node [style=none] (4) at (-1.1, 1.5) {};
		\node [style=none] (5) at (-0.25, 1.5) {};
		\node [style=none] (6) at (-0.25, 1) {};
		\node [style=none] (7) at (-1.1, 1) {};
		\node [style=none] (10) at (-0.5, 1) {};
		\node [style=whiteTest] (15) at (-0.675, 1.375) {\scriptsize$R_1$};
		\node [style=none] (16) at (0.25, 1.75) {};
		\node [style=none] (17) at (1.1, 1.75) {};
		\node [style=none] (18) at (0.25, 1.25) {};
		\node [style=none] (19) at (1.1, 1.25) {};
		\node [style=none] (20) at (0.25, 1.5) {};
		\node [style=none] (21) at (1.1, 1.5) {};
		\node [style=none] (22) at (1.1, 1) {};
		\node [style=none] (23) at (0.25, 1) {};
		\node [style=none] (26) at (0.5, 1) {};
		\node [style=whiteTest] (31) at (0.675, 1.375) {\scriptsize$R_2$};
		\node [style=tensorwLink] (32) at (0, 0) {};
		\node [style=none] (33) at (0.225, -0.35) {\tiny$n$};
		\node [style=agent] (34) at (-0.5, 0.5) {\scriptsize$\smash{A}\vphantom{X}$};
		\node [style=agent] (35) at (0.5, 0.5) {\scriptsize$\smash{B}\vphantom{X}$};
		\node [style=agent] (40) at (0, -0.75) {\scriptsize$\smash{A{\otimes}B}\vphantom{X}$};
		\node [style=DOTagent] (41) at (0, -1.3) {};
		\node [style=none] (42) at (-0.875, 0.35) {};
		\node [style=none] (43) at (-0.875, 0.125) {};
		\node [style=none] (44) at (-0.825, 0.125) {};
		\node [style=none] (45) at (-0.825, 0.35) {};
		\node [style=DOTagent] (46) at (-0.85, -0.075) {};
		\node [style=none] (47) at (-0.875, 1) {};
		\node [style=none] (48) at (-0.875, 0.775) {};
		\node [style=none] (49) at (-0.825, 0.775) {};
		\node [style=none] (50) at (-0.825, 1) {};
		\node [style=agent] (51) at (-0.85, 0.5) {\scriptsize$\smash{\Gamma}\vphantom{X}$};
		\node [style=none] (52) at (0.825, 0.35) {};
		\node [style=none] (53) at (0.825, 0.125) {};
		\node [style=none] (54) at (0.875, 0.125) {};
		\node [style=none] (55) at (0.875, 0.35) {};
		\node [style=DOTagent] (56) at (0.85, -0.075) {};
		\node [style=none] (57) at (0.825, 1) {};
		\node [style=none] (58) at (0.825, 0.775) {};
		\node [style=none] (59) at (0.875, 0.775) {};
		\node [style=none] (60) at (0.875, 1) {};
		\node [style=agent] (61) at (0.85, 0.5) {\scriptsize$\smash{\Delta}\vphantom{X}$};
	\end{pgfonlayer}
	\begin{pgfonlayer}{edgelayer}
		\draw [style=roundedCornerBlackFill] (2.center)
			 to (0.center)
			 to (1.center)
			 to (3.center);
		\draw [style=roundedCornerBlackFill] (4.center)
			 to (7.center)
			 to (6.center)
			 to (5.center);
		\draw [style=roundedCornerBlackFill] (18.center)
			 to (16.center)
			 to (17.center)
			 to (19.center);
		\draw [style=roundedCornerBlackFill] (20.center)
			 to (23.center)
			 to (22.center)
			 to (21.center);
		\draw [style=DinLEFT] (34) to (32);
		\draw [style=DinRIGHT] (35) to (32);
		\draw [style=simple] (10.center) to (34);
		\draw [style=simple] (26.center) to (35);
		\draw [style=simple] (32) to (40);
		\draw [style=Dsimple] (40) to (41);
		\draw [style=simpleB] (47.center) to (48.center);
		\draw [style=simpleB] (50.center) to (49.center);
		\draw [style=Dsimple] (45.center) to (44.center);
		\draw [style=Dsimple] (42.center) to (43.center);
		\draw [style=simpleB] (57.center) to (58.center);
		\draw [style=simpleB] (60.center) to (59.center);
		\draw [style=Dsimple] (55.center) to (54.center);
		\draw [style=Dsimple] (52.center) to (53.center);
	\end{pgfonlayer}
\end{tikzpicture}
 			%  	}
   \caption{\label{subfig:desequentialization-tensor}$\otimes$ rule.}
  \end{subfigure}
  \hfill
  \begin{subfigure}{0.15\textwidth}
   \centering
 		%  	\scalebox{\deseqscale}{
   \begin{tikzpicture}
	\begin{pgfonlayer}{nodelayer}
		\node [style=none] (0) at (-1.45, 1.75) {};
		\node [style=none] (1) at (0.4, 1.75) {};
		\node [style=none] (2) at (-1.45, 1.25) {};
		\node [style=none] (3) at (0.4, 1.25) {};
		\node [style=none] (4) at (-1.45, 1.5) {};
		\node [style=none] (5) at (0.4, 1.5) {};
		\node [style=none] (6) at (0.4, 1) {};
		\node [style=none] (7) at (-1.45, 1) {};
		\node [style=none] (10) at (-0.85, 1) {};
		\node [style=whiteTest] (15) at (-0.525, 1.375) {\scriptsize$R_1$};
		\node [style=none] (16) at (0.15, 1) {};
		\node [style=parrwLink] (17) at (-0.35, 0) {};
		\node [style=agent] (18) at (-0.35, -0.75) {\scriptsize$\smash{A{\parr}B}\vphantom{X}$};
		\node [style=none] (19) at (-0.125, -0.35) {\tiny $n$};
		\node [style=agent] (20) at (-0.85, 0.5) {\scriptsize$\smash{A}\vphantom{X}$};
		\node [style=agent] (21) at (0.15, 0.5) {\scriptsize$\smash{B}\vphantom{X}$};
		\node [style=DOTagent] (22) at (-0.35, -1.2875) {};
		\node [style=none] (23) at (-1.225, 0.35) {};
		\node [style=none] (24) at (-1.225, 0.125) {};
		\node [style=none] (25) at (-1.175, 0.125) {};
		\node [style=none] (26) at (-1.175, 0.35) {};
		\node [style=DOTagent] (27) at (-1.2, -0.075) {};
		\node [style=none] (28) at (-1.225, 1) {};
		\node [style=none] (29) at (-1.225, 0.775) {};
		\node [style=none] (30) at (-1.175, 0.775) {};
		\node [style=none] (31) at (-1.175, 1) {};
		\node [style=agent] (32) at (-1.2, 0.5) {\scriptsize$\smash{\Gamma}\vphantom{X}$};
	\end{pgfonlayer}
	\begin{pgfonlayer}{edgelayer}
		\draw [style=roundedCornerBlackFill] (2.center)
			 to (0.center)
			 to (1.center)
			 to (3.center);
		\draw [style=roundedCornerBlackFill] (4.center)
			 to (7.center)
			 to (6.center)
			 to (5.center);
		\draw [style=simple] (17) to (18);
		\draw [style=Dsimple] (18) to (22);
		\draw [style=simple] (10.center) to (20);
		\draw [style=simple] (16.center) to (21);
		\draw [style=DinLEFT] (20) to (17);
		\draw [style=DinRIGHT] (21) to (17);
		\draw [style=simpleB] (28.center) to (29.center);
		\draw [style=simpleB] (31.center) to (30.center);
		\draw [style=Dsimple] (26.center) to (25.center);
		\draw [style=Dsimple] (23.center) to (24.center);
	\end{pgfonlayer}
\end{tikzpicture}
 			%  	}
   \caption{\label{subfig:desequentialization-par}$\parr$ rule.}
  \end{subfigure}
  \hfill \null
  \caption{\label{fig:desequentialization}Desequentialization $(\cdot)^\deseq$ of \seqc proofs in $\mllu$ into \pss of $\mllu$. Double arrows to a $\bullet$ node are a shorthand for an arbitrary number of arcs, each targeting a distinct $\bullet$ node.}
 \end{figure}

 \begin{definition}
  Let $\mathcal{F}$ be a fragment of $\mllu$. A \ps $R$ of $\mathcal{F}$ is \emph{sequentializable} in $\mathcal{F}$ if~there exists a \seqc proof $\pi$ in $\mathcal{F}$ such that $R = \pi^\circ$.\footnote{In the rest of the paper, we will not give the full details of the correspondence between the conclusion sequent of $\pi$ and the conclusions of $R$.}
 \end{definition}
 
% The following definition generalizes the notion of sequentializable \ps to the untyped~setting. Intuitively, a \seq \ps represents a \seqc proof with no formula.
 
 \begin{definition}[By induction on $\sharp \ar(R)$]
  \label{def:sequentialps}
  Let $R$ be a \ps and let $n$ be a node of $R$. We say that $R$ is \emph{$n$-\seq} if it has the shape in the sub-figure of \Cref{fig:desequentialization} determined by the label of $n$ (in the picture, all formulae and sequents are to be ignored),~for~some~\seq~\pss~$R_1, R_2$,~where a \ps is \emph{\seq} if it is $n$-\seq for some $n$.
 \end{definition}

 \begin{remark}
  \label{rmk:sequential}
  Let $\mathcal{F}$ be a fragment of $\mllu$.
  \begin{itemize}
   \item
	If $\pi$ is a \seqc proof in $\mathcal{F}$, then $\pi^\deseq$ is a \seq \ps of $\mathcal{F}$;
   \item
	If $R$ is a \seq \ps of $\mathcal{F}$, then $R$ is sequentializable in $\mathcal{F}$.
  \end{itemize}
 \end{remark}

\subsection{Correctness criteria and \texorpdfstring{\cutelim}{cut elimination}}
 
 \begin{definition}
  \label{def:acyclicity-connectivity}
  Let $R$ be a \ps and let $\varphi$ be a switching of $R$. We write:
  
%  \noindent $R^\varphi \models \AC$ if $R^\varphi$ is acyclic; $R^\varphi \models \C$ if $\sharp \cc(R^\varphi) = 1$; $R^\varphi \models \Cw$ if $\smash{\sharp \cc(R^\varphi) = \sharp \w(R) + 1}$.

  \medskip

  \noindent $R^\varphi \models \AC$ if $R^\varphi$ is acyclic; \hfill $R^\varphi \models \C$ if $\sharp \cc(R^\varphi) = 1$; \hfill $R^\varphi \models \Cw$ if $\smash{\sharp \cc(R^\varphi) = \sharp \w(R) + 1}$.
 %  \begin{itemize}
%   \item
% 	$G^\varphi \models \AC$ if $G^\varphi$ is acyclic;
%   \item
% 	$G^\varphi \models \C$ if $\sharp \cc(G^\varphi) = 1$;
%   \item
% 	$G^\varphi \models \Cw$ if $\sharp \cc(G^\varphi) = \sharp \w(G) + 1$.
%  \end{itemize}
 \end{definition}
 
 \begin{definition}
  \label{def:acyclicity-connectivity-global}
  Let $R$ be a \ps. We write:
  \begin{itemize}
   \item
    $R \models \AC$ (resp.~$\C$, $\Cw$) if, for every switching $\varphi$ of $R$, we have $R^\varphi \models \AC$ (resp.~$\C$, $\Cw$);
   \item
    $R \models \ACC$ if $R \models \AC$ and $R \models \C$; $R \models \ACCw$ if $R \models \AC$ and $R \models \Cw$.
  \end{itemize}
 \end{definition}
 
 The following proposition %(for $G$ graph) 
 expresses some results on the number of connected components of a \sg, which all descend from a result of graph theory: the number of connected components of an acyclic graph equals the number of its vertices minus~the~number~of~its~arcs.
 
 \begin{proposition}
  \label{prop:acyclicity-cc}
  Let $R$ be a \ps.
  \begin{enumerate}[(i)]
   \item \label{itm:acyclicity-cc-number}
    For every switching $\varphi$ of $R$, if $R^\varphi \models \AC$, then $\sharp \cc(R^\varphi) = \sharp \ve(R^\varphi) - \sharp \ar(R^\varphi)$;
   \item \label{itm:acyclicity-cc-same}
    If $R \models \AC$, then, for every two switchings $\varphi_1$ and $\varphi_2$ of $R$, we have $\sharp \cc(R^{\varphi_1}) = \sharp \cc(R^{\varphi_2})$;
   \item \label{itm:acyclicity-cc-existential}
    If $R \models \AC$, then $R \models \C$ (resp.~$R \models \Cw$) if and only if there exists a switching $\varphi$ of $R$ such that $R^\varphi \models \C$ (resp.~$R^\varphi \models \Cw$).
  \end{enumerate}
 \end{proposition}
 
% The following remark is proven straightforwardly by induction.
 
 \begin{remark}
  \label{rmk:sequent-connected-components}
  Let $\pi$ be a \seqc proof in $\mllu$. Then $\pi^\deseq \models \ACCw$.
 \end{remark}
 
% As an immediate consequence, we obtain a necessary condition for a \ps of $\mllu$ to~be~the desequentialization of a \seqc proof.
 
 \begin{proposition}
  \label{prop:sequentializable}
  Let $R$ be a \ps of $\mllu$. If $R$ is sequentializable in $\mllu$, then $R \models \ACCw$.
 \end{proposition}
 
 \begin{remark}
  \label{rmk:sequentializable}
  The converse of~\Cref{prop:sequentializable} does not hold, because of the example exhibited in \Cref{subfig:counterexample-mllu}, due to Laurent Regnier (common knowledge in the community, implicit in~\cite{regnier1992lambda}).
 	
  \begin{figure}
   \centering
   \hfill
   \begin{subfigure}{0.225\textwidth}
   	\centering
   	%  	\scalebox{\deseqscale}{
   	\begin{tikzpicture}
	\begin{pgfonlayer}{nodelayer}
		\node [style=parrwLink] (1) at (-0.5, 0.75) {};
		\node [style=tensorwLink] (2) at (0, -0.5) {};
		\node [style=DOTagent] (4) at (0, -1.75) {};
		\node [style=agent] (11) at (-1, 1.25) {\scriptsize$\boldsymbol{1}$};
		\node [style=agent] (12) at (0, 1.25) {\scriptsize$\boldsymbol{1}$};
		\node [style=agent] (13) at (-0.5, 0) {\scriptsize $\boldsymbol{1} \parr \boldsymbol{1}$};
		\node [style=agent] (14) at (0.5, 0) {\scriptsize $\bot$};
		\node [style=agent] (15) at (0, -1.25) {\scriptsize $(\boldsymbol{1} \parr \boldsymbol{1}) \otimes \bot$};
		\node [style=onewLink] (16) at (-1, 2) {};
		\node [style=onewLink] (17) at (0, 2) {};
		\node [style=botwLink] (18) at (0.5, 0.75) {};
	\end{pgfonlayer}
	\begin{pgfonlayer}{edgelayer}
		\draw [style=DinLEFT] (11) to (1);
		\draw [style=DinRIGHT] (12) to (1);
		\draw [style=simple] (1) to (13);
		\draw [style=DinLEFT] (13) to (2);
		\draw [style=DinRIGHT] (14) to (2);
		\draw [style=simple] (2) to (15);
		\draw [style=Dsimple] (15) to (4);
		\draw [style=simple] (16) to (11);
		\draw [style=simple] (17) to (12);
		\draw [style=simple] (18) to (14);
	\end{pgfonlayer}
\end{tikzpicture}
   		%  	}
   	\caption{\label{subfig:counterexample-mllu}Classical (Regnier).}
   \end{subfigure}
   \hfill
   \begin{subfigure}{0.275\textwidth}
   	\centering
   	%  	\scalebox{\deseqscale}{
   	\begin{tikzpicture}
	\begin{pgfonlayer}{nodelayer}
		\node [style=parrwLink] (1) at (-0.5, 0.75) {};
		\node [style=tensorwLink] (2) at (0, -0.5) {};
		\node [style=DOTagent] (4) at (0, -1.75) {};
		\node [style=agent] (12) at (0, 1.25) {\scriptsize$\boldsymbol{1}$};
		\node [style=agent] (13) at (-0.5, 0) {\scriptsize$\smash{X^\perp \parr \boldsymbol{1}}\vphantom{X}$};
		\node [style=agent] (14) at (0.5, 0) {\scriptsize$\smash{\bot}\vphantom{X}$};
		\node [style=agent] (15) at (0, -1.25) {\scriptsize$\smash{(X^\perp \parr \boldsymbol{1}) \otimes \bot}\vphantom{(\boldsymbol{1} \parr \boldsymbol{1}) \otimes \bot}$};
		\node [style=onewLink] (17) at (0, 2) {};
		\node [style=botwLink] (18) at (0.5, 0.75) {};
		\node [style=agent] (19) at (-1, 1.25) {\scriptsize$\smash{X^\perp}\vphantom{X}$};
		\node [style=axwLink] (20) at (-1.5, 2) {};
		\node [style=agent] (21) at (-2, 1.25) {\scriptsize$X$};
		\node [style=DOTagent] (22) at (-2, 0.75) {};
	\end{pgfonlayer}
	\begin{pgfonlayer}{edgelayer}
		\draw [style=DinRIGHT] (12) to (1);
		\draw [style=simple] (1) to (13);
		\draw [style=DinLEFT] (13) to (2);
		\draw [style=DinRIGHT] (14) to (2);
		\draw [style=simple] (2) to (15);
		\draw [style=Dsimple] (15) to (4);
		\draw [style=simple] (17) to (12);
		\draw [style=simple] (18) to (14);
		\draw [style=outRIGHT] (20) to (19);
		\draw [style=outLEFT] (20) to (21);
		\draw [style=Dsimple] (21) to (22);
		\draw [style=DinLEFT] (19) to (1);
	\end{pgfonlayer}
\end{tikzpicture}
   		%  	}
   	\caption{\label{subfig:counterexample-imllu}Intuitionistic (\Cref{sec:imell}).}
   \end{subfigure}
   \hfill \null
   \caption{\label{fig:counterexamples}Two \pss satisfying $\ACCw$ which are not sequentializable.}
  \end{figure}
 \end{remark}
 
 \begin{remark}
  In $\mll$ (no units), $\ACCw$ collapses to $\ACC$, thus it is sufficient for sequentialization (see~\cite{girard1987linear,danos1989structure,llhandbook}): if $R$ is a \ps of $\mll$~and~$R \models \ACCw$,~then~$R$~is~sequentializable~in~$\mll$.
 \end{remark}
 
 We recall the dynamics of \pss in \Cref{subsec:cut-elimination-steps}, by illustrating the \cutelim steps. A binary relation $\tocut$ on the set of \pss is then defined: we have that $R \tocut R'$ if $R$~reduces~to~$R'$~in exactly one \cutelim step. The following result is well-known (see~\cite{girard1987linear,danos1990logique,llhandbook}).
 
 % We recall the following properties of binary relations.
 % 
 % \begin{notation}
 	%  Let $\rho$ be a binary relation on a set. The reflexive transitive closure of $\rho$ is denoted by $\rho^*$.
 	% \end{notation}
 % 
 % \begin{definition}
 	%  Let $\rho$ be a binary relation on a set $X$. We say that $\rho$ is:
 	%  \begin{itemize}
 		%   \item
 		%    \emph{Confluent} if $\forall \, x, x_1, x_2 \in X : x \mathrel{\rho} x_1$ and $x \mathrel{\rho} x_2$ imply that $\exists \, x_0 \in X : x_1 \mathrel{\rho}^* x_0$ and $x_2 \mathrel{\rho}^* x_0$;
 		%   \item
 		%    \emph{Strongly normalizing} if $\forall \, x_0 \in X$ there exists no infinite sequence $x_1, x_2, \dots$ of elements of $X$ such that $x_{i-1} \mathrel{\rho} x_i$ for every positive integer $i$.
 		%  \end{itemize}
 	% \end{definition}
 % 
 % \begin{lemma}
 	%  Let $\rho$ be a binary relation on a set $X$. If $\rho$ is confluent and strongly normalizing, then every $x \in X$ has a unique normal form $x_0$.\footnote{Recall that $x_0$ is a normal form of $x$ if $x \mathrel{\rho}^* x_0$ and there is no $y \in X$ such that $x_0 \mathrel{\rho} y$.}
 	% \end{lemma}
 
 \begin{theorem}
  \label{thm:cut-elim}
  The binary relation $\tocut$ is confluent and strongly normalizing on $\{R : \text{$R$ \ps}\}$.\footnote{Consequently, every \ps $R$ has a unique normal form.}
 \end{theorem}
 
 % We can then define the notion of normal form of an untyped proof-structure.
 % 
 % \begin{definition}
 	%  \label{def:normal-form}
 	%  Let $R$ be an untyped proof-structure. The unique cut-free untyped proof-structure $R'$ such that $R \to^* R'$ is called the \emph{normal form} of $R$.
 	% \end{definition}
 
 A remarkable property of the condition $\ACCw$ is the fact that it is stable with respect to \cutelim. This is common knowledge in the community (see~\cite{laurent2016introduction,llhandbook}~for~the~crucial~elements). Anyway, we provide a complete proof in \Cref{subsec:cut-elimination-steps}.
 
 \begin{theorem}[restate = stabilityCutElimination, name = ]
  \label{thm:stability-cut-elimination}
  Let $R,R'$ be \pss such that $R \tocut R'$. If $R \models \ACCw$, then $R' \models \ACCw$.
 \end{theorem}
 
 We conclude this section by pointing out the interest of $\ACCw$ as a correctness criterion only for the \cf \pss of some fragment of $\mllu$. The next result is proven in \Cref{subsec:cut-elimination-steps}.
 
 \begin{proposition}[restate = sequentialityCutElimination, name = ]
  \label{prop:sequentiality-cut-elimination}
  Let $R, R'$ be \pss such that $R \tocut R'$. If $R$ is \seq,~then~$R'$~is~\seq.
 \end{proposition}
 
 By \Cref{prop:sequentiality-cut-elimination} we know, in particular, that the normal form of every sequentializable~\ps is sequentializable. Moreover, the stability of $\ACCw$ entails the following proposition.
 
 \begin{proposition}
  \label{prop:correctness-criterion}
  If, for every \cf \ps $R$ of a fragment $\mathcal{F}$ of $\mllu$, $R \models \ACCw$ $\Leftrightarrow$ $R$ is sequentializable in $\mathcal{F}$, the normal form of a \ps $R$ of $\mathcal{F}$ s.t.~$R \models \ACCw$ is~sequentializable~in~$\mathcal{F}$.
 \end{proposition}
 
 \section{Untyped sequentiality theorem}
  \label{sec:untyped-sequentiality-theorem}
   Sequentialization proofs find a node which ``splits'' the ps into two pss. Girard's original~proof relies on the existence of a terminal ``splitting'' $\otimes$ node (see~\cite{girard1987linear,llhandbook}). In our setting:
 
 \begin{definition}
  \label{def:splitting}
  Let $R$ be a \ps and let $n$ be a terminal $\cutpn$ (resp.~$\otimes$) node of $R$. We say~that:
  \begin{itemize}
   \item
    $n$ \emph{splits} $R$ into two \pss $R_1$ and $R_2$ if $R$ (possibly untyped) is as in \Cref{subfig:desequentialization-cut} (resp.~\ref{subfig:desequentialization-tensor});
   \item
    $n$ is \emph{splitting} for $R$ if there exist two \pss $R_1$ and $R_2$ such that $n$ splits $R$ into $R_1$ and $R_2$.
  \end{itemize}
 \end{definition}
 
 \begin{lemma}
  \label{lemma:splitting}
  Let $R$ be a \ps with at least one terminal $\cutpn$ or $\otimes$ node and no terminal~$\parr$~node. If $R \models \AC$, then there exists a splitting $\cutpn$ or $\otimes$ node of $R$.
 \end{lemma}

 \begin{remark}
  Let $R$ be a \ps and let $n$ be a splitting $\cutpn$ or $\otimes$ node of $R$.
  \begin{itemize}
   \item
  	If $R \models \ACC$, then $n$ splits $R$ into two \pss which are uniquely determined;
  	\item
  	If $R \models \AC$, then $n$ splits $R$ into two \pss $R_1$ and $R_2$ which are \emph{not} uniquely determined: they depend on the choice of distribution of the connected components. But, whatever~the choice, one has $R_1 \models \AC$ and $R_2 \models \AC$;
  	\item
  	If $R \models \ACCw$, then $n$ splits $R$ into two \pss $R_1$ and $R_2$ which are \emph{not} uniquely determined, as before. And it might be the case that $R_1 \models \ACCw$ and $R_2 \models \ACCw$ for some splitting but not for all. Consider, for instance, the \ps $R$ in \Cref{subfig:example-connected-components}. For every $i \in \{0, 1, 2\}$, let $R_i$ be the disjoint union of a terminal $\bot$ node and $i$ terminal $\one$ nodes. Then $n$ splits $R$ into two copies of $R_1$, but $n$ also splits $R$ into $R_0$ and $R_2$,~and~$R_i \models \ACCw$~if~and~only~if~$i = 1$.
  \end{itemize}
  
  \begin{figure}
   \centering
   \hfill
   \begin{subfigure}{0.3\textwidth}
  	\centering
  	\AxiomC{}
  	\RightLabel{\scriptsize$\one$}
  	\UnaryInfC{\infcstrut$\vdash \one$}
  	\RightLabel{\scriptsize$\bot$}
  	\UnaryInfC{\infcstrut$\vdash \one, \bot$}
  	\AxiomC{}
  	\RightLabel{\scriptsize$\one$}
  	\UnaryInfC{\infcstrut$\vdash \one$}
  	\RightLabel{\scriptsize$\bot$}
  	\UnaryInfC{\infcstrut$\vdash \one, \bot$}
  	\RightLabel{\scriptsize\exsc}
  	\UnaryInfC{\infcstrut$\vdash \bot, \one$}
  	\RightLabel{\scriptsize$\otimes$}
  	\BinaryInfC{\infcstrut$\vdash \one, \bot \otimes \bot, \one$}
  	\DisplayProof
  	\caption{A \seqc~proof~$\pi$.}
   \end{subfigure}
   \hfill
   \begin{subfigure}{0.3875\textwidth}
  	\centering
  	\begin{tikzpicture}
	\begin{pgfonlayer}{nodelayer}
		\node [style=tensorwLink] (2) at (0, -0.5) {};
		\node [style=DOTagent] (4) at (0, -1.75) {};
		\node [style=agent] (14) at (0.5, 0) {\scriptsize$\bot$};
		\node [style=agent] (15) at (0, -1.25) {\scriptsize$\smash{\bot{\otimes}\bot}\vphantom{\boldsymbol{1}}$};
		\node [style=botwLink] (18) at (0.5, 0.75) {};
		\node [style=none] (19) at (0.225, -0.85) {\tiny$n$};
		\node [style=agent] (20) at (-0.5, 0) {\scriptsize$\bot$};
		\node [style=botwLink] (21) at (-0.5, 0.75) {};
		\node [style=agent] (22) at (-1, -1.25) {\scriptsize$\boldsymbol{1}$};
		\node [style=onewLink] (23) at (-1, -0.5) {};
		\node [style=DOTagent] (25) at (-1, -1.7625) {};
		\node [style=agent] (26) at (1, -1.25) {\scriptsize$\boldsymbol{1}$};
		\node [style=onewLink] (27) at (1, -0.5) {};
		\node [style=DOTagent] (28) at (1, -1.7625) {};
	\end{pgfonlayer}
	\begin{pgfonlayer}{edgelayer}
		\draw [style=DinRIGHT] (14) to (2);
		\draw [style=simple] (2) to (15);
		\draw [style=Dsimple] (15) to (4);
		\draw [style=simple] (18) to (14);
		\draw [style=simple] (21) to (20);
		\draw [style=DinLEFT] (20) to (2);
		\draw [style=simple] (23) to (22);
		\draw [style=Dsimple] (22) to (25);
		\draw [style=simple] (27) to (26);
		\draw [style=Dsimple] (26) to (28);
	\end{pgfonlayer}
\end{tikzpicture}
  	\caption{\label{subfig:example-connected-components}The desequentialization $R = \pi^\deseq$~of~$\pi$.}
   \end{subfigure}
   \hfill \null
   \caption{A \ps for which the splitting is not uniquely determined.}
  \end{figure}
 \end{remark}
 
 We show that, if a \ps $R$ satisfies a particular geometric condition (it is a $\wten$-\ps), $R$ has no terminal $\bot$ or $\parr$ node and $R \models \Cw$, then $R$ is connected (\Cref{lemma:wten-connected}), so the splitting of $R$ into two \pss $R_1$ and $R_2$ is uniquely determined. And we prove that, if $R \models \ACCw$, then $R_1 \models \ACCw$ and $R_2 \models \ACCw$: this is the key step in the proof of \Cref{thm:untyped-sequentiality}.
 
 The notions we introduce next are illustrated by the examples in \Cref{fig:examples}.
 
 \begin{figure}
  \centering
  \begin{subfigure}{0.19\textwidth}
   \centering
   \begin{tikzpicture}
	\begin{pgfonlayer}{nodelayer}
		\node [style=cutwLink] (32) at (0, 0.625) {};
		\node [style=onewLink] (33) at (0.625, 1.25) {};
		\node [style=parrwLink] (34) at (-0.625, 1.25) {};
		\node [style=onewLink] (35) at (-1.25, 1.875) {};
		\node [style=parrwLinkHighlighted] (36) at (0, 1.875) {};
		\node [style=botwLinkHighlighted] (37) at (-0.625, 2.5) {};
		\node [style=botwLinkHighlighted] (38) at (0.625, 2.5) {};
	\end{pgfonlayer}
	\begin{pgfonlayer}{edgelayer}
		\draw [style=out1LEFT] (34) to (32);
		\draw [style=out1RIGHT] (33) to (32);
		\draw [style=out1RIGHT] (36) to (34);
		\draw [style=out1LEFT] (35) to (34);
		\draw [style=out1LEFT] (37) to (36);
		\draw [style=out1RIGHT] (38) to (36);
	\end{pgfonlayer}
\end{tikzpicture}
   \caption{\label{subfig:example-erasing}}
  \end{subfigure}
  \begin{subfigure}{0.19\textwidth}
   \centering
   \begin{tikzpicture}
	\begin{pgfonlayer}{nodelayer}
		\node [style=cutwLink] (32) at (0, 0.625) {};
		\node [style=onewLink] (33) at (0.625, 1.25) {};
		\node [style=parrwLink] (34) at (-0.625, 1.25) {};
		\node [style=onewLink] (35) at (-1.25, 1.875) {};
		\node [style=parrwLink] (36) at (0, 1.875) {};
		\node [style=botwLink] (37) at (-0.625, 2.5) {};
		\node [style=botwLink] (38) at (0.625, 2.5) {};
		\node [style=DOTagent] (39) at (0, 1.125) {};
		\node [style=DOTagent] (40) at (0.625, 1.75) {};
	\end{pgfonlayer}
	\begin{pgfonlayer}{edgelayer}
		\draw [style=out1LEFT] (34) to (32);
		\draw [style=out1RIGHT] (33) to (32);
		\draw [style=out1LEFT] (35) to (34);
		\draw [style=out1LEFT] (37) to (36);
		\draw [style=simpleB] (38) to (40);
		\draw [style=simpleB] (36) to (39);
	\end{pgfonlayer}
\end{tikzpicture}
   \caption{\label{subfig:example-w-compatible}}
  \end{subfigure}
  \begin{subfigure}{0.19\textwidth}
   \centering
   \begin{tikzpicture}
	\begin{pgfonlayer}{nodelayer}
		\node [style=cutwLink] (32) at (0, 0.625) {};
		\node [style=onewLink] (33) at (0.625, 1.25) {};
		\node [style=parrwLink] (34) at (-0.625, 1.25) {};
		\node [style=onewLink] (35) at (-1.25, 1.875) {};
		\node [style=parrwLink] (36) at (0, 1.875) {};
		\node [style=botwLink] (37) at (-0.625, 2.5) {};
		\node [style=botwLink] (38) at (0.625, 2.5) {};
		\node [style=DOTagent] (39) at (-1.25, 1.125) {};
		\node [style=DOTagent] (40) at (0.625, 1.75) {};
	\end{pgfonlayer}
	\begin{pgfonlayer}{edgelayer}
		\draw [style=out1LEFT] (34) to (32);
		\draw [style=out1RIGHT] (33) to (32);
		\draw [style=out1RIGHT] (36) to (34);
		\draw [style=out1LEFT] (37) to (36);
		\draw [style=simpleB] (38) to (40);
		\draw [style=simpleB] (35) to (39);
	\end{pgfonlayer}
\end{tikzpicture}
   \caption{\label{subfig:example-not-w-compatible}}
  \end{subfigure}
  \begin{subfigure}{0.19\textwidth}
   \centering
   \begin{tikzpicture}
	\begin{pgfonlayer}{nodelayer}
		\node [style=cutwLink] (32) at (0, 0.625) {};
		\node [style=onewLink] (33) at (0.575, 1.25) {};
		\node [style=parrwLink] (34) at (-0.625, 1.25) {};
		\node [style=onewLink] (35) at (-1.25, 1.875) {};
		\node [style=parrwLink] (36) at (0, 1.875) {};
		\node [style=botwLink] (37) at (-0.625, 2.5) {};
		\node [style=botwLink] (38) at (0.625, 2.5) {};
	\end{pgfonlayer}
	\begin{pgfonlayer}{edgelayer}
		\draw [style=out1RIGHT] (36) to (34);
		\draw [style=out1LEFT] (37) to (36);
		\draw [style=out1RIGHT] (38) to (36);
		\draw [style=out1LEFThighlight] (35) to (34);
		\draw [style=out1LEFThighlight] (34) to (32);
		\draw [style=up1RIGHThighlight] (32) to (33);
	\end{pgfonlayer}
\end{tikzpicture}
   \caption{}
  \end{subfigure}
  \begin{subfigure}{0.19\textwidth}
   \centering
   \begin{tikzpicture}
	\begin{pgfonlayer}{nodelayer}
		\node [style=cutwLink] (32) at (0, 0.625) {};
		\node [style=onewLink] (33) at (0.625, 1.25) {};
		\node [style=parrwLink] (34) at (-0.625, 1.25) {};
		\node [style=onewLink] (35) at (-1.25, 1.875) {};
		\node [style=parrwLink] (36) at (0, 1.875) {};
		\node [style=botwLink] (37) at (-0.625, 2.5) {};
		\node [style=botwLink] (38) at (0.625, 2.5) {};
	\end{pgfonlayer}
	\begin{pgfonlayer}{edgelayer}
		\draw [style=out1RIGHT] (33) to (32);
		\draw [style=out1LEFT] (35) to (34);
		\draw [style=out1RIGHT] (38) to (36);
		\draw [style=out1LEFThighlight] (37) to (36);
		\draw [style=out1RIGHThighlight] (36) to (34);
		\draw [style=out1LEFThighlight] (34) to (32);
	\end{pgfonlayer}
\end{tikzpicture}
   \caption{}
  \end{subfigure}
  \caption{\label{fig:examples}From left to right: a $\wten$-\ps with \er nodes highlighted, a \wcs, a non-\wcs, a \wsp, a non-\wsp.}
 \end{figure}
 
 \begin{definition}[By induction on $\prec$]
  \label{def:erasing}
  A node of a graph is \emph{\er} when it is a $\bot$, $\parr$~or~$\bullet$~node whose premises are all conclusions of \er nodes.
 \end{definition}
 
 \begin{remark}
  An erasing node of a \ps $R$ is erasing in $R^\varphi$ too, for every switching $\varphi$ of $R$.
 \end{remark}
 
 \begin{definition}
  A switching $\psi$ of a \ps $R$ is \emph{\wc} if, for every $\parr$ node $n$ of $R$ such that exactly one of its premises, say $a$, is not the conclusion of an~\er~node,~we~have~$\psi(n) = a$.
 \end{definition}
 
 \begin{remark}
  Let $R$ be a \ps and let $n$ be a $\parr$ node of $R$ with exactly one premise $a$ which is the conclusion of an \er node. If $\psi$ is a \wcs, then $a$ is a~conclusion~of~$R^\psi$.
% Let $G$ be a graph and let $n$ be a $\parr$ or $\ctpn$ node of $G$ with two premises~such~that~exactly one premise $a$ of $n$ is the conclusion of an \er node. If $\psi$ is a \wcs of $G$, then $a$ is a conclusion of $G^\psi$.
 \end{remark}
 
 \begin{definition}
  A \spath $\gamma$ of a \ps $R$ is \emph{\ws} if, for every $\parr$~node~$n$~of~$\gamma$~such that exactly one premise $a$ of $n$ is the conclusion of an \er node, we have $a \notin \gamma$.
%Let $G$ be a graph. A \spath $\gamma$ of $G$ is \emph{\ws} if, for every $\parr$ or $\ctpn$ node $n$ of $\gamma$ with two premises such that exactly one of its premises, say $a$, is~the~conclusion of an \er node, we have $a \notin \gamma$.
 \end{definition}
 
 \begin{remark}
  \label{rmk:w-switching-erasing}
  Let $R$ be a \ps, let $a$ be the conclusion of an \er node and a premise of a $\parr$ node $n$ of $R$, and let $\gamma$ be a \wsp of $R$. If $a \in \gamma$, then $n$ is \er.
 \end{remark}
 
 \begin{remark}
  \label{rmk:w-switching-paths}
  $\gamma$ path of a \ps $R$ is \ws iff $\gamma$ is~a~path~of~$R^\psi$~for~some~\wc~$\psi$~of~$R$.
 \end{remark}
 
 \newcommand{\Thrs}{Threads\xspace}
 
 \begin{definition}
  \label{def:thread}
  An \emph{\thr}, or
  \renewcommand{\thr}{thread\xspace}%
  \emph{\thr} for short,~is~a~(connected)~graph~$G$~with~exactly one $\bot$ node and such that every other node of $G$ has exactly one premise.
 \end{definition}
 
 \renewcommand{\thr}{thread\xspace}
 \newcommand{\thrs}{threads\xspace}
 \newcommand{\Thr}{Thread\xspace}
 
 \begin{example}
  In \Cref{subfig:example-w-compatible}, two connected components are threads: the one which contains~a terminal $\bot$ node, and the one containing a $\bot$ node and a terminal $\parr$ node.
 \end{example}
 
 \begin{remark}
  \label{rmk:thread}
  Let $R$ be a \ps and let $\varphi$ be a switching of $R$. Then every \er node of $R$ has at most one premise in $R^\varphi$. In particular, if $\varphi$ is \wc, then, for every $C \in \cc(R^\varphi)$,~$C$~is~a \thr if and only if $C$ contains only \er nodes of $R$ and an (\er)~$\bullet$~node of~$R^\varphi$.
 \end{remark}
 
 \begin{definition}
  \label{def:cw-forall}
  Let $R$ be a \ps and let $\varphi$ be a switching of $R$. We write:
  \begin{itemize}
   \item
    $R^\varphi \models \Cwf$ if any connected component of $R^\varphi$ either has no \er node of $R$ or is a~\thr;
   \item
    $R \models \Cwf$ if, for every \wcs $\psi$ of $R$, we have $R^\psi \models \Cwf$.
  \end{itemize}
 \end{definition}
 
 \begin{example}
  \label{example:cw-forall}
  Let $R$ be the \ps in \Cref{subfig:example-erasing}, and let $\psi$ (resp.~$\varphi$) be the switching~of~$R$~which induces the graph in \Cref{subfig:example-w-compatible} (resp.~\ref{subfig:example-not-w-compatible}). Then $R^\psi \models \Cwf$, but $R^\varphi \models \Cwf$ does not hold: the connected component which contains the $\cutpn$ node is not a thread, and it contains a $\bot$ node (\er). However, $R \models \Cwf$: $\varphi$ is \emph{not} a \wcs! We will~see~that~the~condition $R \models \Cwf$ is equivalent to $R$ being a $\wten$-\ps (\Cref{def:wten,lemma:cw-forall-tensor}).
 \end{example}
 
 \begin{remark}
  \label{rmk:cw-one}
  For $R$ \ps, $\varphi$ switching of $R$ s.t.~$R^\varphi \models \Cwf$: $R^\varphi \models \Cw$ if and only if exactly one connected component of $R^\varphi$ has no \er node of $R$ (thus,~the~other~components~are~threads).
 \end{remark}
 
 The proof of Lemmas~\ref{lemma:cw-forall},~\ref{lemma:cw-forall-tensor} and~\ref{lemma:wten-connected} can be found in \Cref{subsec:cw-forall}.
 
 \begin{lemma}[restate = cwForall, name = ]
  \label{lemma:cw-forall}
  Let $R$ be a \ps. We have $R \models \Cwf$ iff every \wsp of $R$~starting~from~an \er node $n$ and having the conclusion of $n$ as its first arc crosses only \er nodes.
 \end{lemma}

 \begin{definition}
  \label{def:wten}
  A \ps $R$ is a \emph{$\wten$-\ps} if, for every $\cutpn$ or $\otimes$ node $n$ of $R$, no~premise~of~$n$~is the conclusion of an \er node.
 \end{definition}
 
% The proof of the two following results can be found in \Cref{subsec:cw-forall}.
 
 \begin{lemma}[restate = cwForallTensor, name = ]
  \label{lemma:cw-forall-tensor}
  Let $R$ be a \ps. The following are equivalent:
  \begin{enumerate}[(i)]
   \item \label{itm:wten-graph}
    $R$ is a $\wten$-\ps;
   \item \label{itm:cw-forall}
    $R \models \Cwf$;
   \item \label{itm:cw-forall-existential}
    There exists a \wcs $\psi$ of $R$ such that $R^\psi \models \Cwf$.
  \end{enumerate}
 \end{lemma}

 \begin{lemma}[restate = wtenConnected, name = ]
  \label{lemma:wten-connected}
  Let $R$ be a $\wten$-\ps such that $R \models \Cw$. Then:
  \begin{enumerate}[(i)]
   \item \label{itm:existence-non-erasing}
    There exists a non-\er node of $R$;
   \item \label{itm:erasing-component}
    If $R$ contains a $\cutpn$ or $\otimes$ node $n$ and $R_0$ is a connected component~of~$R$~which~does~not contain $n$, then every node of $R_0$ is \er;
   \item \label{itm:connectivity}
    If $R$ has no terminal $\bot$ or $\parr$ node, then $R$ is a connected graph.
  \end{enumerate}
 \end{lemma}
 
 We now prove our main result: in order for a $\wten$-\ps $R$ to be \seq, it is sufficient that every \sg of $R$ is acyclic and that the number of its connected components~is exactly one more than the number of $\bot$ nodes.
 
 \begin{theorem}[restate = untypedSequentiality, name = ]
  \label{thm:untyped-sequentiality}
  Let $R$ be a $\wten$-\ps. If $R \models \ACCw$, then $R$ is \seq.
 \end{theorem}
 
 \begin{proof}
  By induction on $\sharp \ar(R)$. We provide here the details only of the key case (full proof in \Cref{subsec:cw-forall}). If $R$ has a terminal $\bot$ or $\parr$ node, the result is a straightforward application of the induction hypothesis. Otherwise, $R$ is connected (\Cref{lemma:wten-connected}-\ref{itm:connectivity}). If $R$~has~no~terminal~$\cutpn$ or $\otimes$ node, then $R$ is a terminal $\axpn$ or $\onepn$ node, and one easily concludes. 
  
  The crucial case is when $R$ has a terminal $\cutpn$ or $\otimes$ node. By \Cref{lemma:splitting}, there exist a $\cutpn$ or $\otimes$ node $n$ and \pss $R_1, R_2$ such that $n$ splits $R$ into $R_1$ and $R_2$. Moreover, since $R$ is connected, $R_1, R_2$ are uniquely determined. It is sufficient to prove that $R_1, R_2$ are \seq: then $R$ is $n$-\seq. We only prove that $R_1$ is \seq: $R_2$ is \seq by an analogous argument. Since $\smash{\sharp \ar(R_1) \leq \sharp \ar(R) - 1 < \sharp \ar(R)}$, we can conclude by applying~the~induction~hypothesis on $R_1$. It is easy to check that $R_1 \models \AC$, and $R_1$ is a $\wten$-\ps~because $R$ is.
  
  The only thing left to prove is that $R_1 \models \Cw$. By \Cref{lemma:cw-forall-tensor}, $R \models \Cwf$ and $R_1 \models \Cwf$. Let $\psi_1$ be a \wcs of $R_1$, and let $\psi$ be any extension of $\psi_1$ to a \wcs of $R$. Then $R^\psi \models \Cwf$ and $R_1^{\psi_1} \models \Cwf$. Since $R^\psi \models \Cwf$ and $R^\psi \models \Cw$, by \Cref{rmk:cw-one} exactly one connected component $C_{\neg \w}$ of $R^\psi$ has no \er node of $R$. And since $\cutpn$ and $\otimes$ nodes are never \er, $n$ is a node of $C_{\neg \w}$. Then the connected component $C_{\neg \w}'$ of $R_1^{\psi_1}$ which contains a premise of $n$ has no \er node of $R$. Also, for any~$C_1 \in \cc(R_1^{\psi_1})$~with~$C_1 \neq C_{\neg \w}'$, $C_1$ is actually a connected component of $R^\psi$, and $C_1 \neq C_{\neg \w}$. Since $R^\psi \models \Cwf$, $C_1$ is a \thr, so $C_1$ has a $\bot$ node. Thus,~$C_{\neg \w}'$~is~the~only~connected~component~of~$R_1^{\psi_1}$~with~no~\er~node~of~$R$. Since $R_1^{\psi_1} \models \Cwf$, by \Cref{rmk:cw-one} we have $R_1^{\psi_1} \models \Cw$. Then~$R_1 \models \Cw$ (\Cref{prop:acyclicity-cc}-\ref{itm:acyclicity-cc-existential}).
 \end{proof}
 
 \begin{remark}
  \label{remark:AppendixRegnier}
  In Appendix B of~\cite{regnier1992lambda}, the author proves a sequentialization result in the typed framework of $\mllu$ for every $\wten$-\ps $R$ such that $R \models \ACCw$: he actually states the result for the fragment $\btenlls$ (see also \Cref{rem:CCforMLLu,def:btenlls}). The proof given in~\cite{regnier1992lambda} is a bit sketchy and uses a technique different from ours (it relies on the mutilation of some $\parr$ node); it might be the case that one could adapt that proof to the (more~general)~untyped~framework considered in \Cref{thm:untyped-sequentiality}, which would yield~an~alternative~proof~of~the~result.
 \end{remark}
 
 \begin{remark}
  The \ps in \Cref{subfig:example-erasing} satisfies the hypotheses of \Cref{thm:untyped-sequentiality} and is \seq. The \pss in \Cref{fig:counterexamples} are not \seq: they satisfy $\ACCw$, but they are not $\wten$-\pss.
 \end{remark}
 
 \begin{convention}
  In view of \Cref{prop:sequentiality-cut-elimination,prop:correctness-criterion}, \seqc proofs and \pss in the following sections are assumed to be \cf, unless explicitly stated otherwise.
 \end{convention}
 
% \begin{corollary}
%  Let $R$ be an untyped $\wten$-proof-structure. If $R \models \mathsf{ACC}_{\sharp w}$, the normal form of $R$ is sequential.
% \end{corollary}
% 
% \begin{proof}
%  The result follows immediately from Lemma~\ref{theorem:stability-cut-elimination} and Corollary~\ref{thm:untyped-sequentiality}.
% \end{proof}
 
 \section{The fragment \texorpdfstring{$\btenll$}{(NBT)LL}}
  \label{sec:btenll}
   We present a fragment of $\mllu$ which satisfies the hypothesis of \Cref{prop:correctness-criterion} and prove two notable properties: 1) a refined version of the sequentialization process, allowing us to assign a ``canonical'' position to $\bot$ rules (\Cref{thm:untyped-sequentialityBottomTens}); 2) the equivalence on sequent calculus proofs is ``easy'' to decide: just check that they have the same desequentialization (\Cref{cor:equivalenceFeasible}). There are two kinds of formulae in it, written $A$ and $E$: formulae that can be the type of a conclusion of an \er node (of kind $E$) are forbidden~from~being~types~of~premises~of~$\otimes$~nodes.
 
 %We present a fragment of $\mllu$ which satisfies the hypothesis of \Cref{prop:correctness-criterion}. There are~two kinds of formulae in it, denoted by $A$ and $E$: formulae that can be the type of a~conclusion~of~an \er node (of kind $E$) are forbidden~from~being~types~of~premises~of~$\otimes$~nodes.
 
 % We present some concrete fragments of $\mell$ satisfying the hypothesis of \Cref{prop:correctness-criterion}. In each of these fragments there are two kinds of formulae, which will be denoted by $A$ and $E$ (possibly with indexes). Therefore, in this section we have three kinds of formulae: formulae of kind $A$, formulae of kind $E$, and general formulae which can be of any kind~and~are~denoted~by $F, G, H$. The idea of the following definition is to restrain the usage of formulae of $\mell$ that can be the type of a conclusion of an \er node: we forbid any such formula, which~will~be~of kind $E$, from being the type of a premise of a $\otimes$ node.
 
 \begin{definition}
  \label{def:btenll}
  The fragment $\btenll$ of $\mllu$ is defined by the following grammar:
  \[
   A \Coloneqq X \mid \one \mid A \otimes A \mid A \parr A \mid A \parr E \mid E \parr A \qquad E \Coloneqq \bot \mid E \parr E
  \]
 \end{definition}
 
% \begin{definition}
% 	\label{def:tenll}
% 	The fragment $\wntenlls$ of $\mell$ is defined by the following grammar: \small
% 	\[
% 	A \Coloneqq X \mid \one \mid A \otimes A \mid A \parr A \mid A \parr E \mid E \parr A \mid{} !A \mid{} !E \qquad E \Coloneqq \bot \mid E \parr E \mid{} ?A \mid{} ?E
% 	\]
% \end{definition} \normalsize
% 
% \begin{remark}
% 	If $F$ is a formula of $\wntenlls$, then $F^\perp$ may not be a formula of $\wntenlls$.
% \end{remark}
% 
% In $\wntenlls$ a generalization of a crucial property of polarized $\linl$ holds: if a provable sequent $\Gamma$ contains a positive formula, then every other formula of $\Gamma$ is negative (\Cref{cor:polarized-context}).
 
% \begin{proposition}[restate = tenll, name = ]
% 	\label{prop:tenll}
% 	\quad
% 	\begin{enumerate}[(i)]
% 		\item \label{itm:tenll-context}
% 		Let $\Sigma$ be a sequent of $\wntenlls$ and let $E$ be a formula of $\wntenlls$ s.t.~$E^\perp$ is a formula of $\wntenlls$ too. If $\Sigma, E^\perp$ is provable, then $\Sigma = E_1, \dots, E_k$ for some $k \geq 0$;
% 		\item \label{itm:tenll-context-par}
% 		Let $\Sigma$ be a sequent of $\wntenlls$, let $E, F$ be formulae of $\wntenlls$, let $H \in \{E^\perp \parr F, F \parr E^\perp\}$. If the sequent $\Sigma, H$ is provable, then $\Sigma = E_1, \dots, E_k$ for some $k \geq 0$;
% 		\item \label{itm:tenll-not-provable}
% 		Let $E_1, E_2$ be formulae of $\wntenlls$. Then $\smash{\Gamma, E_1^\perp \parr E_2^\perp}$ is \emph{not} provable in $\wntenlls$, for every sequent $\Gamma$ of $\wntenlls$.
% 	\end{enumerate}
% \end{proposition}
 
 \begin{remark}
  \label{rmk:erasing}
  \hspace{-1em}
  \begin{minipage}{0.42\textwidth}
   \begin{enumerate}[(i)]
    \item
     $\mll$ is a fragment of $\btenll$.
   \end{enumerate}    
  \end{minipage}
  \begin{minipage}{0.43\textwidth}
   \begin{enumerate}[(i)]
    \setcounter{enumi}{1}
    \item \label{itm:wten-proof-structure}
     Any \ps of $\btenll$ is a $\wten$-\ps.
   \end{enumerate} 
  \end{minipage}

  %More precisely, for every formula $F$ of $\mll$, there exists a formula $A$ of $\btenll$ such that $F = A$.
  %$\mll$ is a fragment of $\btenll$: for $F$ $\mll$ formula: there is $A$ $\btenll$ formula s.t.~$F = A$ or $F = A^\perp$.
 \end{remark}
 
% \begin{remark}
%  \label{rmk:erasing}
%  Let $R$ be a ps of $\btenll$.
%  \begin{enumerate}[(i)]
%   \item \label{itm:conclusion-erasing-type}
%    If $n$ is an \er node of $R$, then the conclusion of $n$ has type $E$;
%   \item 
%    $R$ is a $\wten$-\ps. Indeed, every premise of a~$\otimes$~node~of~$R$~has~type~$A$ and thus, by \Cref{itm:conclusion-erasing-type}, it is not the conclusion of an \er node.
%  \end{enumerate}
% \end{remark}
 
 By \Cref{thm:untyped-sequentiality,rmk:erasing} (and \ref{rmk:sequential}), $\ACCw$ is a correctness criterion for $\btenll$:
 
 \begin{corollary}
  \label{cor:sequentialization}
  For any \ps $R$ of $\btenll$: $R \models \ACCw$ $\Leftrightarrow$ $R$ is sequentializable in $\btenll$.
 \end{corollary}
 
 \begin{remark}
  The two \pss in \Cref{fig:counterexamples} are \emph{not} \pss of $\btenll$. Indeed, neither $\smash{(\one \parr \one) \otimes \bot}$ or $\smash{(X^\perp \parr \one) \otimes \bot}$ is a formula of $\btenll$. Therefore, \Cref{cor:sequentialization}~is~consistent~with~\Cref{rmk:sequentializable}.
 \end{remark}
 
 \begin{remark}
  \label{rem:CCforMLLu}
  \Cref{rmk:sequentializable} entails that $\ACCw$ is not enough to sequentialize all the \pss of $\mllu$. More deeply, no feasible correctness criterion can ever be found for $\mllu$ due to the complexity result in~\cite{LincolnW94}. Therefore, $\btenll$ is a fragment of $\mllu$ that contains the units and for which a feasible correctness criterion exists, as~it~was~already~noticed~in~the~restricted~framework of $\btenlls$ (\Cref{def:btenlls}) in Appendix B of~\cite{regnier1992lambda}.
 \end{remark}
 
% \begin{figure}
%  \centering
%  %
%  \AxiomC{\infcstrut$\vdash \Gamma$}
%  \RightLabel{\scriptsize$\bot$}
%  \UnaryInfC{\infcstrut$\vdash \Gamma, \bot_1$}
%  \RightLabel{\scriptsize$\bot$}
%  \UnaryInfC{\infcstrut$\vdash \Gamma, \bot_1, \bot_2$}
%  \DisplayProof
%  $\equivseqc$\hspace{-0.625em}
%  \AxiomC{\infcstrut$\vdash \Gamma$}
%  \RightLabel{\scriptsize$\bot$}
%  \UnaryInfC{\infcstrut$\vdash \Gamma, \bot_2$}
%  \RightLabel{\scriptsize$\bot$}
%  \UnaryInfC{\infcstrut$\vdash \Gamma, \bot_2, \bot_1$}
%  \RightLabel{\scriptsize\exsc}
%  \UnaryInfC{\infcstrut$\vdash \Gamma, \bot_1, \bot_2$}
%  \DisplayProof
%  %
%  \quad
%  %
%  \AxiomC{\infcstrut$\vdash \Gamma, A, B$}
%  \RightLabel{\scriptsize$\parr$}
%  \UnaryInfC{\infcstrut$\Gamma, A \parr B$}
%  \RightLabel{\scriptsize$\bot$}
%  \UnaryInfC{\infcstrut$\Gamma, A \parr B, \bot$}
%  \DisplayProof
%  $\equivseqc$\hspace{-1.25em}
%  \AxiomC{\infcstrut$\vdash \Gamma, A, B$}
%  \RightLabel{\scriptsize$\bot$}
%  \UnaryInfC{\infcstrut$\Gamma, A, B, \bot$}
%  \RightLabel{\scriptsize\exsc}
%  \UnaryInfC{\infcstrut$\Gamma, A, \bot, B$}
%  \RightLabel{\scriptsize\exsc}
%  \UnaryInfC{\infcstrut$\Gamma, \bot, A, B$}
%  \RightLabel{\scriptsize$\parr$}
%  \UnaryInfC{\infcstrut$\Gamma, \bot, A \parr B$}
%  \RightLabel{\scriptsize\exsc}
%  \UnaryInfC{\infcstrut$\Gamma, A \parr B, \bot$}
%  \DisplayProof
%  %
%  \caption{\label{fig:rule-permutations}Rule permutations.}
% \end{figure}
 
Until the end of this section: we consider general graphs (not only \llgs); a \ps $R$ comes with a \emph{partial} jump function $J_R$ from $\w(R)$ to $\ve(R)$, with domain $\dom(J_R)$; if $\varphi$ is a switching of $R$, $R^\varphi$ denotes the \sg of $R$ induced by $\varphi$ where, for every $n \in \dom(J_R)$, we add an arc (\emph{jump}) from $n$ to $J_R(n)$. We generalize $\Cw$ (and $\ACCw$): $R^\varphi \models \Cw$ if $\smash{\sharp \cc(R^\varphi) = \sharp \w(R) - \sharp \dom(J_R) + 1}$, and $R \models \Cw$~if~this~holds~for~every~switching~$\varphi$~of~$R$.
 
 \begin{definition}
 \label{def:Rjumps}
  A \ps $R$ is \emph{\jt} if $\dom(J_R) = \w(R)$, \emph{\jf} if $\dom(J_R)$ is empty, \emph{\jc} if it is \jt and $R \models \ACC$. For $n \in \dom(J_R)$, we denote by $R_{\widecheck{n}}$ the \ps obtained from $R$ by restricting the domain of $J_R$ as follows: $\smash{\dom(J_{R_{\widecheck{n}}}) \coloneq \dom(J_R) \setminus \{n\}}$. The \emph{underlying \jf \ps of $R$} is $\smash{\ujf{R} \coloneq (R_{\widecheck{n_1}} \cdots)_{\widecheck{n_k}}}$, where $\dom(J_R) = \{n_1, \dots, n_k\}$.
 \end{definition}
 
 \begin{remark}
  \label{rmk:jump}
  Let $R$ be a \ps.
  \begin{enumerate}[(i)]
   \item \label{itm:acc-jump}
 	If $R$ is \jt, then $R \models \ACC$ if and only if $R \models \ACCw$.
   \item \label{itm:removing-jump}
    If $R \models \ACCw$ then, for $n \in \dom(J_R)$, $R_{\widecheck{n}} \models \ACCw$ (\Cref{prop:acyclicity-cc}), thus $\ujf{R} \models \ACCw$.
%    For $n \in \dom(J_R)$: $R \models \ACCw \Rightarrow R_{\widecheck{n}} \models \ACCw$ (\Cref{prop:acyclicity-cc}). Thus, $\ujf{R} \models \ACCw$.
  \end{enumerate}
 \end{remark}

% \begin{remark}
%  \label{rmk:jump}
%  Let $R$ be a \ps of $\mllu$.
%  \begin{enumerate}[(i)]
%   \item \label{itm:acc-jump}
%    If $R$ is \jt, then $R \models \ACC$ if and only if $R \models \ACCw$;
%   \item \label{itm:removing-jump}
%    For every $n \in \dom(J_R)$, if $R \models \ACCw$, then $R_{\widecheck{n}} \models \ACCw$ (by \Cref{itm:acyclicity-cc-number} of \Cref{prop:acyclicity-cc});
%   \item \label{itm:underlying-jump-free}
%    The \ps $\ujf{R}$ is \jf and, if $R \models \ACCw$, then $\ujf{R} \models \ACCw$ (by \Cref{itm:removing-jump}).
%  \end{enumerate}
% \end{remark}
 
 There are several ways of adding jumps to a non-\jt \ps. In general, the condition $\ACCw$ is not preserved. There is however a ``canonical~way''~that~preserves~$\ACCw$~in~$\btenll$.
 
 \begin{definition}
  Let $R$ be a \ps of $\btenll$. We split $\w(R)$ in two: $\wparr(R) \coloneq \{n : \text{$n \in \w(R)$}$ $\text{and there exists a non-\er $\parr$ node $n_0$ of $R$ such that $n \prec n_0$}\}$, $\smash{\we(R) \coloneq \w(R) \setminus \wparr(R)}$.
 \end{definition}
 
 \begin{definition}[\Cref{fig:canonical-jumps-btenll}]
  \label{def:canonical-jump}
  For $R$ \ps of $\btenll$, $\smash{n \in \w(R) \setminus \dom(J_R)}$ and $m \in \ve(R)$,~we denote by $R_n^m$ the \ps obtained by extending $J_R$ as follows: $J_{R_n^m}(n) \coloneq m$. For every $n \in \wparr(R)$, we set $R_n \coloneq R_n^{n'}$, where $n'$ is the non-\er node whose conclusion~is~a~premise~of~the~least (\Cref{rmk:prec}) non-erasing $\parr$ node $n_0$ of $R$ s.t.~$n \prec n_0$. Finally, we set $\smash{\can{R} \coloneq (\ujf{R}_{n_1} \cdots)_{n_k}}$, where $\wparr(R) = \{n_1, \dots, n_k\}$,\footnote{It is obvious that $\can{R}$ is well-defined: for all $\smash{n_1, n_2 \in \w(R) \setminus \dom(J_R)}$, we have that $(R_{n_1})_{n_2} = (R_{n_2})_{n_1}$.} and $\canext{R}{m} \coloneq ((\can{R})_{m_1}^m \cdots)_{m_h}^m$, where $\we(R) = \{m_1, \dots, m_h\}$.
  
  \begin{figure}
   \centering
   \begin{subfigure}{0.49\textwidth}
    \centering
    \begin{tikzpicture}
	\begin{pgfonlayer}{nodelayer}
		\node [style=DOTagent] (4) at (0, -2.3) {};
		\node [style=agent] (13) at (1, -0.05) {\scriptsize$\smash{\bot{\parr}\bot}\vphantom{(}$};
		\node [style=agent] (14) at (-1, -0.05) {\scriptsize$\smash{X^\perp{\parr}(X{\otimes}\boldsymbol{1})}\vphantom{(}$};
		\node [style=agent] (15) at (0, -1.8) {\scriptsize$\smash{A}\vphantom{X}$};
		\node [style=parrwLink] (23) at (0, -1.05) {};
		\node [style=parrwLink] (33) at (-1, 0.7) {};
		\node [style=agent] (35) at (-1.5, 1.2) {\scriptsize$\smash{X^\perp}\vphantom{X}$};
		\node [style=agent] (36) at (0, 2.45) {\scriptsize$\smash{\boldsymbol{1}}\vphantom{X}$};
		\node [style=agent] (37) at (-1, 2.45) {\scriptsize$\smash{X}\vphantom{X}$};
		\node [style=tensorwLink] (38) at (-0.5, 1.95) {};
		\node [style=agent] (41) at (-0.5, 1.2) {\scriptsize$\smash{X{\otimes}\boldsymbol{1}}\vphantom{X}$};
		\node [style=parrwLink] (42) at (1, 0.7) {};
		\node [style=botwLink] (43) at (0.5, 1.95) {};
		\node [style=agent] (44) at (0.5, 1.2) {\scriptsize$\smash{\bot}\vphantom{X}$};
		\node [style=agent] (46) at (1.5, 1.2) {\scriptsize$\smash{\bot}\vphantom{X}$};
		\node [style=botwLink] (47) at (1.5, 1.95) {};
		\node [style=axwLink] (48) at (-1.375, 3.2) {};
		\node [style=onewLink] (49) at (0, 3.2) {};
		\node [style=botwLink] (50) at (-1.75, -1.05) {};
		\node [style=DOTagent] (51) at (-1.75, -2.3) {};
		\node [style=agent] (52) at (-1.75, -1.8) {\scriptsize$\smash{\bot}\vphantom{X}$};
		\node [style=none] (53) at (-2.075, -1.325) {\tiny$n_1$};
		\node [style=none] (54) at (0.875, 1.675) {\tiny$n_2$};
		\node [style=none] (55) at (1.875, 1.675) {\tiny$n_3$};
		\node [style=none] (56) at (0.275, -1.425) {\tiny$n_0$};
		\node [style=none] (57) at (-1.225, 0.825) {\tiny$n'$};
        \coordinate (58) at ($(43)-(0.45,0)$) {};
        \coordinate (59) at ($(33)+(1,0)$) {};
        \coordinate (60) at ($(33)+(0.75,0)$) {};
        \coordinate (61) at ($(33)+(-0.625,-0.3125)$) {};
        \node [draw=none, anchor=north, trapezium, trapezium angle=110, rounded corners=2pt, inner sep=1.2pt, font={$\scriptstyle\phantom{\bot}$}] (62) at (-1.7, -1.05) {};
        \coordinate (63) at ($(48)+(-0.4,0.35)$) {};
        \node [draw=none, anchor=north, trapezium, trapezium angle=110, rounded corners=2pt, inner sep=1.2pt, font={$\scriptstyle\phantom{\bot}$}] (64) at (-1.8, -1.05) {};
        \coordinate (65) at ($(63)+(-0.19,0.19)$) {};
	\end{pgfonlayer}
	\begin{pgfonlayer}{edgelayer}
		\draw [style=Dsimple] (15) to (4);
		\draw [style=simple] (23) to (15);
		\draw [style=simple] (33) to (14);
		\draw [style=DinLEFT] (35) to (33);
		\draw [style=DinLEFT] (37) to (38);
		\draw [style=DinRIGHT] (36) to (38);
		\draw [style=simple] (38) to (41);
		\draw [style=DinRIGHT] (41) to (33);
		\draw [style=simple] (43) to (44);
		\draw [style=DinLEFT] (44) to (42);
		\draw [style=DinRIGHT] (46) to (42);
		\draw [style=simple] (42) to (13);
		\draw [style=simple] (47) to (46);
		\draw [style=outRIGHT] (48) to (37);
		\draw [style=outLEFT, looseness = 0.625] (48) to (35);
		\draw [style=simple] (49) to (36);
		\draw [style=Dsimple] (52) to (51);
		\draw [style=simple] (50) to (52);
		\draw [style=DinLEFT] (14) to (23);
		\draw [style=DinRIGHT] (13) to (23);
        
        \draw [draw=red, -, draw opacity=0, -latex, line width=0.3pt, postaction={draw, opacity=1, -, line width=1.5pt, shorten >=3pt, shorten <=-1pt}] (43) to [in=60, out=210] (59) to [in=300, out=240] (33);
        \draw [draw=red, -, draw opacity=0, -latex, line width=0.3pt, postaction={draw, opacity=1, -, line width=1.5pt, shorten >=3pt, shorten <=-1pt}] (47) to [in=60, out=150] (58) to [in=60, out=240] (60) to [in=330, out=240] (33);
        \draw [draw=red, dashed, draw opacity=0, -latex, line width=0.3pt, postaction={draw, opacity=1, -, line width=1.5pt, shorten >=3pt, shorten <=-1pt}, in=210, out=330] (50) to (23);
        \draw [draw=red, dashed, draw opacity=0, -latex, line width=0.3pt, postaction={draw, opacity=1, -, line width=1.5pt, shorten >=3pt, shorten <=-1pt}] (62) to [in=195, out=105, looseness=1.3] (33);
        \draw [draw=red, dashed, draw opacity=0, -latex, line width=0.3pt, postaction={draw, opacity=1, -, line width=1.5pt, shorten >=3pt, shorten <=-1pt}] (50) to [in=165, out=120, looseness=0.6] (48);
        \draw [draw=red, dashed, draw opacity=0, -latex, line width=0.3pt, postaction={draw, opacity=1, -, line width=1.5pt, shorten >=3pt, shorten <=-1pt}] (64) to [in=180, out=135, looseness=0.6] (63) to [in=90, out=0, looseness=1.2] (38);
        \draw [draw=red, dashed, draw opacity=0, -latex, line width=0.3pt, postaction={draw, opacity=1, -, line width=1.5pt, shorten >=3pt, shorten <=-1pt}] (64) to [in=180, out=135, looseness=0.6] (63) to [in=90, out=0, looseness=1.2] (38);
        \draw [draw=red, dashed, draw opacity=0, -latex, line width=0.3pt, postaction={draw, opacity=1, -, line width=1.5pt, shorten >=3pt, shorten <=-1pt}] (50) to [in=190, out=160, looseness=0.6] (65) to [in=150, out=10] (49);
	\end{pgfonlayer}
\end{tikzpicture}
    \caption{\label{fig:canonical-jumps-btenll}$\btenll$, $\smash{A \coloneq (X^\perp \parr (X \otimes \one)) \parr (\bot \parr \bot)}$}
   \end{subfigure}
   \begin{subfigure}{0.49\textwidth}
    \centering
    \begin{tikzpicture}
	\begin{pgfonlayer}{nodelayer}
		\node [style=DOTagent] (4) at (0, -2.3) {};
		\node [style=agent] (12) at (1.5, 1.2) {\scriptsize$\smash{\boldsymbol{1}}\vphantom{X}$};
		\node [style=agent] (13) at (1, -0.05) {\scriptsize$\smash{\bot{\otimes}\boldsymbol{1}}\vphantom{(}$};
		\node [style=agent] (14) at (-1, -0.05) {\scriptsize$\smash{\bot{\parr}(\boldsymbol{1}{\otimes}\boldsymbol{1})}\vphantom{(}$};
		\node [style=agent] (15) at (0, -1.8) {\scriptsize$\smash{O}\vphantom{X}$};
		\node [style=agent] (19) at (0.5, 1.2) {\scriptsize$\smash{\bot}\vphantom{X}$};
		\node [style=parrwLink] (23) at (0, -1.05) {};
		\node [style=tensorwLink] (25) at (1, 0.7) {};
		\node [style=botwLink] (26) at (0.5, 1.95) {};
		\node [style=onewLink] (32) at (1.5, 1.95) {};
		\node [style=parrwLink] (33) at (-1, 0.7) {};
		\node [style=botwLink] (34) at (-1.5, 1.95) {};
		\node [style=agent] (35) at (-1.5, 1.2) {\scriptsize$\smash{\bot}\vphantom{X}$};
		\node [style=agent] (36) at (0, 2.45) {\scriptsize$\smash{\boldsymbol{1}}\vphantom{X}$};
		\node [style=agent] (37) at (-1, 2.45) {\scriptsize$\smash{\boldsymbol{1}}\vphantom{X}$};
		\node [style=tensorwLink] (38) at (-0.5, 1.95) {};
		\node [style=onewLink] (39) at (-1, 3.2) {};
		\node [style=agent] (41) at (-0.5, 1.2) {\scriptsize$\smash{\boldsymbol{1}{\otimes}\boldsymbol{1}}\vphantom{X}$};
		\node [style=onewLink] (42) at (0, 3.2) {};
		\node [style=botwLink] (43) at (-1.75, -1.05) {};
		\node [style=DOTagent] (44) at (-1.75, -2.3) {};
		\node [style=agent] (45) at (-1.75, -1.8) {\scriptsize$\smash{\bot}\vphantom{X}$};
		\node [style=none] (46) at (-2.075, -1.325) {\tiny$n_1$};
		\node [style=none] (47) at (-1.85, 1.675) {\tiny$n_2$};
		\node [style=none] (48) at (0.875, 1.675) {\tiny$n_3$};
		\node [style=none] (49) at (0.225, -1.4) {\tiny$m$};
		\node [style=none] (50) at (-0.325, 2.05) {\tiny$m_\otimes$};
        \node [draw=none, trapezium, trapezium angle=110, rounded corners=2pt, inner sep=1.2pt, font={$\scriptstyle\phantom{\varotimes}$}] (51) at ($(38) - (0.125,0)$) {};
        \coordinate (52) at ($(39)+(0,0.3)$) {};
        \coordinate (53) at ($(39)+(-0.05,0.3)$) {};
        \coordinate (54) at ($(51)+(0,0.3)$) {};
        \coordinate (55) at ($(43)+(0,2.125)$) {};
	\end{pgfonlayer}
	\begin{pgfonlayer}{edgelayer}
		\draw [style=Dsimple] (15) to (4);
		\draw [style=simple] (23) to (15);
		\draw [style=simple] (25) to (13);
		\draw [style=DinLEFT] (19) to (25);
		\draw [style=DinRIGHT] (12) to (25);
		\draw [style=simple] (26) to (19);
		\draw [style=simple] (32) to (12);
		\draw [style=simple] (33) to (14);
		\draw [style=simple] (34) to (35);
		\draw [style=DinLEFT] (35) to (33);
		\draw [style=DinLEFT] (37) to (38);
		\draw [style=DinRIGHT] (36) to (38);
		\draw [style=simple] (39) to (37);
		\draw [style=simple] (38) to (41);
		\draw [style=DinRIGHT] (41) to (33);
		\draw [style=simple] (42) to (36);
		\draw [style=Dsimple] (45) to (44);
		\draw [style=simple] (43) to (45);
		\draw [style=DinLEFT] (14) to (23);
		\draw [style=DinRIGHT] (13) to (23);
		\draw [draw=red, -, draw opacity=0, -latex, line width=0.3pt, postaction={draw, opacity=1, -, line width=1.5pt, shorten >=3pt, shorten <=-1pt}, in=210, out=330] (34) to (38);
		\draw [draw=red, -, draw opacity=0, -latex, line width=0.3pt, postaction={draw, opacity=1, -, line width=1.5pt, shorten >=3pt, shorten <=-1pt}, in=330, out=210] (26) to (38);
        \draw [draw=red, -, draw opacity=0, -latex, line width=0.3pt, postaction={draw, opacity=1, -, line width=1.5pt, shorten >=3pt, shorten <=-1pt}] (43) to [in=180, out=90, looseness=1.1] (53) to [in=90, out=0, looseness=1.2] (54) to [in=90, out=270] (51);
	\end{pgfonlayer}
\end{tikzpicture}
    \caption{$\icomll$, $\smash{O \coloneq (\bot \parr (\one \otimes \one)) \parr (\bot \otimes \one)}$}
   \end{subfigure}
   \caption{\label{fig:canonical-jumps}Highlighted, canonical jumps in $\btenll$ (\Cref{def:canonical-jump}) and $\icomll$ (\Cref{def:canonical-jump-intuitionistic}). In the $\btenll$ case, every dashed arc represents a possible choice of jump for $n_1$.}
  \end{figure}
 \end{definition}
 
% \begin{example}
%  Let $R$ be the unique \jf \ps with no $\axpn$ node and with exactly one conclusion of type $\smash{(1 \otimes 1) \parr \bot}$, and let $n$ be the $\bot$ node of $R$. Then $J_{R_n}(n)$ is the~$\otimes$~node~of~$R$.
% \end{example}
 
 \begin{remark}
  \label{rmk:canonical-jump}
  Let $R$ be a \ps of $\btenll$, and let $m$ be a node of $R$. Then:
  \begin{enumerate}[(i)]
   \item \label{itm:canonical-jump-total}
    $\canext{R}{m}$ is \jt;
   \item \label{itm:canonical-jump-terminal}
    If $R$ has no terminal \er node, then $\w(R) = \wparr(R)$;
   \item \label{itm:canonical-jump}
    If $\w(R) = \wparr(R)$, then $\canext{R}{m} = \can{R}$.
  \end{enumerate}
 \end{remark}
 
% The two following results are established in \Cref{subsec:equivalence-canonical-jumps} by using \Cref{itm:acyclicity-cc-number} of \Cref{prop:acyclicity-cc}.
 
% \begin{proof}
%  It is enough to prove $R_n \models \AC$ (\Cref{subsec:equivalence-canonical-jumps}) and apply \Cref{itm:acyclicity-cc-number} of \Cref{prop:acyclicity-cc}.
% \end{proof}

% From \Cref{rmk:jump,rmk:canonical-jump} and \Cref{lemma:adding-jump} we deduce the following result.
%
% \begin{corollary}
%  \label{cor:adding-jump}
%  For $R$ \ps of $\btenll$, $m$ non-\er node: $R \models \ACCw$ $\Rightarrow$ $\canext{R}{m} \models \ACC$.
% \end{corollary}
 
 With every \seqc proof $\pi$ in $\mllu$ are associated several \jt \pss: the desequentialization function $\pi \mapsto \pi^\deseq$ (see~\Cref{fig:desequentialization}) is extended to~a~binary~relation~$\deseqjump$~between \seqc proofs and \jt \pss (Definition~3.2 in~\cite{heijltjes2014no}).
% With a proof $\pi$ in $\mllu$ are associated several \jt \pss: the function $\pi \mapsto \pi^\deseq$ (\Cref{fig:desequentialization}) becomes a binary relation~$\deseqjump$~between~\seqc~proofs~and~\jt~\pss~(\cite{heijltjes2014no}).
 
 \begin{remark}
  \label{rmk:underlying-jump-free}
  Let $\pi$ be a \seqc proof in $\mllu$, and let $R$ be~a~\jt~\ps~of~$\mllu$. If $\pi \deseqjump R$, then $\ujf{R} = \pi^\deseq$.
%  For every $\pi$ proof in $\mllu$, $R$ \jt \ps of $\mllu$: if $\pi \deseqjump R$, then $\ujf{R} = \pi^\deseq$.
 \end{remark}
 
 \begin{proposition}[\cite{TdF2000,llhandbook}]
  \label{prop:acc-jump}
  Let $R$ be a \jt \ps of $\mllu$. Then $R \models \ACC$ if and only~if there exists a \seqc proof $\pi$ of $\mllu$ such that $\pi \deseqjump R$.
 \end{proposition}
 
  In $\btenll$, we refine \Cref{thm:untyped-sequentiality} (by \Cref{rmk:underlying-jump-free}): when $R \models \ACCw$, $R$ has a sequentialization which can be desequentialized to a~\jt~\ps with jumps in a ``canonical''~position.
 
 \begin{theorem}[restate = sequentializationCanonicalJump, name = ]
 \label{thm:untyped-sequentialityBottomTens}
  Let $R$ be a \jf \ps of $\btenll$ s.t.~$R \models \ACCw$, let $m$ be a non-\er node of $R$. Then there exists a \seqc proof $\pi$ in $\btenll$ such that $\pi \deseqjump \canext{R}{m}$.
 \end{theorem}
 
 \begin{proof}
  We reason by induction on $\sharp \ar(R)$. Suppose $R$ has a terminal \er node $n$. Then $n$ is either a $\bot$ or $\parr$ node. Hence, there exists a \ps $R_1$ such that $R$ has the shape~in~\Cref{subfig:desequentialization-bottom}~or~\ref{subfig:desequentialization-par}.  Since $n$ is \er and $m$ is non-\er, we have $n \neq m$, so $m$ is a node of~$R_1$. It is easy to check that $R_1$ satisfies the induction hypothesis. Therefore, there exists~a~\seqc proof $\pi_1$ in $\btenll$ such that $\pi_1 \deseqjump \canext{R_1}{m}$. Let $\pi$ be the \seqc~proof~obtained from $\pi_1$ by applying a $\bot$ or $\parr$ rule introducing the type of the conclusion~of~$n$. Then~$\pi \deseqjump \canext{R}{m}$,~because:
  \begin{itemize}
  	\item
  	If $n$ is a $\bot$ node, by \Cref{def:canonical-jump}, we have $J_{\canext{R}{m}}(n) = m$;
  	\item
  	If $n$ is a $\parr$ node, we have $\wparr(R) = \wparr(R_1)$, and thus $J_{\canext{R}{m}} = J_{\canext{R_1}{m}}$.
  \end{itemize}
  
  We can now suppose that $R$ has no terminal \er node. By \Cref{itm:canonical-jump-terminal,itm:canonical-jump} of \Cref{rmk:canonical-jump}, $\w(R) = \wparr(R)$ and $\canext{R}{m} = \can{R}$. We assume that $R$ has a terminal non-\er $\parr$ node $n$. By \Cref{def:erasing}, there exists a non-\er node $m_1$ having a~premise~of~$n$~among~its~conclusions. As before, there is a \ps $R_1$ such that $R$ is as in \Cref{subfig:desequentialization-par} and, since $m_1 \in \ve(R_1)$, we can apply the induction hypothesis on $R_1, m_1$. Hence, there is a \seqc~proof~$\pi_1$~in~$\btenll$ s.t.~$\pi_1 \deseqjump \canext{R_1}{{m_1}}$. Let $\pi$ be the \seqc proof obtained from $\pi_1$ by applying a $\parr$ rule introducing the type of the conclusion of $n$. Then $\pi \deseqjump \can{R}$. Indeed,~let~$n_1 \in \w(R_1) = \w(R)$:
  \begin{itemize}
   \item
    If $n_1 \in \wparr(R_1)$, then we immediately get $J_{\can{R}}(n_1) = J_{\canext{R_1}{{m_1}}}(n_1)$;
   \item
    If $n_1 \in \we(R_1)$, then $n$ is the least non-\er node of $R$ such that $n_1 \prec n$, and therefore:
    \[
     J_{\can{R}}(n_1) = m_1 = J_{\canext{R_1}{{m_1}}}(n_1)
    \]
  \end{itemize}
  
  Now suppose $R$ has no terminal $\bot$ or $\parr$ node. If $R$ has no terminal $\otimes$ node,~$R$~is~a~terminal $\axpn$ or $\onepn$ node, and we are done. Otherwise, it is sufficient to follow the argument in the proof of \Cref{thm:untyped-sequentiality} and apply the induction hypothesis straightforwardly (see \Cref{subsec:equivalence-canonical-jumps}).
 \end{proof}
 
% \begin{remark}
%  \label{rmk:sequentialization}
%  Canonical jumps allow us to prove sequentialization for $\btenll$ (\Cref{cor:sequentialization}) without using \Cref{thm:untyped-sequentiality}. Indeed, consider a (\cf) \jf \ps $R$ of $\btenll$~such~that $R \models \ACCw$. By \Cref{rmk:cw-one}, \Cref{lemma:cw-forall-tensor} and \Cref{itm:wten-proof-structure} of \Cref{rmk:erasing}, there is a non-\er node $m$ of $R$. Then, by \Cref{cor:adding-jump}, $\canext{R}{m} \models \ACC$. Hence, by \Cref{prop:acc-jump}, there exists a \seqc proof $\pi$ of $\mllu$ s.t.~$\canext{R}{m} = \pi^\deseqjump$. Therefore, by \Cref{rmk:underlying-jump-free}:~$R = \ujf{\canext{R}{m}} = \ujf{\pi^\deseqjump} = \pi^\deseq$.
% \end{remark}

 In $\btenll$, the equivalence relation $\equivseqc$ between \seqc proofs (precisely defined in Figure 4 of~\cite{heijltjes2014no}) is ``easy'' to decide: by \Cref{cor:equivalenceFeasible}, $\pi_1 \equivseqc \pi_2$ holds if and only if $\pi_1$ and $\pi_2$ desequentialize to the same \jf \ps. This is in sharp contrast~with~the~full~fragment~$\mllu$, for which this decision problem is proven to be $\mathit{PSPACE}$-complete in~\cite{heijltjes2014no}.
 
 \begin{definition}[\cite{heijltjes2014no}]
  \label{def:rewiring}
  \emph{Rewiring} is defined on \jc \pss: $R \rewpns R'$ if $R'$ is obtained from $R$ by redirecting exactly one jump. $\equivpns$~is the smallest equivalence relation~which~contains~$\rewpns$.
%  \emph{Rewiring} is the binary relation $\rewpns$ on the set of \jc \pss defined as follows: $R \rewpns R'$ if $R'$ is obtained from $R$ by redirecting exactly one jump. \emph{Equivalence}~$\equivpns$~is the smallest equivalence relation containing $\rewpns$.
 \end{definition}
 
% The next result follows from \Cref{rmk:jump} and \Cref{lemma:adding-jump} (see \Cref{subsec:equivalence-canonical-jumps}).
 
 \begin{proposition}[restate = equivalenceCanonicalJump, name = ]
  \label{prop:equivalence-canonical-jump}
  For $R$ \jc \ps of $\btenll$, $m$ non-\er node of $R$: $R \equivpns \canext{R}{m}$.
 \end{proposition}
 
 \begin{corollary}[restate = equivalence, name = ]
  \label{cor:equivalence}
  Let $R_1, R_2$ be \jc \pss of $\btenll$. If $\ujf{R_1} = \ujf{R_2}$, then $R_1 \equivpns R_2$.
 \end{corollary}
 
 \begin{corollary}[restate = rulePermutationsRewiring, name = ]
  \label{cor:equivalenceFeasible}
  Let $\pi_1$, $\pi_2$ be \seqc proofs in $\btenll$, let $R_1$, $R_2$ be \jt \pss of $\btenll$ s.t.~$\pi_1 \deseqjump R_1$, $\pi_2 \deseqjump R_2$. By \Cref{prop:acc-jump}, $R_1$ and $R_2$~are~\jc,~and:
  \[
   \pi_1 \sim \pi_2 \quad \underset{1}{\Leftrightarrow} \quad R_1 \equivpns R_2 \quad \underset{2}{\Leftrightarrow} \quad \pi_1^\deseq = \pi_2^\deseq
  \]
  %$\pi_1 \sim \pi_2$ $\underset{1}{\Leftrightarrow}$ $R_1 \equivpns R_2$ $\underset{2}{\Leftrightarrow}$ $\pi_1^\deseq = \pi_2^\deseq$.
  %where $(\cdot)^\circ$ is desequentialization to \jf \pss (\Cref{fig:desequentialization}).
 \end{corollary}
 
 \begin{proof}
  (1) is Proposition~3.8 in~\cite{heijltjes2014no}, (2) is a consequence of \Cref{cor:equivalence} (\Cref{subsec:equivalence-canonical-jumps}).
 \end{proof}
 
 We finally consider a variant of the fragment $\btenll$ which is closed by linear negation.
 
 \begin{definition}
  \label{def:btenlls}
  The fragment $\btenlls$ of $\mllu$ is defined by the following grammar:
  \begin{align*}
   \widebar{A} \Coloneqq{} & A \mid A^\perp & A \Coloneqq{} & X \mid \widebar{A} \parr \widebar{A} \mid \widebar{A} \parr E \mid E \parr \widebar{A} \\
   \widebar{E} \Coloneqq{} & E \mid E^\perp & E \Coloneqq{} & \bot \mid E \parr E \\
   \intertext{\indent \textup{For the sake of readability, we present the possible shapes of a formula of~kind~$A^\perp$~or~$E^\perp$:}}
   A^\perp ={} & X \mid \widebar{A} \otimes \widebar{A} \mid \widebar{A} \otimes E^\perp \mid E^\perp \otimes \widebar{A} & E^\perp ={} & \one \mid E^\perp \otimes E^\perp
  \end{align*}
 \end{definition}
 
 \begin{remark}
  While containing $\mll$, the fragment $\btenlls$ is quite small: it does not even include formulas such as $\smash{\bot \parr \one}$ (i.e.~the linear implication $\smash{\one \multimap \one}$). The more general $\btenll$ contains this formula, but is not closed under linear negation. Observe that $\smash{\bot \parr \one}$ is polarized by the intuitionistic polarization, which motivates the study of $\imell$~(\Cref{sec:imell}).
 \end{remark}

 \section{Extensions to \texorpdfstring{$\mell$}{MELL}}
  \label{sec:mell}
   %In this section, we examine how our results extend to the more general setting of multiplicative exponential linear logic ($\mell$). We recall the grammar which~defines~the~formulae~of~$\mell$:
 
 Many of the results presented extend straightforwardly to multiplicative exponential linear logic ($\mell$); we discuss some of them in this section. Recall the grammar defining~$\mell$:
 \[
  A \Coloneqq X \mid \one \mid \bot \mid A \otimes A \mid A \parr A \mid{} !A \mid{} ?A
 \]
% Linear negation is extended by defining $(!A)^\perp \coloneq{} ?A^\perp$ and $(?A)^\perp \coloneq{} !A^\perp$. \Seqc is enriched with the rules of promotion, dereliction, contraction and weakening (\cite{girard1987linear}). These~are mirrored in \pss by the addition of exponential boxes and of $\drpn$, $\ctpn$ and $\wkpn$ nodes. Every $\drpn$ node has exactly one premise and one conclusion, whereas $\ctpn$ and $\wkpn$ nodes have the~same~geometric structure as $\parr$ and $\bot$ nodes respectively. Every property introduced by \Cref{def:acyclicity-connectivity-global} (and in particular $\ACCw$) is extended to these \pss by requiring that it is satisfied by the ``ground structure'' and, inductively, by the content of every box. $\wten$-\ps are similarly defined, as long as $\drpn$, $\ctpn$ and $\wkpn$ nodes are included in \Cref{def:erasing}, and \Cref{thm:untyped-sequentiality} holds in this setting. \Cutelim is extended, but the following statement (\cite{Pag-TdF2010}) holds~instead~of~\Cref{thm:cut-elim}:

%\begin{theorem}
%  The binary relation $\to$ is:
%  \begin{itemize}
%   \item
%    Confluent on $\{R : \text{$R$ \ps and $R \models \AC$}\}$;
%   \item
%    Strongly normalizing on $\{R : \text{$R$ \ps of $\mell$ and $R \models \AC$}\}$.
%  \end{itemize}
% \end{theorem}

 Linear negation is extended by defining $(!A)^\perp \coloneq{} ?A^\perp$ and $(?A)^\perp \coloneq{} !A^\perp$. \Seqc is enriched with the rules of promotion, dereliction, contraction and weakening (\cite{girard1987linear}). These~are mirrored in \pss by the addition of exponential boxes and of $\drpn$, $\ctpn$ and $\wkpn$ nodes. Every $\drpn$ node has exactly one premise and one conclusion, whereas $\ctpn$ and $\wkpn$ nodes have the~same~geometric structure as $\parr$ and $\bot$ nodes respectively. Every~property~introduced~by~\Cref{def:acyclicity-connectivity-global} (and in particular $\ACCw$) can be extended to these \pss by requiring that it is satisfied by the ``ground structure'' and, inductively, by the content of every box. \Cutelim is extended too, and \Cref{thm:cut-elim} holds for the set $\{R : \text{$R$ \ps and $R \models \AC$}\}$ (\cite{Pag-TdF2010}). The stability of $\ACCw$ with respect to \cutelim (\Cref{thm:stability-cut-elimination}) so as \Cref{prop:sequentiality-cut-elimination,prop:correctness-criterion} straightforwardly~extend to this broader setting. The main result of \Cref{sec:untyped-sequentiality-theorem} is extended as well: $\wten$-\ps are similarly defined, as long as $\drpn$, $\ctpn$ and $\wkpn$ nodes are included in~\Cref{def:erasing},~and~\Cref{thm:untyped-sequentiality}~holds. 

 Concerning \Cref{sec:btenll}, the fragment $\btenll$ can be enriched with the exponentials:

 \begin{definition}
  The fragment $\wntenll$ of $\mell$ is defined by the following grammar:
  \[
   A \Coloneqq X \mid \one \mid A \otimes A \mid A \parr A \mid A \parr E \mid E \parr A \mid{} !A \mid{} !E \qquad E \Coloneqq \bot \mid E \parr E \mid{} ?A \mid{} ?E
  \]
 \end{definition}

 A~variant~that~is~closed~under~linear~negation is obtained by adding $?$ formulae to the formulae of kind $E$ in \Cref{def:btenlls}. By \Cref{thm:untyped-sequentiality}, $\ACCw$ is a correctness~criterion~for~$\wntenll$:

 \begin{corollary}
  For any \ps $R$ of $\wntenll$: $R \models \ACCw$ $\Leftrightarrow$ $R$ is sequentializable in $\wntenll$.
 \end{corollary}

 The classical polarization of linear logic without the additives (\cite{laurent1999polarized,laurent2005polarized} for the full fragment) is a fragment of $\wntenll$ (positive formulae are of kind~$A$,~negative~formulae~are~of~kind~$E$):

 \begin{definition}
  The \emph{polarized} fragment $\llpol$ of $\mell$ is defined~by~the~following~grammar:
  \[
   P \Coloneqq \one \mid{} !X \mid{} !N \mid{} P \otimes P \qquad \text{(\emph{positive})} \qquad N \Coloneqq \bot \mid{} ?X \mid{} ?P \mid N \parr N \qquad \text{(\emph{negative})}
  \]
 \end{definition}
 
 Hence, $\ACCw$ is a correctness criterion for $\llpol$, and indeed in such a polarized framework $\ACCw$ is known to be equivalent to the criterion studied by Laurent (see~\cite{laurent1999polarized,laurent2002etude,laurent2003polarized}).
 
 \begin{corollary}
  \label{cor:sequentialization-polarized}
  Let $R$ be a \ps of $\llpol$. Then: $R \models \ACCw$ $\Leftrightarrow$ $R$ is~sequentializable~in~$\llpol$.
 \end{corollary}
 
 The call-by-name \lc can be embedded into \pss of $\llpol$ (see~\cite{laurent2003polarized}). We have an analogous result to \Cref{cor:sequentialization-polarized} for the fragment of $\llpol$ needed for the translation.
 
 \begin{proposition}
  Let $\lamlinl$ be the fragment of $\llpol$ defined by the following grammar:
  \[
   P \Coloneqq{} !X \mid{} !N \otimes P \qquad \text{(\emph{positive})} \qquad N \Coloneqq{} ?X \mid{} ?P \parr N \qquad \text{(\emph{negative})}
  \]
  For every \ps $R$ of $\lamlinl$, we have: $R \models \ACCw$ $\Leftrightarrow$ $R$ is sequentializable in $\lamlinl$.
 \end{proposition}
 
 As we did for $\btenll$ in \Cref{sec:btenll}, one can refine \Cref{thm:untyped-sequentiality} and prove the analogue of \Cref{thm:untyped-sequentialityBottomTens} for $\wntenll$ with ``the same'' canonical position for jumps. On the other hand, although some analogue of \Cref{cor:equivalenceFeasible} certainly holds for $\wntenll$,~a~precise~statement would require a definition of the equivalence relation between \seqc~proofs~in~$\mell$.
 
 \section{The intuitionistic polarization \texorpdfstring{(\cite{danos1990logique,regnier1992lambda,lamarcheIMELL})}{}}
  \label{sec:imell}
  % We now turn our attention to the intuitionistic polarization of linear logic (see \cite{danos1990logique,regnier1992lambda,lamarcheIMELL}).
 
 \begin{definition}
  \label{def:imell}
  The fragment $\imell$ of $\mell$ is defined by splitting the set of atomic formulae into two disjoint subsets, one with output atoms (denoted by $X$)~and~the~other~one with their duals (denoted by $X^\perp$) and by considering the following grammar:
  \begin{align*}
   O & \Coloneqq X \mid \one \mid O \otimes O \mid O \parr I \mid I \parr O \mid{} !O && \text{(\emph{outputs})} \\
   I & \Coloneqq X^\perp \mid \bot \mid I \parr I \mid I \otimes O \mid O \otimes I \mid{} ?I && \text{(\emph{inputs})}
  \end{align*}
 \end{definition}
 
% As explained in \cite{lamarcheIMELL}, having exactly one formula on the right side of the turnstile in intuitionistic $\linl$ with two-sided sequents is equivalent to having exactly one output formula.
 
 \begin{proposition}[\cite{lamarcheIMELL}]
  \label{prop:exactly-one-output}
  Every provable sequent of $\imell$ has exactly one~output~formula.
 \end{proposition}
 
 For \pss satisfying $\AC$, the necessary condition given by \Cref{prop:exactly-one-output} can be characterized geometrically: it is equivalent to $\ACCw$ (\Cref{prop:accw-output-conclusion}, proof in \Cref{subsec:connectivity-polarities}).
 
 \begin{definition}
  \label{def:polarities-arcs-nodes}
  Let $R$ be a \ps of $\imll$. An arc of $R$ is \emph{output} (resp.~\emph{input}) if its type is output (resp.~input). A node of $R$ is \emph{output} (resp.~\emph{input}) if it has exactly one~output~(resp.~input) conclusion. We denote by $\out(R)$ the set of output conclusions of $R$.% If~$\sharp \out(R) = 1$,~we~denote~by $n_\out$ the node of $R$ having among its conclusions the unique output conclusion of $R$.
 \end{definition}
 
 \begin{proposition}[restate = accwOutputConclusion, name = ]
  \label{prop:accw-output-conclusion}
  For $R$ \ps ($\cutpn$ allowed) of $\imell$ s.t.~$R \models \AC$: $R \models \ACCw$ $\Leftrightarrow$~$\sharp \out(R) = 1$.
 \end{proposition}
 
 The fragment of $\imell$ obtained by dropping the exponentials is denoted by $\imll$.
 
 \begin{remark}
  \label{rm:ACCwnoIMELL}
  By \Cref{thm:untyped-sequentiality}, $\ACCw$ is enough to sequentialize $\wten$-\pss of $\imll$. However, some pss of $\imll$ are not $\wten$-\pss, satisfy $\ACCw$ and are not sequentializable (\Cref{subfig:counterexample-imllu}).
%  Some pss of $\imll$ are not $\wten$-\pss. In \Cref{subfig:counterexample-imllu}, for instance, we have~a~\ps of $\imll$ which is not a $\wten$-\ps, satisfies $\ACCw$ and is not sequentializable in $\imll$. By \Cref{thm:untyped-sequentiality}, of course, $\ACCw$ is sufficient to sequentialize $\wten$-\pss of $\imll$.
  %Some pss of the fragment $\imll$ of $\imell$ (obtained by dropping the exponentials) are not $\wten$-\pss: the \ps of $\imll$ in \Cref{subfig:counterexample-imllu} satisfies $\ACCw$ but is not sequentializable. By \Cref{thm:untyped-sequentiality}, $\ACCw$ is enough to sequentialize $\wten$-\pss of $\imll$.
 \end{remark}
 
 Following the previous sections, to simplify the presentation we would consider $\imll$. We actually further restrict our analysis to the intuitionistic fragment $\icomll$ of $\comll$ (\cite{LincolnW94}), obtained from $\imll$ by dropping the propositional letters $X$ and $X^\perp$: we consider \seqc proofs with no \axsc rule and \pss with no $\axpn$ node. Inspired by the proof of \Cref{thm:untyped-sequentialityBottomTens}, we show that, for \pss of $\icomll$ satisfying $\ACCw$, there is a ``canonical''~position~for jumps (\Cref{def:canonical-jump-intuitionistic,fig:canonical-jumps}) allowing us to sequentialize (\Cref{thm:sequentialization-icomll}). In the same perspective of \Cref{sec:btenll}, a \ps $R$ comes with a partial jump function $J_R$, and \Cref{def:Rjumps,rmk:jump} are easily adapted to $\icomll$; we can also introduce the~notation~$R_n^m$~of~\Cref{def:canonical-jump}~for~$R$~\ps of $\icomll$, $\smash{n \in \w(R) \setminus \dom(J_R)}$ and $m \in \ve(R)$.

%We thus retrieve the variant of \pss we considered in \Cref{sec:btenll}~and, for simplicity, we reuse the same notations. There is no ambiguity for a \ps of both $\btenll$ and $\icomll$,~provided we choose $m$ as a terminal non-\er node in \Cref{def:canonical-jump}.

%The two following results are the intuitionistic counterpart of \Cref{lemma:erasing-btenll} and are immediate consequences of \Cref{def:polarities-arcs-nodes}.

 \begin{lemma}
  \label{lemma:output}
  Let $R$ be a \ps of $\icomll$ and let $n$ be an output node of $R$.
  \begin{enumerate}[(i)]
   \item
    Every output premise of $n$ is a conclusion of some $\onepn$, output $\otimes$ or output $\parr$ node of $R$;
   \item \label{itm:output-leaf}
    There exists a unique $\onepn$ or output $\otimes$ node of $R$, denoted by $\ta{n}$, such that $\ta{n} = n$ or $\ta{n} \prec n$ and,~for every node $m$ of $R$ such that $\ta{n} \prec m \prec n$, $m$ is an~output~$\parr$~node~of~$R$.
  \end{enumerate}
 \end{lemma}

 \begin{definition}
  \label{def:wparr-intuitionistic}
  Let $R$ be a \ps of $\icomll$. We split $\w(R)$ in two: $\wparr(R) \coloneq \{n : \text{$n \in \w(R)$}$ $\text{and there exists an output $\parr$ node $m$ of $R$ such that $n \prec m$}\}$ and $\smash{\wi(R) \coloneq \w(R) \setminus \wparr(R)}$.
 \end{definition}

 \begin{definition}
  \label{def:canonical-jump-intuitionistic}
  Let $R$ be a \ps of $\icomll$ s.t.~$\sharp \out(R) = 1$, $\smash{n \in \w(R) \setminus \dom(J_R)}$,~and~let~$m$~be:
  \begin{itemize}
   \item
    The least (\Cref{rmk:prec}) output $\parr$ node of $R$ such that~$n \prec m$, if $n \in \wparr(R)$;
   \item
    The node having the unique output conclusion of $R$ as its conclusion, if $n \in \wi(R)$.
  \end{itemize}
  We set $R_n \coloneq R_n^{\ta{m}}$ (see \Cref{lemma:output}-\ref{itm:output-leaf}) and $\smash{\can{R} \coloneq (\ujf{R}_{n_1} \cdots)_{n_k}}$, where $\w(R) = \{n_1, \dots, n_k\}$.
 \end{definition}

%We can characterize explicitly \seqc proofs $\pi$ such that $\pi \deseqjump \can{\pi^\deseq}$.

 \begin{theorem}[restate = sequentializationICOMLL, name = ]
  \label{thm:sequentialization-icomll}
  Let $R$ be a \jf \ps of $\icomll$ such that $R \models \ACCw$. Then~there~exists a \seqc proof $\pi$ in $\icomll$ such that $\pi \deseqjump \can{R}$.
 \end{theorem}
 
 \begin{fw}
  We believe that the analogue of \Cref{cor:equivalenceFeasible}~holds~for~$\icomll$:~two~\seqc proofs are equivalent if and only if their desequentialization is the same (\jf) \ps. This would entail that, for every provable\footnote{Notice, by the way, that the converse of \Cref{prop:exactly-one-output} holds for $\icomll$: a sequent with exactly one output formula is provable (the proof is straightforward, by induction on the complexity of the sequent).} sequent of $\icomll$, there exists a unique \seqc proof of it (up to the~equivalence),~because~there~exists~a~unique~\jf~\ps~with such a conclusion.
  
  The fragment $\icomll$ is indeed small, but recall that provability in $\comll$ (which is small too!) is $\mathit{NP}$-hard (\cite{LincolnW94}). Also, remember that the general aim of this work is ``to let connectivity talk'' and understand its relations with logic, no matter (at least at the beginning) the logical system concerned, provided it emerges naturally from our analysis. Finally, notice that an atomic $\axpn$ node with conclusions $X$ and $X^{\bot}$ can be thought as a jump from $\bot$ to $1$, and thus an $\imll$ \ps can be thought as an $\icomll$ \ps $R$ equipped with a (partial) jump function $J_R$. Therefore, finding a correctness correctness criterion~for $\imll$ comes down to finding a property $\Prp$ such that, for every \ps $R$ (equipped with a jump function $J_R$) of $\icomll$, $R \models \Prp$ if and only if there exist a \seqc proof $\pi$ in $\imll$ and a \jt \ps $R'$ of $\imll$ such that $\pi \deseqjump R'$ and $J_{R'}$ extends $J_R$.\footnote{Of course, in order for $\Prp$ to be a proper correctness criterion, $\Prp$ should also be checkable in polynomial~time and stable under \cutelim, but this is the same as usual.} \Cref{thm:sequentialization-icomll} can be seen as the solution of the base case (the \jf case) in this incremental approach to the problem. All of this corroborates our intuition that the correctness criterion~for~$\imll$~(and~thus~$\imell$)~of~\cite{lamarcheIMELL}~can~be~expressed~in terms of connected components (in~the~style~of~$\ACCw$).
 \end{fw}

 \bibliography{Bibliography}
 
 \appendix
 
 \crefalias{subsection}{appendix}

 \section{Technical appendix}
   \subsection{Complements of \texorpdfstring{\Cref{sec:ps}}{Section 2}}
 \label{subsec:cut-elimination-steps}
 \begin{definition}
  Let $R$ be a \ps. We say that a $\cutpn$ node is an \emph{\axcut} if one of its premises is a conclusion of an $\axpn$ node, a \emph{unit} (resp.~\emph{multiplicative}) \emph{cut} if its premises are the conclusions of a $\onepn$ (resp.~$\otimes$) node and of a $\bot$ (resp.~$\parr$) node. The~set~of~\pss~is~endowed~with~the~rewriting rules illustrated in \Cref{fig:cutelim}, called \emph{\cutelim steps}.\footnote{To be precise, in the \axcut elimination step it is required that there exists a unique \dpath between the $\axpn$ and $\cutpn$ nodes explicitly represented in \Cref{subfig:cutelim-axiom}: the one consisting only of the arc which is simultaneously a conclusion of the axiom node and a premise of the cut node (see~\cite{llhandbook}).}
  
  \begin{figure}
   \centering
   \begin{subfigure}{0.475\textwidth}
   	\centering
   	\includegraphics[scale=0.2]{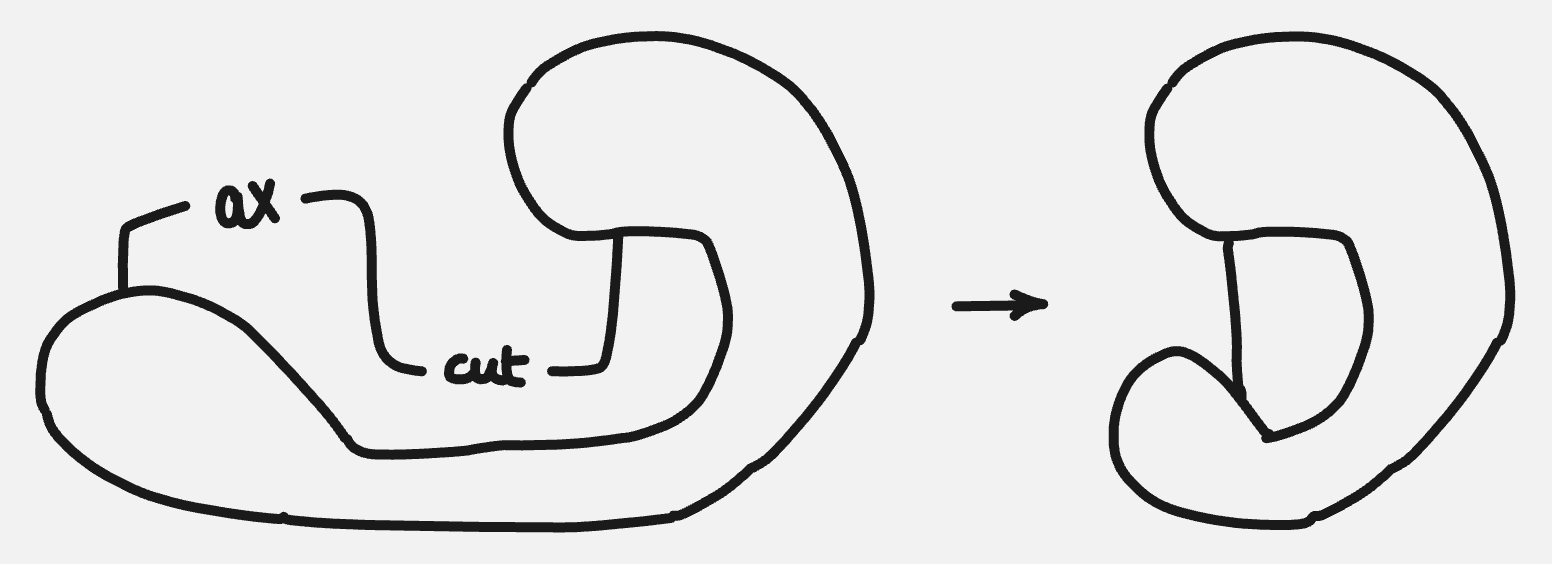}
   	\caption{\label{subfig:cutelim-axiom}\Axcut elimination step.}
   \end{subfigure}
   \begin{subfigure}{0.475\textwidth}
   	\centering
%   	\scalebox{\figscale}{
   	 \begin{tikzpicture}
	\begin{pgfonlayer}{nodelayer}
		\node [style=cutwLink] (2) at (1, -0.5) {};
		\node [style=none] (12) at (2, -0.375) {};
		\node [style=none] (13) at (2.25, -0.375) {};
		\node [style=none] (17) at (-1.25, 0.5) {};
		\node [style=none] (18) at (-0.1, 0.5) {};
		\node [style=none] (19) at (-1.25, 0) {};
		\node [style=none] (20) at (-0.1, 0) {};
		\node [style=none] (21) at (-1.25, 0.25) {};
		\node [style=none] (22) at (-0.1, 0.25) {};
		\node [style=none] (23) at (-0.1, -0.25) {};
		\node [style=none] (24) at (-1.25, -0.25) {};
		\node [style=none] (25) at (-1.05, -0.25) {};
		\node [style=none] (26) at (-0.3, -0.25) {};
		\node [style=none] (32) at (-0.675, -0.5) {\scriptsize$\hspace{0.45375ex}\dots$};
		\node [style=DOTagent] (33) at (-1.05, -0.75) {};
		\node [style=DOTagent] (34) at (-0.3, -0.75) {};
		\node [style=botwLink] (35) at (1.625, 0.125) {};
		\node [style=onewLink] (36) at (0.375, 0.125) {};
		\node [style=none] (37) at (2.5, 0.5) {};
		\node [style=none] (38) at (3.65, 0.5) {};
		\node [style=none] (39) at (2.5, 0) {};
		\node [style=none] (40) at (3.65, 0) {};
		\node [style=none] (41) at (2.5, 0.25) {};
		\node [style=none] (42) at (3.65, 0.25) {};
		\node [style=none] (43) at (3.65, -0.25) {};
		\node [style=none] (44) at (2.5, -0.25) {};
		\node [style=none] (45) at (2.7, -0.25) {};
		\node [style=none] (46) at (3.45, -0.25) {};
		\node [style=none] (47) at (3.075, -0.5) {\scriptsize$\hspace{0.45375ex}\dots$};
		\node [style=DOTagent] (48) at (2.7, -0.75) {};
		\node [style=DOTagent] (49) at (3.45, -0.75) {};
	\end{pgfonlayer}
	\begin{pgfonlayer}{edgelayer}
		\draw [style=reduction] (12.center) to (13.center);
		\draw [style=roundedCornerBlackFill] (19.center)
			 to (17.center)
			 to (18.center)
			 to (20.center);
		\draw [style=roundedCornerBlackFill] (21.center)
			 to (24.center)
			 to (23.center)
			 to (22.center);
		\draw [style=simpleB] (25.center) to (33);
		\draw [style=simpleB] (26.center) to (34);
		\draw [style=out1LEFT] (36) to (2);
		\draw [style=out1RIGHT] (35) to (2);
		\draw [style=roundedCornerBlackFill] (39.center)
			 to (37.center)
			 to (38.center)
			 to (40.center);
		\draw [style=roundedCornerBlackFill] (41.center)
			 to (44.center)
			 to (43.center)
			 to (42.center);
		\draw [style=simpleB] (45.center) to (48);
		\draw [style=simpleB] (46.center) to (49);
	\end{pgfonlayer}
\end{tikzpicture}
%   	}
   	\caption{\label{subfig:cutelim-unit}\Ucut elimination step.}
   \end{subfigure}
   \vskip 1em
   \begin{subfigure}{0.5\textwidth}
   	\centering
%   	\scalebox{\figscale}{
   	 \begin{tikzpicture}
	\begin{pgfonlayer}{nodelayer}
		\node [style=none] (0) at (-0.75, 2) {};
		\node [style=none] (1) at (2.25, 2) {};
		\node [style=none] (2) at (-0.75, 1.5) {};
		\node [style=none] (3) at (2.25, 1.5) {};
		\node [style=none] (5) at (-0.75, 1.75) {};
		\node [style=none] (6) at (2.25, 1.75) {};
		\node [style=none] (7) at (2.25, 1) {};
		\node [style=none] (9) at (-0.75, 1) {};
		\node [style=none] (10) at (-0.5, 1) {};
		\node [style=none] (14) at (0.5, 1) {};
		\node [style=parrwLink] (17) at (0, 0.75) {};
		\node [style=none] (19) at (1, 1) {};
		\node [style=none] (20) at (2, 1) {};
		\node [style=tensorwLink] (23) at (1.5, 0.75) {};
		\node [style=cutwLink] (24) at (0.75, 0.25) {};
		\node [style=whiteTest] (25) at (-0.5, 1.25) {\tiny$1$};
		\node [style=whiteTest] (26) at (0.5, 1.25) {\tiny$2$};
		\node [style=whiteTest] (27) at (1, 1.25) {\tiny$3$};
		\node [style=whiteTest] (28) at (2, 1.25) {\tiny$4$};
		\node [style=none] (29) at (3, 2) {};
		\node [style=none] (30) at (6, 2) {};
		\node [style=none] (31) at (3, 1.5) {};
		\node [style=none] (32) at (6, 1.5) {};
		\node [style=none] (33) at (3, 1.75) {};
		\node [style=none] (34) at (6, 1.75) {};
		\node [style=none] (35) at (6, 1) {};
		\node [style=none] (36) at (3, 1) {};
		\node [style=none] (37) at (3.25, 1) {};
		\node [style=none] (38) at (4.25, 1) {};
		\node [style=none] (40) at (4.75, 1) {};
		\node [style=none] (41) at (5.75, 1) {};
		\node [style=cutwLink] (43) at (4, 0.75) {};
		\node [style=whiteTest] (44) at (3.25, 1.25) {\tiny$1$};
		\node [style=whiteTest] (45) at (4.25, 1.25) {\tiny$2$};
		\node [style=whiteTest] (46) at (4.75, 1.25) {\tiny$3$};
		\node [style=whiteTest] (47) at (5.75, 1.25) {\tiny$4$};
		\node [style=cutwLink] (48) at (5, 0.75) {};
		\node [style=none] (49) at (2.5, 0.75) {};
		\node [style=none] (50) at (2.75, 0.75) {};
	\end{pgfonlayer}
	\begin{pgfonlayer}{edgelayer}
		\draw [style=roundedCornerBlackFill] (2.center)
			 to (0.center)
			 to (1.center)
			 to (3.center);
		\draw [style=roundedCornerBlackFill] (5.center)
			 to (9.center)
			 to (7.center)
			 to (6.center);
		\draw [style=roundedCornerBlackFill] (31.center)
			 to (29.center)
			 to (30.center)
			 to (32.center);
		\draw [style=roundedCornerBlackFill] (33.center)
			 to (36.center)
			 to (35.center)
			 to (34.center);
		\draw [style=Dout1LEFT] (10.center) to (17);
		\draw [style=Dout1LEFT] (17) to (24);
		\draw [style=Dout1LEFT] (19.center) to (23);
		\draw [style=Dout1LEFT] (37.center) to (43);
		\draw [style=Dout1LEFT] (38.center) to (48);
		\draw [style=Dout1RIGHT] (14.center) to (17);
		\draw [style=Dout1RIGHT] (20.center) to (23);
		\draw [style=Dout1RIGHT] (23) to (24);
		\draw [style=Dout1RIGHT] (40.center) to (43);
		\draw [style=Dout1RIGHT] (41.center) to (48);
		\draw [style=reduction] (49.center) to (50.center);
	\end{pgfonlayer}
\end{tikzpicture}
%   	}
   	\caption{\label{subfig:cutelim-multiplicative}\Mcut elimination step.}
   \end{subfigure}
   \caption{\label{fig:cutelim}The \axcut, \ucut and \mcut elimination steps.}
  \end{figure}
 \end{definition}
 
 \begin{remark}
  Not every $\cutpn$ node is reducible (see \Cref{fig:clash}). The $\cutpn$ nodes which are not reducible are called \emph{clashes} (see for example~\cite{Pag-TdF2010}).
  
  \begin{figure}
   \centering
   \includegraphics[scale=0.2]{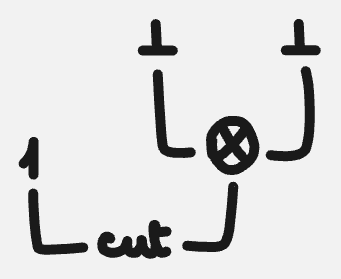}
   \caption{\label{fig:clash}Clash.}
  \end{figure}
 \end{remark}
 
 \begin{definition}
  \label{def:desequentialization}
  Let $\Sigma \coloneq A_1, \dots, A_k$ be a sequent of $\mllu$. For every \seqc proof $\pi$ in $\mllu$ with conclusion $\Sigma$, we define a \ps $\pi^\deseq$ of $\mllu$ with conclusions of types~$A_1, \dots, A_k$. The definition is by induction on $\pi$. Consider the last rule of $\pi$. Cases:
  \begin{itemize}
   \item
 	\axsc: Then $\pi$ is only made of an \axsc rule. Hence, we have $k = 2$ and $A_2 = A_1^\perp$. We~define~$\pi^\deseq$ as in \Cref{subfig:desequentialization-axiom} with $A \coloneq A_1$;
   \item
 	\cutsc: Then $\Sigma = \Gamma, \Delta$ and $\pi$ is obtained by applying a \cutsc rule to \seqc proofs $\pi_1$ and $\pi_2$ in $\mllu$ with conclusions $\Gamma, A$ and $A^\perp, \Delta$ respectively. By induction hypothesis on $\pi_1$ and on $\pi_2$, the \pss $\pi_1^\deseq$ and $\pi_2^\deseq$ of $\mllu$ are defined. We define $\pi^\deseq$ as in \Cref{subfig:desequentialization-cut}~with $R_i \coloneq \pi_i^\deseq$ for each $i \in \{1,2\}$;
   \item
    \exsc: Then $\Sigma = \Gamma, B, A, \Delta$ and $\pi$ is obtained by applying an \exsc rule to a \seqc proof $\pi_1$ in $\mllu$ with conclusion $\Gamma, A, B, \Delta$. By induction hypothesis on $\pi_1$, the \ps $\pi_1^\deseq$ of $\mllu$ is defined, with ordered conclusions $c_1, \dots, c_k, a, b, d_1, \dots, d_h$ such that $\Gamma$ is the sequence of the types of $c_1, \dots, c_k$, $A$ and $B$ are the types of $a$ and $b$ respectively, and $\Delta$ is the sequence of the types of $d_1, \dots, d_k$. We define $\pi^\deseq$ as the \ps of $\mllu$ obtained~from~$\pi_1^\deseq$~by just changing the total ordering of its conclusions as follows: $c_1, \dots, c_k, b, a, d_1, \dots, d_h$;
   \item
    $\one$: Then $\pi$ is only made of a $\one$ rule. We have $k = 1$, $A_1 = \one$: we define $\pi^\deseq$ as in \Cref{subfig:desequentialization-one};
   \item
    $\bot$: Then $\Sigma = \Gamma, \bot$ and $\pi$ is obtained by applying a $\bot$ rule to a \seqc proof $\pi_1$ in $\mllu$ with conclusion $\Gamma$. By induction hypothesis on $\pi_1$, the \ps $\pi_1^\deseq$ of $\mllu$ is defined. We define $\pi^\deseq$ as in \Cref{subfig:desequentialization-bottom} with $R_1 \coloneq \pi_1^\deseq$;
   \item
 	$\otimes$: Then $\smash{\Sigma = \Gamma, A \otimes B, \Delta}$ and $\pi$ is obtained by applying a $\otimes$ rule to \seqc proofs $\pi_1$ and $\pi_2$ in $\mllu$ with conclusions $\Gamma, A$ and $B, \Delta$ respectively. By induction hypothesis on $\pi_1$ and on $\pi_2$, the \pss $\pi_1^\deseq$ and $\pi_2^\deseq$ of $\mllu$ are defined. We define $\pi^\deseq$~as~in~\Cref{subfig:desequentialization-tensor}~with $R_i \coloneq \pi_i^\deseq$ for each $i \in \{1,2\}$;
   \item
 	$\parr$: Then $\smash{\Sigma = \Gamma, A \parr B}$ and $\pi$ is obtained by applying a $\parr$ rule to a \seqc proof $\pi_1$ in $\mllu$ with conclusion $\Gamma, A, B$. By induction hypothesis on $\pi_1$, the \ps $\pi_1^\deseq$ of $\mllu$~is defined. We define $\pi^\deseq$ as in \Cref{subfig:desequentialization-par} with $R_1 \coloneq \pi_1^\deseq$.
  \end{itemize}
 \end{definition}
 
 We recall the following well-known result, expressing the stability of $\AC$ with respect~to~the \cutelim steps (see~\cite{danos1990logique,llhandbook}).
 
 \begin{lemma}
  \label{lemma:acyclicity-preserved}
  Let $R,R'$ be \pss. If $R \tocut R'$, then: $R \models \AC$ $\Rightarrow$ $R' \models \AC$.
 \end{lemma}
 
 We now prove that also $\ACCw$ is stable under \cutelim.
 
% \begin{remark}
%  \label{rmk:correctness-criteria-boxes}
%  Let $R$ be a \ps and let $\Prp \in \{\C, \Cw\}$. Then $R \models \ACPrp$ if and only if $G_R \models \ACPrp$ and, for every $R_0 \sqsubset R$, we have $R_0 \models \ACPrp$.
% \end{remark}
 
 \stabilityCutElimination*
 
 \begin{proof}
  By Lemma~\ref{lemma:acyclicity-preserved}, we have $R' \models \AC$. Hence, it is sufficient to prove $R' \models \Cw$. Let $\varphi'$ be a switching of $R'$, and let $\varphi$ be any extension of $\varphi'$ to a switching of $R$. By Definitions~\ref{def:acyclicity-connectivity},~\ref{def:acyclicity-connectivity-global} and \Cref{itm:acyclicity-cc-number} of~\Cref{prop:acyclicity-cc}, from the hypothesis $R \models \ACCw$ and~from~$R' \models \AC$~we~deduce~that:
  \begin{gather}
   \sharp \w(R) + 1 \vphantom{{R'}^{\varphi'}} = \sharp \ve(R^\varphi) - \sharp \ar(R^\varphi) \vphantom{{R'}^{\varphi'}} \\
   \sharp \cc({R'}^{\varphi'}) = \sharp \ve({R'}^{\varphi'}) - \sharp \ar({R'}^{\varphi'})
  \end{gather}
  We distinguish the following cases:
  \begin{itemize}
   \item
    If $n$ is an \axcut (resp.~\mcut), then $R, R'$ have the shapes in Figure~\ref{subfig:cutelim-axiom} (resp.~Figure~\ref{subfig:cutelim-multiplicative}). Hence:
   \begin{gather*}
   	\sharp \ve({R'}^{\varphi'}) = \sharp \ve(R^\varphi) - 2 \\
   	\sharp \ar({R'}^{\varphi'}) = \sharp \ar(R^\varphi) - 2 \\
   	\sharp \w(R') = \sharp \w(R) \vphantom{{R'}^{\varphi'}}
   \end{gather*}
   Therefore:
   \[
    \sharp \cc(G_{R'}^{\varphi'}) \underset{2}{=} \sharp \ve(G_{R'}^{\varphi'}) - \sharp \ar(G_{R'}^{\varphi'}) = \sharp \ve(G_R^\varphi) - \sharp \ar(G_R^\varphi) \underset{1}{=} \sharp \w(G_R) + 1 = \sharp \w(G_{R'}) + 1
   \]
  \item
   If $n$ is a \ucut, then $R, R'$ have the shapes in Figure~\ref{subfig:cutelim-unit}. Hence:
   \begin{gather*}
   	\sharp \ve(G_{R'}^{\varphi'}) = \sharp \ve(G_R^\varphi) - 3 \\
   	\sharp \ar(G_{R'}^{\varphi'}) = \sharp \ar(G_R^\varphi) - 2 \\
   	\sharp \w(G_{R'}) = \sharp \w(G_R) - 1 \vphantom{G_{R'}^{\varphi'}}
   \end{gather*}
   Therefore:
   \[
   \sharp \cc(G_{R'}^{\varphi'}) \underset{2}{=} \sharp \ve(G_{R'}^{\varphi'}) - \sharp \ar(G_{R'}^{\varphi'}) = \sharp \ve(G_R^\varphi) - \sharp \ar(G_R^\varphi) - 1 \underset{1}{=} \sharp \w(G_R) = \sharp \w(G_{R'}) + 1 \qedhere
   \]
  \end{itemize}
 \end{proof}
 
 \begin{remark}
  \label{rmk:sequential-terminal}
  If $R$ is an $n$-\seq \ps, then $n$ is a \emph{terminal} node of $R$.
 \end{remark}
 
 \begin{remark}
  \label{rmk:exactly-one-premise}
  \begin{figure}
   \centering \hfill
   \begin{subfigure}{0.1125\textwidth}
 	\centering
 	%  	\scalebox{\figscale}{
 	\begin{tikzpicture}
	\begin{pgfonlayer}{nodelayer}
		\node [style=DOTagent] (1) at (-0.5, -0.5) {};
		\node [style=DOTagent] (2) at (0.5, -0.5) {};
		\node [style=axwLink] (6) at (0, 0.25) {};
		\node [style=none] (7) at (0.3, 0.325) {\tiny $n$};
	\end{pgfonlayer}
	\begin{pgfonlayer}{edgelayer}
		\draw [style=EoutLEFT] (6) to (1);
		\draw [style=EoutRIGHT] (6) to (2);
	\end{pgfonlayer}
\end{tikzpicture}
 	%  	}
 	\caption{$\ell = \axpn$}
   \end{subfigure}
   \hfill
   \begin{subfigure}{0.1875\textwidth}
   	\centering
   	%  	\scalebox{\figscale}{
   	\begin{tikzpicture}
	\begin{pgfonlayer}{nodelayer}
		\node [style=none] (0) at (-1.175, 1.75) {};
		\node [style=none] (1) at (-0.375, 1.75) {};
		\node [style=none] (2) at (-1.175, 1.25) {};
		\node [style=none] (3) at (-0.375, 1.25) {};
		\node [style=none] (4) at (-1.175, 1.5) {};
		\node [style=none] (5) at (-0.375, 1.5) {};
		\node [style=none] (6) at (-0.375, 1) {};
		\node [style=none] (7) at (-1.175, 1) {};
		\node [style=none] (10) at (-0.625, 1) {};
		\node [style=whiteTest] (15) at (-0.775, 1.375) {\scriptsize$R_1$};
		\node [style=none] (16) at (0.375, 1.75) {};
		\node [style=none] (17) at (1.175, 1.75) {};
		\node [style=none] (18) at (0.375, 1.25) {};
		\node [style=none] (19) at (1.175, 1.25) {};
		\node [style=none] (20) at (0.375, 1.5) {};
		\node [style=none] (21) at (1.175, 1.5) {};
		\node [style=none] (22) at (1.175, 1) {};
		\node [style=none] (23) at (0.375, 1) {};
		\node [style=none] (26) at (0.625, 1) {};
		\node [style=whiteTest] (31) at (0.775, 1.375) {\scriptsize$R_2$};
		\node [style=cutwLink] (32) at (0, 0.625) {};
		\node [style=none] (33) at (0.3, 0.3) {\tiny$n$};
		\node [style=none] (34) at (-0.95, 1) {};
		\node [style=none] (35) at (-0.95, 0.7) {};
		\node [style=none] (36) at (-0.9, 0.7) {};
		\node [style=none] (37) at (-0.9, 1) {};
		\node [style=DOTagent] (38) at (-0.925, 0.5) {};
		\node [style=none] (39) at (0.9, 1) {};
		\node [style=none] (40) at (0.9, 0.7) {};
		\node [style=none] (41) at (0.95, 0.7) {};
		\node [style=none] (42) at (0.95, 1) {};
		\node [style=DOTagent] (43) at (0.925, 0.5) {};
	\end{pgfonlayer}
	\begin{pgfonlayer}{edgelayer}
		\draw [style=roundedCornerBlackFill] (2.center)
			 to (0.center)
			 to (1.center)
			 to (3.center);
		\draw [style=roundedCornerBlackFill] (4.center)
			 to (7.center)
			 to (6.center)
			 to (5.center);
		\draw [style=roundedCornerBlackFill] (18.center)
			 to (16.center)
			 to (17.center)
			 to (19.center);
		\draw [style=roundedCornerBlackFill] (20.center)
			 to (23.center)
			 to (22.center)
			 to (21.center);
		\draw [style=out1LEFT] (10.center) to (32);
		\draw [style=out1RIGHT] (26.center) to (32);
		\draw [style=simpleB] (34.center) to (35.center);
		\draw [style=simpleB] (37.center) to (36.center);
		\draw [style=simpleB] (39.center) to (40.center);
		\draw [style=simpleB] (42.center) to (41.center);
	\end{pgfonlayer}
\end{tikzpicture}
 	%  	}
 	\caption{\label{subfig:sequential-cut}$\ell = \cutpn$}
   \end{subfigure}
   \hfill
   \begin{subfigure}{0.0875\textwidth}
   	\centering
   	%  	\scalebox{\figscale}{
    \begin{tikzpicture}
	\begin{pgfonlayer}{nodelayer}
		\node [style=DOTagent] (1) at (0, -0.75) {};
		\node [style=onewLink] (7) at (0, 0) {};
		\node [style=none] (8) at (0.25, 0.075) {\tiny$n$};
	\end{pgfonlayer}
	\begin{pgfonlayer}{edgelayer}
		\draw [style=simpleB] (7) to (1);
	\end{pgfonlayer}
\end{tikzpicture}
 	%  	}
 	\caption{$\ell = \onepn$}
   \end{subfigure}
   \hfill
   \begin{subfigure}{0.1\textwidth}
   	\centering
    %  	\scalebox{\figscale}{
 	\begin{tikzpicture}
	\begin{pgfonlayer}{nodelayer}
		\node [style=none] (0) at (-1.375, 1.25) {};
		\node [style=none] (1) at (-0.875, 1.25) {};
		\node [style=none] (2) at (-1.375, 0.75) {};
		\node [style=none] (3) at (-0.875, 0.75) {};
		\node [style=none] (4) at (-1.375, 1) {};
		\node [style=none] (5) at (-0.875, 1) {};
		\node [style=none] (6) at (-0.875, 0.5) {};
		\node [style=none] (7) at (-1.375, 0.5) {};
		\node [style=whiteTest] (15) at (-1.125, 0.875) {\scriptsize$R_1$};
		\node [style=botwLink] (16) at (-0.5, 0.75) {};
		\node [style=DOTagent] (17) at (-0.5, 0) {};
		\node [style=none] (18) at (-0.2, 0.825) {\tiny $n$};
		\node [style=none] (19) at (-1.15, 0.5) {};
		\node [style=none] (20) at (-1.15, 0.2) {};
		\node [style=none] (21) at (-1.1, 0.2) {};
		\node [style=none] (22) at (-1.1, 0.5) {};
		\node [style=DOTagent] (23) at (-1.125, 0) {};
	\end{pgfonlayer}
	\begin{pgfonlayer}{edgelayer}
		\draw [style=roundedCornerBlackFill] (2.center)
			 to (0.center)
			 to (1.center)
			 to (3.center);
		\draw [style=roundedCornerBlackFill] (4.center)
			 to (7.center)
			 to (6.center)
			 to (5.center);
		\draw [style=simpleB] (16) to (17);
		\draw [style=simpleB] (19.center) to (20.center);
		\draw [style=simpleB] (22.center) to (21.center);
	\end{pgfonlayer}
\end{tikzpicture}
 	%  	}
 	\caption{\label{subfig:sequential-bottom}$\ell = \bot$}
   \end{subfigure}
   \hfill
   \begin{subfigure}{0.1625\textwidth}
   	\centering
   	%  	\scalebox{\figscale}{
   	\begin{tikzpicture}
	\begin{pgfonlayer}{nodelayer}
		\node [style=none] (0) at (-1.05, 1.75) {};
		\node [style=none] (1) at (-0.25, 1.75) {};
		\node [style=none] (2) at (-1.05, 1.25) {};
		\node [style=none] (3) at (-0.25, 1.25) {};
		\node [style=none] (4) at (-1.05, 1.5) {};
		\node [style=none] (5) at (-0.25, 1.5) {};
		\node [style=none] (6) at (-0.25, 1) {};
		\node [style=none] (7) at (-1.05, 1) {};
		\node [style=none] (10) at (-0.5, 1) {};
		\node [style=whiteTest] (15) at (-0.65, 1.375) {\scriptsize$R_1$};
		\node [style=none] (16) at (0.25, 1.75) {};
		\node [style=none] (17) at (1.05, 1.75) {};
		\node [style=none] (18) at (0.25, 1.25) {};
		\node [style=none] (19) at (1.05, 1.25) {};
		\node [style=none] (20) at (0.25, 1.5) {};
		\node [style=none] (21) at (1.05, 1.5) {};
		\node [style=none] (22) at (1.05, 1) {};
		\node [style=none] (23) at (0.25, 1) {};
		\node [style=none] (26) at (0.5, 1) {};
		\node [style=whiteTest] (31) at (0.65, 1.375) {\scriptsize$R_2$};
		\node [style=tensorwLink] (32) at (0, 0.625) {};
		\node [style=none] (33) at (0.225, 0.275) {\tiny$n$};
		\node [style=DOTagent] (34) at (0, -0.125) {};
		\node [style=none] (35) at (-0.825, 1) {};
		\node [style=none] (36) at (-0.825, 0.7) {};
		\node [style=none] (37) at (-0.775, 0.7) {};
		\node [style=none] (38) at (-0.775, 1) {};
		\node [style=DOTagent] (39) at (-0.8, 0.5) {};
		\node [style=none] (40) at (0.775, 1) {};
		\node [style=none] (41) at (0.775, 0.7) {};
		\node [style=none] (42) at (0.825, 0.7) {};
		\node [style=none] (43) at (0.825, 1) {};
		\node [style=DOTagent] (44) at (0.8, 0.5) {};
	\end{pgfonlayer}
	\begin{pgfonlayer}{edgelayer}
		\draw [style=roundedCornerBlackFill] (2.center)
			 to (0.center)
			 to (1.center)
			 to (3.center);
		\draw [style=roundedCornerBlackFill] (4.center)
			 to (7.center)
			 to (6.center)
			 to (5.center);
		\draw [style=roundedCornerBlackFill] (18.center)
			 to (16.center)
			 to (17.center)
			 to (19.center);
		\draw [style=roundedCornerBlackFill] (20.center)
			 to (23.center)
			 to (22.center)
			 to (21.center);
		\draw [style=out1LEFT] (10.center) to (32);
		\draw [style=out1RIGHT] (26.center) to (32);
		\draw [style=simpleB] (32) to (34);
		\draw [style=simpleB] (35.center) to (36.center);
		\draw [style=simpleB] (38.center) to (37.center);
		\draw [style=simpleB] (40.center) to (41.center);
		\draw [style=simpleB] (43.center) to (42.center);
	\end{pgfonlayer}
\end{tikzpicture}
    %  	}
 	\caption{\label{subfig:sequential-tensor}$\ell = \otimes$}
   \end{subfigure}
   \hfill
   \begin{subfigure}{0.15\textwidth}
    \centering
 	%  	\scalebox{\figscale}{
 	 \begin{tikzpicture}
	\begin{pgfonlayer}{nodelayer}
		\node [style=none] (0) at (-1.55, 1.75) {};
		\node [style=none] (1) at (0.25, 1.75) {};
		\node [style=none] (2) at (-1.55, 1.25) {};
		\node [style=none] (3) at (0.25, 1.25) {};
		\node [style=none] (4) at (-1.55, 1.5) {};
		\node [style=none] (5) at (0.25, 1.5) {};
		\node [style=none] (6) at (0.25, 1) {};
		\node [style=none] (7) at (-1.55, 1) {};
		\node [style=none] (10) at (-1, 1) {};
		\node [style=whiteTest] (15) at (-0.65, 1.375) {\scriptsize$R_1$};
		\node [style=none] (16) at (0, 1) {};
		\node [style=parrwLink] (17) at (-0.5, 0.625) {};
		\node [style=DOTagent] (18) at (-0.5, -0.125) {};
		\node [style=none] (19) at (-0.275, 0.275) {\tiny $n$};
		\node [style=none] (20) at (-1.325, 1) {};
		\node [style=none] (21) at (-1.325, 0.7) {};
		\node [style=none] (22) at (-1.275, 0.7) {};
		\node [style=none] (23) at (-1.275, 1) {};
		\node [style=DOTagent] (24) at (-1.3, 0.5) {};
	\end{pgfonlayer}
	\begin{pgfonlayer}{edgelayer}
		\draw [style=roundedCornerBlackFill] (2.center)
			 to (0.center)
			 to (1.center)
			 to (3.center);
		\draw [style=roundedCornerBlackFill] (4.center)
			 to (7.center)
			 to (6.center)
			 to (5.center);
		\draw [style=out1LEFT] (10.center) to (17);
		\draw [style=out1RIGHT] (16.center) to (17);
		\draw [style=simpleB] (17) to (18);
		\draw [style=simpleB] (20.center) to (21.center);
		\draw [style=simpleB] (23.center) to (22.center);
	\end{pgfonlayer}
\end{tikzpicture}
 	%  	}
 	\caption{\label{subfig:sequential-par}$\ell = \parr$}
   \end{subfigure}
   \hfill \null
   \caption{\label{fig:sequential}Possible shapes of an $n$-\seq \ps, according to the label $\ell$ of $n$.}
  \end{figure}
 	
  Let $R$ be an $n$-\seq \ps with $n$ a $\cutpn$ (resp.~$\otimes$) node, and let $R_1, R_2$ be \seq \pss such that $R$ has the shape illustrated in \Cref{subfig:sequential-cut} (resp.~\Cref{subfig:sequential-tensor}). Then, for each $i \in \{1, 2\}$, exactly one premise of $n$ is in $R_i$.
 \end{remark}
 
 The proof of \Cref{prop:sequentiality-cut-elimination} relies on the following result.
 
 \begin{corollary}
  \label{cor:sequential}
  Let $R$ be an $n$-\seq \ps with $n$ a $\cutpn$ (resp.~$\otimes$) node, and let $R_1, R_2$ be \seq \pss such that $R$ has the shape in \Cref{subfig:sequential-cut} (resp.~\Cref{subfig:sequential-tensor}). Let $i \in \{1, 2\}$, and let $n_i$ be the node of $R_i$ having a premise of $n$ among its conclusions ($n_i$ is uniquely~determined~by \Cref{rmk:exactly-one-premise}). If, for every node $m$ of $R_i$, $R$ is not $m$-\seq, then $R_i$ is $n_i$-\seq.
 \end{corollary}
 
 \begin{proof}
  If $R_i$ is not $n_i$-\seq, then, since $R_i$ is \seq, there exists a node $m \neq n_i$ of $R_i$ such that $R_i$ is $m$-\seq. From $m \neq n_i$ and \Cref{rmk:sequential-terminal} it follows~that~$m$~is~a~terminal~node of $R$. Then $R$ is $m$-\seq, by \Cref{lemma:sequential}.
 	%It is enough to observe that, if $R_i$ is not $n_i$-\seq, then, since $R_i$ is \seq, there is a node $m \neq n_i$ of $G_{R_i}$ such that $R_i$ is $m$-\seq. From $m \neq n_i$ and \Cref{rmk:exactly-one-premise,rmk:sequential-terminal} we deduce that $m$ is a terminal node of $G_R$. Then $R$ is $m$-\seq, by \Cref{lemma:sequential}.
 \end{proof}
 
 The following key result is essentially a variant of Item~1 of Proposition~2 in~\cite{maieli2022proof}.
 
 \begin{lemma}
  \label{lemma:sequential}
  Let $R$ be an $n$-\seq \ps with $n$ a $\cutpn$, $\bot$, $\otimes$ or $\parr$ node. Let $k \coloneq 2$ if $n$ is a $\cutpn$ or $\otimes$ node, $k \coloneq 1$ otherwise, and let $R_1, \dots, R_k$ be \seq \pss such that $R$ has the shape in the sub-figure of \Cref{fig:sequential} determined by the label of $n$. Let $i \in \{1, \dots, k\}$, and let $m$ be a node of $R_i$ such that $m$ is a terminal node of $R$. Finally, suppose that either $n$ is not a $\bot$ node, or $m$ is not an $\axpn$ or $\onepn$ node, and that~either~$n$~is~not~a~$\parr$~node,~or~$m$~is~not~a~$\cutpn$~or~$\otimes$~node. If $R_i$ is $m$-\seq, then $R$ is $m$-\seq.
 \end{lemma}
 
 \begin{proof}
  \begin{figure}
   \centering
   \begin{subfigure}{0.9875\textwidth}
  	\centering
  	% 	\scalebox{\figscale}{
  	\begin{tikzpicture}
	\begin{pgfonlayer}{nodelayer}
		\node [style=none] (72) at (-1.65, 1.25) {};
		\node [style=none] (73) at (-0.85, 1.25) {};
		\node [style=none] (74) at (-1.65, 0.75) {};
		\node [style=none] (75) at (-0.85, 0.75) {};
		\node [style=none] (76) at (-1.65, 1) {};
		\node [style=none] (77) at (-0.85, 1) {};
		\node [style=none] (78) at (-0.85, 0.5) {};
		\node [style=none] (79) at (-1.65, 0.5) {};
		\node [style=none] (82) at (-1.1, 0.5) {};
		\node [style=whiteTest] (86) at (-1.25, 0.875) {\scriptsize$R_1'$};
		\node [style=none] (87) at (-0.1, 1.25) {};
		\node [style=none] (88) at (0.7, 1.25) {};
		\node [style=none] (89) at (-0.1, 0.75) {};
		\node [style=none] (90) at (0.7, 0.75) {};
		\node [style=none] (91) at (-0.1, 1) {};
		\node [style=none] (92) at (0.7, 1) {};
		\node [style=none] (93) at (0.7, 0.5) {};
		\node [style=none] (94) at (-0.1, 0.5) {};
		\node [style=none] (97) at (0.15, 0.5) {};
		\node [style=whiteTest] (101) at (0.3, 0.875) {\scriptsize$R_2'$};
		\node [style=cutwLink] (102) at (-0.475, 0.125) {};
		\node [style=none] (103) at (-0.175, -0.2) {\tiny$m$};
		\node [style=none] (189) at (1.075, 0.5) {$=$};
		\node [style=none] (190) at (-3.775, 1.25) {};
		\node [style=none] (191) at (-2.975, 1.25) {};
		\node [style=none] (192) at (-3.775, 0.75) {};
		\node [style=none] (193) at (-2.975, 0.75) {};
		\node [style=none] (194) at (-3.775, 1) {};
		\node [style=none] (195) at (-2.975, 1) {};
		\node [style=none] (196) at (-2.975, 0.5) {};
		\node [style=none] (197) at (-3.775, 0.5) {};
		\node [style=whiteTest] (198) at (-3.375, 0.875) {\scriptsize$R_1$};
		\node [style=botwLink] (199) at (-4.15, 0.75) {};
		\node [style=DOTagent] (200) at (-4.15, 0) {};
		\node [style=none] (201) at (-3.9, 0.4) {\tiny $n$};
		\node [style=none] (202) at (-3.25, 0.5) {};
		\node [style=none] (203) at (-3.25, 0.2) {};
		\node [style=none] (204) at (-3.2, 0.2) {};
		\node [style=none] (205) at (-3.2, 0.5) {};
		\node [style=DOTagent] (206) at (-3.225, 0) {};
		\node [style=whiteTest] (207) at (-3.525, 0.6) {\tiny$m$};
		\node [style=none] (208) at (-2.6, 0.5) {$=$};
		\node [style=botwLink] (209) at (-2.025, 0.75) {};
		\node [style=DOTagent] (210) at (-2.025, 0) {};
		\node [style=none] (211) at (-1.775, 0.4) {\tiny $n$};
		\node [style=none] (212) at (-1.425, 0.5) {};
		\node [style=none] (213) at (-1.425, 0.2) {};
		\node [style=none] (214) at (-1.375, 0.2) {};
		\node [style=none] (215) at (-1.375, 0.5) {};
		\node [style=DOTagent] (216) at (-1.4, 0) {};
		\node [style=none] (217) at (0.425, 0.5) {};
		\node [style=none] (218) at (0.425, 0.2) {};
		\node [style=none] (219) at (0.475, 0.2) {};
		\node [style=none] (220) at (0.475, 0.5) {};
		\node [style=DOTagent] (221) at (0.45, 0) {};
		\node [style=none] (222) at (1.45, 1.25) {};
		\node [style=none] (223) at (2.55, 1.25) {};
		\node [style=none] (224) at (1.45, 0.75) {};
		\node [style=none] (225) at (2.55, 0.75) {};
		\node [style=none] (226) at (1.45, 1) {};
		\node [style=none] (227) at (2.55, 1) {};
		\node [style=none] (228) at (2.55, 0.5) {};
		\node [style=none] (229) at (1.45, 0.5) {};
		\node [style=none] (230) at (2.3, 0.5) {};
		\node [style=whiteTest] (231) at (2, 0.875) {\scriptsize$R'\vphantom{R_1'}$};
		\node [style=none] (232) at (3.3, 1.25) {};
		\node [style=none] (233) at (4.1, 1.25) {};
		\node [style=none] (234) at (3.3, 0.75) {};
		\node [style=none] (235) at (4.1, 0.75) {};
		\node [style=none] (236) at (3.3, 1) {};
		\node [style=none] (237) at (4.1, 1) {};
		\node [style=none] (238) at (4.1, 0.5) {};
		\node [style=none] (239) at (3.3, 0.5) {};
		\node [style=none] (240) at (3.55, 0.5) {};
		\node [style=whiteTest] (241) at (3.7, 0.875) {\scriptsize$R_2'$};
		\node [style=cutwLink] (242) at (2.925, 0.125) {};
		\node [style=none] (243) at (3.225, -0.2) {\tiny$m$};
		\node [style=none] (244) at (1.975, 0.5) {};
		\node [style=none] (245) at (1.975, 0.2) {};
		\node [style=none] (246) at (2.025, 0.2) {};
		\node [style=none] (247) at (2.025, 0.5) {};
		\node [style=DOTagent] (248) at (2, 0) {};
		\node [style=none] (249) at (3.825, 0.5) {};
		\node [style=none] (250) at (3.825, 0.2) {};
		\node [style=none] (251) at (3.875, 0.2) {};
		\node [style=none] (252) at (3.875, 0.5) {};
		\node [style=DOTagent] (253) at (3.85, 0) {};
		\node [style=none] (254) at (1.7, 0.5) {};
		\node [style=DOTagent] (255) at (1.7, 0) {};
		\node [style=whiteTest] (256) at (1.7, 0.6) {\tiny$n$};
	\end{pgfonlayer}
	\begin{pgfonlayer}{edgelayer}
		\draw [style=roundedCornerBlackFill] (74.center)
			 to (72.center)
			 to (73.center)
			 to (75.center);
		\draw [style=roundedCornerBlackFill] (76.center)
			 to (79.center)
			 to (78.center)
			 to (77.center);
		\draw [style=roundedCornerBlackFill] (89.center)
			 to (87.center)
			 to (88.center)
			 to (90.center);
		\draw [style=roundedCornerBlackFill] (91.center)
			 to (94.center)
			 to (93.center)
			 to (92.center);
		\draw [style=out1LEFT] (82.center) to (102);
		\draw [style=out1RIGHT] (97.center) to (102);
		\draw [style=roundedCornerBlackFill] (192.center)
			 to (190.center)
			 to (191.center)
			 to (193.center);
		\draw [style=roundedCornerBlackFill] (194.center)
			 to (197.center)
			 to (196.center)
			 to (195.center);
		\draw [style=simpleB] (199) to (200);
		\draw [style=simpleB] (202.center) to (203.center);
		\draw [style=simpleB] (205.center) to (204.center);
		\draw [style=simpleB] (209) to (210);
		\draw [style=simpleB] (212.center) to (213.center);
		\draw [style=simpleB] (215.center) to (214.center);
		\draw [style=simpleB] (217.center) to (218.center);
		\draw [style=simpleB] (220.center) to (219.center);
		\draw [style=roundedCornerBlackFill] (224.center)
			 to (222.center)
			 to (223.center)
			 to (225.center);
		\draw [style=roundedCornerBlackFill] (226.center)
			 to (229.center)
			 to (228.center)
			 to (227.center);
		\draw [style=roundedCornerBlackFill] (234.center)
			 to (232.center)
			 to (233.center)
			 to (235.center);
		\draw [style=roundedCornerBlackFill] (236.center)
			 to (239.center)
			 to (238.center)
			 to (237.center);
		\draw [style=out1LEFT] (230.center) to (242);
		\draw [style=out1RIGHT] (240.center) to (242);
		\draw [style=simpleB] (244.center) to (245.center);
		\draw [style=simpleB] (247.center) to (246.center);
		\draw [style=simpleB] (249.center) to (250.center);
		\draw [style=simpleB] (252.center) to (251.center);
		\draw [style=simpleB] (254.center) to (255);
	\end{pgfonlayer}
\end{tikzpicture}
  	% 	}
   	\caption{$n$ is a $\bot$ node and $m$ is a $\cutpn$ node.}
   \end{subfigure} \vskip 1em
   \begin{subfigure}{0.9875\textwidth}
 	\centering
 	% 	\scalebox{\figscale}{
 	\begin{tikzpicture}
	\begin{pgfonlayer}{nodelayer}
		\node [style=none] (72) at (-3, 1.25) {};
		\node [style=none] (73) at (-0.9, 1.25) {};
		\node [style=none] (74) at (-3, 0.75) {};
		\node [style=none] (75) at (-0.9, 0.75) {};
		\node [style=none] (76) at (-3, 1) {};
		\node [style=none] (77) at (-0.9, 1) {};
		\node [style=none] (78) at (-0.9, 0.5) {};
		\node [style=none] (79) at (-3, 0.5) {};
		\node [style=none] (80) at (-2.775, 0.5) {};
		\node [style=none] (81) at (-1.15, 0.5) {};
		\node [style=DOTagent] (84) at (-1.15, 0) {};
		\node [style=whiteTest] (86) at (-1.95, 0.875) {\scriptsize$R_1$};
		\node [style=whiteTest] (105) at (-1.15, 0.6) {\tiny$m$};
		\node [style=none] (189) at (-0.525, 0.5) {$=$};
		\node [style=none] (191) at (-2.775, 0.2) {};
		\node [style=none] (192) at (-2.725, 0.2) {};
		\node [style=none] (193) at (-2.725, 0.5) {};
		\node [style=DOTagent] (194) at (-2.75, 0) {};
		\node [style=none] (195) at (-2.45, 0.5) {};
		\node [style=none] (196) at (-1.45, 0.5) {};
		\node [style=parrwLink] (197) at (-1.95, 0.125) {};
		\node [style=DOTagent] (198) at (-1.95, -0.625) {};
		\node [style=none] (199) at (-1.725, -0.225) {\tiny$n$};
		\node [style=none] (200) at (-0.15, 1.25) {};
		\node [style=none] (201) at (2.95, 1.25) {};
		\node [style=none] (202) at (-0.15, 0.75) {};
		\node [style=none] (203) at (2.95, 0.75) {};
		\node [style=none] (204) at (-0.15, 1) {};
		\node [style=none] (205) at (2.95, 1) {};
		\node [style=none] (206) at (2.95, 0.5) {};
		\node [style=none] (207) at (-0.15, 0.5) {};
		\node [style=none] (208) at (0.075, 0.5) {};
		\node [style=whiteTest] (211) at (1.4, 0.875) {\scriptsize$R_1'$};
		\node [style=none] (213) at (0.075, 0.2) {};
		\node [style=none] (214) at (0.125, 0.2) {};
		\node [style=none] (215) at (0.125, 0.5) {};
		\node [style=DOTagent] (216) at (0.1, 0) {};
		\node [style=none] (217) at (1.7, 0.5) {};
		\node [style=none] (218) at (2.7, 0.5) {};
		\node [style=parrwLink] (219) at (2.2, 0.125) {};
		\node [style=DOTagent] (220) at (2.2, -0.625) {};
		\node [style=none] (221) at (2.425, -0.225) {\tiny$m$};
		\node [style=none] (222) at (0.4, 0.5) {};
		\node [style=none] (223) at (1.4, 0.5) {};
		\node [style=parrwLink] (224) at (0.9, 0.125) {};
		\node [style=DOTagent] (225) at (0.9, -0.625) {};
		\node [style=none] (226) at (1.125, -0.225) {\tiny$n$};
		\node [style=none] (227) at (3.325, 0.5) {$=$};
		\node [style=none] (228) at (3.7, 1.25) {};
		\node [style=none] (229) at (5.8, 1.25) {};
		\node [style=none] (230) at (3.7, 0.75) {};
		\node [style=none] (231) at (5.8, 0.75) {};
		\node [style=none] (232) at (3.7, 1) {};
		\node [style=none] (233) at (5.8, 1) {};
		\node [style=none] (234) at (5.8, 0.5) {};
		\node [style=none] (235) at (3.7, 0.5) {};
		\node [style=none] (236) at (3.925, 0.5) {};
		\node [style=whiteTest] (237) at (4.75, 0.875) {\scriptsize$R'\vphantom{R_1'}$};
		\node [style=none] (238) at (3.925, 0.2) {};
		\node [style=none] (239) at (3.975, 0.2) {};
		\node [style=none] (240) at (3.975, 0.5) {};
		\node [style=DOTagent] (241) at (3.95, 0) {};
		\node [style=none] (247) at (4.55, 0.5) {};
		\node [style=none] (248) at (5.55, 0.5) {};
		\node [style=parrwLink] (249) at (5.05, 0.125) {};
		\node [style=DOTagent] (250) at (5.05, -0.625) {};
		\node [style=none] (251) at (5.275, -0.225) {\tiny$m$};
		\node [style=none] (252) at (4.25, 0.5) {};
		\node [style=DOTagent] (253) at (4.25, 0) {};
		\node [style=whiteTest] (254) at (4.25, 0.6) {\tiny$n$};
	\end{pgfonlayer}
	\begin{pgfonlayer}{edgelayer}
		\draw [style=roundedCornerBlackFill] (74.center)
			 to (72.center)
			 to (73.center)
			 to (75.center);
		\draw [style=roundedCornerBlackFill] (76.center)
			 to (79.center)
			 to (78.center)
			 to (77.center);
		\draw [style=simpleB] (81.center) to (84);
		\draw [style=simpleB] (80.center) to (191.center);
		\draw [style=simpleB] (193.center) to (192.center);
		\draw [style=out1LEFT] (195.center) to (197);
		\draw [style=out1RIGHT] (196.center) to (197);
		\draw [style=simpleB] (197) to (198);
		\draw [style=roundedCornerBlackFill] (202.center)
			 to (200.center)
			 to (201.center)
			 to (203.center);
		\draw [style=roundedCornerBlackFill] (204.center)
			 to (207.center)
			 to (206.center)
			 to (205.center);
		\draw [style=simpleB] (208.center) to (213.center);
		\draw [style=simpleB] (215.center) to (214.center);
		\draw [style=out1LEFT] (217.center) to (219);
		\draw [style=out1RIGHT] (218.center) to (219);
		\draw [style=simpleB] (219) to (220);
		\draw [style=out1LEFT] (222.center) to (224);
		\draw [style=out1RIGHT] (223.center) to (224);
		\draw [style=simpleB] (224) to (225);
		\draw [style=roundedCornerBlackFill] (230.center)
			 to (228.center)
			 to (229.center)
			 to (231.center);
		\draw [style=roundedCornerBlackFill] (232.center)
			 to (235.center)
			 to (234.center)
			 to (233.center);
		\draw [style=simpleB] (236.center) to (238.center);
		\draw [style=simpleB] (240.center) to (239.center);
		\draw [style=out1LEFT] (247.center) to (249);
		\draw [style=out1RIGHT] (248.center) to (249);
		\draw [style=simpleB] (249) to (250);
		\draw [style=simpleB] (252.center) to (253);
	\end{pgfonlayer}
\end{tikzpicture}
 	% 	}
 	\caption{$n$ and $m$ are $\parr$ nodes.}
   \end{subfigure} \vskip 1em
   \begin{subfigure}{0.9875\textwidth}
   	\centering
   	% 	\scalebox{\figscale}{
    \begin{tikzpicture}
	\begin{pgfonlayer}{nodelayer}
		\node [style=none] (72) at (-1.575, 1.25) {};
		\node [style=none] (73) at (-0.475, 1.25) {};
		\node [style=none] (74) at (-1.575, 0.75) {};
		\node [style=none] (75) at (-0.475, 0.75) {};
		\node [style=none] (76) at (-1.575, 1) {};
		\node [style=none] (77) at (-0.475, 1) {};
		\node [style=none] (78) at (-0.475, 0.5) {};
		\node [style=none] (79) at (-1.575, 0.5) {};
		\node [style=none] (82) at (-0.725, 0.5) {};
		\node [style=whiteTest] (86) at (-1.025, 0.875) {\scriptsize$R_1$};
		\node [style=none] (87) at (0.275, 1.25) {};
		\node [style=none] (88) at (1.125, 1.25) {};
		\node [style=none] (89) at (0.275, 0.75) {};
		\node [style=none] (90) at (1.125, 0.75) {};
		\node [style=none] (91) at (0.275, 1) {};
		\node [style=none] (92) at (1.125, 1) {};
		\node [style=none] (93) at (1.125, 0.5) {};
		\node [style=none] (94) at (0.275, 0.5) {};
		\node [style=none] (97) at (0.525, 0.5) {};
		\node [style=whiteTest] (101) at (0.7, 0.875) {\scriptsize$R_2$};
		\node [style=cutwLink] (102) at (-0.1, 0.125) {};
		\node [style=none] (103) at (0.2, -0.2) {\tiny$n$};
		\node [style=whiteTest] (105) at (-1.325, 0.6) {\tiny$m$};
		\node [style=none] (106) at (3.475, 1.25) {};
		\node [style=none] (107) at (4.575, 1.25) {};
		\node [style=none] (108) at (3.475, 0.75) {};
		\node [style=none] (109) at (4.575, 0.75) {};
		\node [style=none] (110) at (3.475, 1) {};
		\node [style=none] (111) at (4.575, 1) {};
		\node [style=none] (112) at (4.575, 0.5) {};
		\node [style=none] (113) at (3.475, 0.5) {};
		\node [style=none] (116) at (4.325, 0.5) {};
		\node [style=whiteTest] (120) at (4.025, 0.875) {\scriptsize$R_1'$};
		\node [style=none] (121) at (5.325, 1.25) {};
		\node [style=none] (122) at (6.175, 1.25) {};
		\node [style=none] (123) at (5.325, 0.75) {};
		\node [style=none] (124) at (6.175, 0.75) {};
		\node [style=none] (125) at (5.325, 1) {};
		\node [style=none] (126) at (6.175, 1) {};
		\node [style=none] (127) at (6.175, 0.5) {};
		\node [style=none] (128) at (5.325, 0.5) {};
		\node [style=none] (131) at (5.575, 0.5) {};
		\node [style=whiteTest] (135) at (5.75, 0.875) {\scriptsize$R_2$};
		\node [style=cutwLink] (136) at (4.95, 0.125) {};
		\node [style=none] (137) at (5.25, -0.2) {\tiny$n$};
		\node [style=none] (138) at (1.875, 1.25) {};
		\node [style=none] (139) at (2.725, 1.25) {};
		\node [style=none] (140) at (1.875, 0.75) {};
		\node [style=none] (141) at (2.725, 0.75) {};
		\node [style=none] (142) at (1.875, 1) {};
		\node [style=none] (143) at (2.725, 1) {};
		\node [style=none] (144) at (2.725, 0.5) {};
		\node [style=none] (145) at (1.875, 0.5) {};
		\node [style=none] (148) at (2.475, 0.5) {};
		\node [style=whiteTest] (152) at (2.3, 0.875) {\scriptsize$R_2'$};
		\node [style=cutwLink] (153) at (3.1, 0.125) {};
		\node [style=none] (154) at (3.725, 0.5) {};
		\node [style=none] (156) at (3.4, -0.2) {\tiny$m$};
		\node [style=none] (189) at (1.5, 0.5) {$=$};
		\node [style=none] (190) at (6.55, 0.5) {$=$};
		\node [style=none] (191) at (-1.05, 0.5) {};
		\node [style=none] (192) at (-1.05, 0.2) {};
		\node [style=none] (193) at (-1, 0.2) {};
		\node [style=none] (194) at (-1, 0.5) {};
		\node [style=DOTagent] (195) at (-1.025, 0) {};
		\node [style=none] (196) at (0.8, 0.5) {};
		\node [style=none] (197) at (0.8, 0.2) {};
		\node [style=none] (198) at (0.85, 0.2) {};
		\node [style=none] (199) at (0.85, 0.5) {};
		\node [style=DOTagent] (200) at (0.825, 0) {};
		\node [style=none] (201) at (2.15, 0.5) {};
		\node [style=none] (202) at (2.15, 0.2) {};
		\node [style=none] (203) at (2.2, 0.2) {};
		\node [style=none] (204) at (2.2, 0.5) {};
		\node [style=DOTagent] (205) at (2.175, 0) {};
		\node [style=none] (206) at (4, 0.5) {};
		\node [style=none] (207) at (4, 0.2) {};
		\node [style=none] (208) at (4.05, 0.2) {};
		\node [style=none] (209) at (4.05, 0.5) {};
		\node [style=DOTagent] (210) at (4.025, 0) {};
		\node [style=none] (211) at (5.85, 0.5) {};
		\node [style=none] (212) at (5.85, 0.2) {};
		\node [style=none] (213) at (5.9, 0.2) {};
		\node [style=none] (214) at (5.9, 0.5) {};
		\node [style=DOTagent] (215) at (5.875, 0) {};
		\node [style=none] (216) at (6.925, 1.25) {};
		\node [style=none] (217) at (7.725, 1.25) {};
		\node [style=none] (218) at (6.925, 0.75) {};
		\node [style=none] (219) at (7.725, 0.75) {};
		\node [style=none] (220) at (6.925, 1) {};
		\node [style=none] (221) at (7.725, 1) {};
		\node [style=none] (222) at (7.725, 0.5) {};
		\node [style=none] (223) at (6.925, 0.5) {};
		\node [style=none] (224) at (7.475, 0.5) {};
		\node [style=whiteTest] (225) at (7.325, 0.875) {\scriptsize$R_2'$};
		\node [style=none] (226) at (8.475, 1.25) {};
		\node [style=none] (227) at (9.575, 1.25) {};
		\node [style=none] (228) at (8.475, 0.75) {};
		\node [style=none] (229) at (9.575, 0.75) {};
		\node [style=none] (230) at (8.475, 1) {};
		\node [style=none] (231) at (9.575, 1) {};
		\node [style=none] (232) at (9.575, 0.5) {};
		\node [style=none] (233) at (8.475, 0.5) {};
		\node [style=none] (234) at (8.725, 0.5) {};
		\node [style=whiteTest] (235) at (9.025, 0.875) {\scriptsize$R'\vphantom{R_1'}$};
		\node [style=cutwLink] (236) at (8.1, 0.125) {};
		\node [style=none] (237) at (8.4, -0.2) {\tiny$m$};
		\node [style=whiteTest] (238) at (9.325, 0.6) {\tiny$n$};
		\node [style=none] (239) at (7.15, 0.5) {};
		\node [style=none] (240) at (7.15, 0.2) {};
		\node [style=none] (241) at (7.2, 0.2) {};
		\node [style=none] (242) at (7.2, 0.5) {};
		\node [style=DOTagent] (243) at (7.175, 0) {};
		\node [style=none] (244) at (9, 0.5) {};
		\node [style=none] (245) at (9, 0.2) {};
		\node [style=none] (246) at (9.05, 0.2) {};
		\node [style=none] (247) at (9.05, 0.5) {};
		\node [style=DOTagent] (248) at (9.025, 0) {};
	\end{pgfonlayer}
	\begin{pgfonlayer}{edgelayer}
		\draw [style=roundedCornerBlackFill] (74.center)
			 to (72.center)
			 to (73.center)
			 to (75.center);
		\draw [style=roundedCornerBlackFill] (76.center)
			 to (79.center)
			 to (78.center)
			 to (77.center);
		\draw [style=roundedCornerBlackFill] (89.center)
			 to (87.center)
			 to (88.center)
			 to (90.center);
		\draw [style=roundedCornerBlackFill] (91.center)
			 to (94.center)
			 to (93.center)
			 to (92.center);
		\draw [style=out1LEFT] (82.center) to (102);
		\draw [style=out1RIGHT] (97.center) to (102);
		\draw [style=roundedCornerBlackFill] (108.center)
			 to (106.center)
			 to (107.center)
			 to (109.center);
		\draw [style=roundedCornerBlackFill] (110.center)
			 to (113.center)
			 to (112.center)
			 to (111.center);
		\draw [style=roundedCornerBlackFill] (123.center)
			 to (121.center)
			 to (122.center)
			 to (124.center);
		\draw [style=roundedCornerBlackFill] (125.center)
			 to (128.center)
			 to (127.center)
			 to (126.center);
		\draw [style=out1LEFT] (116.center) to (136);
		\draw [style=out1RIGHT] (131.center) to (136);
		\draw [style=roundedCornerBlackFill] (140.center)
			 to (138.center)
			 to (139.center)
			 to (141.center);
		\draw [style=roundedCornerBlackFill] (142.center)
			 to (145.center)
			 to (144.center)
			 to (143.center);
		\draw [style=out1LEFT] (148.center) to (153);
		\draw [style=out1RIGHT] (154.center) to (153);
		\draw [style=simpleB] (191.center) to (192.center);
		\draw [style=simpleB] (194.center) to (193.center);
		\draw [style=simpleB] (196.center) to (197.center);
		\draw [style=simpleB] (199.center) to (198.center);
		\draw [style=simpleB] (201.center) to (202.center);
		\draw [style=simpleB] (204.center) to (203.center);
		\draw [style=simpleB] (206.center) to (207.center);
		\draw [style=simpleB] (209.center) to (208.center);
		\draw [style=simpleB] (211.center) to (212.center);
		\draw [style=simpleB] (214.center) to (213.center);
		\draw [style=roundedCornerBlackFill] (218.center)
			 to (216.center)
			 to (217.center)
			 to (219.center);
		\draw [style=roundedCornerBlackFill] (220.center)
			 to (223.center)
			 to (222.center)
			 to (221.center);
		\draw [style=roundedCornerBlackFill] (228.center)
			 to (226.center)
			 to (227.center)
			 to (229.center);
		\draw [style=roundedCornerBlackFill] (230.center)
			 to (233.center)
			 to (232.center)
			 to (231.center);
		\draw [style=out1LEFT] (224.center) to (236);
		\draw [style=out1RIGHT] (234.center) to (236);
		\draw [style=simpleB] (239.center) to (240.center);
		\draw [style=simpleB] (242.center) to (241.center);
		\draw [style=simpleB] (244.center) to (245.center);
		\draw [style=simpleB] (247.center) to (246.center);
	\end{pgfonlayer}
\end{tikzpicture}
 	% 	}
 	\caption{$n$ and $m$ are $\cutpn$ nodes.}
   \end{subfigure}
   \caption{\label{fig:sequential-proof}A few particular and interesting cases of the proof of \Cref{lemma:sequential}.}
  \end{figure}
 	
  We can suppose $i \coloneq 1$ without loss of generality. Since $m$ is a terminal node of~$R$,~$m$ has no premise of $n$ among its conclusions.
 	
  We claim that $m$ is not an $\axpn$ or $\onepn$ node. If $n$ is a $\bot$ node, then this is known by hypothesis. Otherwise, $n$ is a $\cutpn$, $\otimes$ or $\parr$ node. Then $R_1$ contains not only $m$, but also a node having a premise of $n$ among its conclusions. Since $R_1$ is $m$-\seq, $m$ cannot be an $\axpn$ or $\onepn$ node.
 	
  We have then established that $m$ is a $\cutpn$, $\bot$, $\otimes$ or $\parr$ node. Let $h \coloneq 2$ if $m$ is a $\cutpn$ or $\otimes$ node, $h \coloneq 1$ otherwise. Since $R_1$ is $m$-\seq, there exist \seq \pss $R_1', \dots R_h'$ such that $R_1$ has the shape in the sub-figure of \Cref{fig:sequential} (in the picture, $R_j'$ replaces $R_j$ for every $j \in \{1, \dots, h\}$, and $m$ replaces $n$) determined by the label of $m$. Since $m$ has no premise of $n$ among its conclusions, for every premise $a$ of $n$ in~$R_1$,~there~exists~$j \in \{1, \dots, h\}$~such~that~$a$~is in $R_j'$.
 	
  We claim that we can swap these quantifiers: there exists $j \in \{1, \dots, h\}$ such that, for every premise $a$ of $n$ in $R_1$, $a$ is in $R_j'$. This is trivial if $n$ is a $\bot$ node, or if $h = 1$. We then assume that $n$ is a $\cutpn$, $\otimes$ or $\parr$ node, and that $h = 2$. Then $m$ is a $\cutpn$ or $\otimes$ node. By hypothesis, $n$ is not a $\parr$ node. In other words, $n$ is a $\cutpn$ or $\otimes$ node. By~\Cref{rmk:exactly-one-premise},~exactly~one~premise~$a$~of~$n$~is~in $R_1$,~and~we~already~know~that~there~exists~$j \in \{1, \dots, h\}$~such~that~$a$~is~in $R_j'$.
 	
  We can now suppose, without loss of generality, that every premise of $n$ in $R_1$ is in $R_1'$. Let $R'$ be the \ps obtained from $R$ by replacing $R_1$ with $R_1'$. Since $R_1', R_2, \dots, R_k$ are \seq, we obtain that $R'$ is $n$-\seq. Finally, since~$R', R_2', \dots, R_h'$~are~\seq,~we~can conclude that $R$ is $m$-\seq (see \Cref{fig:sequential-proof}).
 \end{proof}
 
 Consider the $n$-\seq \pss in \Cref{fig:sequential-uniqueness}.
 
 \begin{figure}
  \centering
  \begin{subfigure}{0.2\textwidth}
   \centering
   %   \scalebox{\figscale}{
   \begin{tikzpicture}
	\begin{pgfonlayer}{nodelayer}
		\node [style=onewLink] (10) at (-0.75, 1) {};
		\node [style=botwLink] (24) at (-1.5, 1) {};
		\node [style=onewLink] (26) at (0.5, 1) {};
		\node [style=DOTagent] (30) at (-1.5, 0.25) {};
		\node [style=cutwLink] (32) at (-0.125, 0.5) {};
		\node [style=none] (33) at (0.175, 0.175) {\tiny$n$};
		\node [style=none] (35) at (-1.2, 1.075) {\tiny$m$};
	\end{pgfonlayer}
	\begin{pgfonlayer}{edgelayer}
		\draw [style=simpleB] (24) to (30);
		\draw [style=out1LEFT] (10) to (32);
		\draw [style=out1RIGHT] (26) to (32);
	\end{pgfonlayer}
\end{tikzpicture}
   %   }
   \caption{$R_\cutpn$.}
  \end{subfigure}
  \begin{subfigure}{0.2\textwidth}
   \centering
   %   \scalebox{\figscale}{
   \begin{tikzpicture}
	\begin{pgfonlayer}{nodelayer}
		\node [style=botwLink] (24) at (1.25, 1) {};
		\node [style=onewLink] (26) at (0.5, 1) {};
		\node [style=DOTagent] (30) at (1.25, 0.25) {};
		\node [style=DOTagent] (31) at (0.5, 0.25) {};
		\node [style=none] (32) at (0.75, 1.075) {\tiny$m$};
		\node [style=none] (33) at (1.55, 1.075) {\tiny$n$};
	\end{pgfonlayer}
	\begin{pgfonlayer}{edgelayer}
		\draw [style=simpleB, in=90, out=-90] (24) to (30);
		\draw [style=simpleB] (26) to (31);
	\end{pgfonlayer}
\end{tikzpicture}
   %   }
   \caption{$R_\rob$.}
  \end{subfigure}
  \begin{subfigure}{0.2\textwidth}
   \centering
   %   \scalebox{\figscale}{
   \begin{tikzpicture}
	\begin{pgfonlayer}{nodelayer}
		\node [style=DOTagent] (1) at (-0.5, -0.75) {};
		\node [style=tensorwLink] (7) at (-0.5, 0) {};
		\node [style=parrwLink] (8) at (0.5, 0) {};
		\node [style=axwLink] (9) at (-0.75, 0.625) {};
		\node [style=axwLink] (10) at (0.75, 0.625) {};
		\node [style=DOTagent] (11) at (0.5, -0.75) {};
		\node [style=none] (12) at (0.725, -0.35) {\tiny$n$};
		\node [style=none] (13) at (-0.275, -0.35) {\tiny$m$};
	\end{pgfonlayer}
	\begin{pgfonlayer}{edgelayer}
		\draw [style=simpleB] (7) to (1);
		\draw [style=DoACRIGHT] (9) to (8);
		\draw [style=DoACLEFT] (10) to (7);
		\draw [style=DoAC1LEFT] (9) to (7);
		\draw [style=DoAC1RIGHT] (10) to (8);
		\draw [style=simpleB] (8) to (11);
	\end{pgfonlayer}
\end{tikzpicture}
   %   }
   \caption{$R_\rtp$.}
  \end{subfigure}
  \caption{\label{fig:sequential-uniqueness}$n$-\seq \pss.}
 \end{figure}
 
 \begin{example}
  We can apply \Cref{lemma:sequential} to $R_\cutpn$. In this case, we have that $n$ is a~$\cutpn$~node and $m$ is a $\bot$ node. \Cref{lemma:sequential} entails that $R_\cutpn$ is $m$-\seq.
 \end{example}
 
 \begin{remark}
  The \pss $R_\rob$ and $R_\rtp$ are not $m$-\seq. These crucial examples (and obvious variants) show that \Cref{lemma:sequential} cannot be extended to the cases in which~$n$~is~a~$\bot$~node and $m$ is an $\axpn$ or $\onepn$ node, or $n$ is a $\parr$ node and $m$ is a $\cutpn$ or $\otimes$ node.
 \end{remark}
 
 \sequentialityCutElimination*
 
 \begin{proof}
  We reason by induction on $\sharp \ar(R)$. Let $n$ be the $\cutpn$ node reduced in the \cutelim step $R \to R'$. Suppose that there exists a node $m$ of $R$ such that $m \neq n$ and $R$ is $m$-\seq. Since $R$ contains $n$ and $m \neq n$, $m$ cannot be an $\axpn$ or $\onepn$ node. Therefore, $m$ is a $\cutpn$, $\bot$, $\otimes$ or $\parr$ node. By definition, there exist \seq \pss $R_1, R_2$ such that $R$ has the shape in the sub-figure of \Cref{fig:sequential} (in the figure, $m$ replaces $n$) determined by the label of $m$. Since $m \neq n$, we can suppose, without loss of generality, that $n$ is a node of $R_1$. Let~$k \coloneq 2$~if~$n$~is~a~$\cutpn$~or~$\otimes$ node, $k \coloneq 1$ otherwise. Then:
  \[
   \sharp \ar(R_1) \leq \sum_{i = 1}^k \sharp \ar(R_i) \leq \sharp \ar(R)
  \]
  If $n$ is a $\cutpn$ node, then the first inequality is strict, otherwise the second inequality is strict. In every case, $\sharp \ar(R_1) < \sharp \ar(R)$. Now consider the \ps $R_1'$ obtained from $R_1$ by reducing $n$. Since we have $R_1 \tocut R_1'$, we can apply the induction hypothesis on $R_1$, from which we deduce that $R_1'$ is \seq. Since $R'$ has the shape in \Cref{fig:sequential} (in the picture,~$R_1'$~replaces~$R_1$,~and~$m$ replaces $n$) determined by the label of $m$, and~since~$R_1', R_2$~are~\seq,~we~have~that~$R'$~is~\seq.
 	
  Now suppose that, for every node $m$ of $R$, we have that $R$ is not $m$-\seq or $m = n$. Since $R$ is \seq, we obtain that $R$ is $n$-\seq. Hence, there exist \seq \pss~$R_1, R_2$ such that $R$ has the shape in \Cref{subfig:sequential-cut}. For each $i \in \{1, 2\}$, let $n_i$ be the node of $R_i$ having a premise of $n$ among its conclusions. By \Cref{cor:sequential}, we know that $R_i$ is $n_i$-\seq. Hence, there are \seq \pss $R_i', R_i''$ such that $R_i$ has the shape~in~\Cref{fig:sequential}~(in~the~figure,~$R_i'$~replaces $R_1$, $R_i''$ replaces $R_2$, and $n_i$ replaces $n$) determined by the label of $n_i$. We consider the~cases:
  \begin{itemize}
   \item
  	\emph{\Axcut:} Then at least one between $n_1$ and $n_2$ is an $\axpn$ node. We can suppose that $n_1$ is an $\axpn$ node without loss of generality. Then $R' = R_2$. Hence, we have that $R'$~is~\seq;
   \item
  	\emph{\Ucut:} By possibly exchanging the roles of $n_1$ and $n_2$, one may assume that $n_1$~is~a~$\onepn$~node and $n_2$ is a $\bot$ node. Then $R' = R_2'$,~so $R'$ is \seq;
   \item
  	\emph{\Mcut:} By possibly exchanging the roles of $n_1$ and $n_2$, one may suppose that $n_1$ is a $\otimes$ node and $n_2$ is a $\parr$ node. For each $i \in \{1, 2\}$, let $a_i'$ (resp.~$a_i''$) be the~left~(resp.~right) premise of $n_i$. Then $a_1'$ is an arc of $R_1'$, $a_1''$ is an arc of $R_1''$, and $a_2', a_2''$ are arcs of $R_2'$. Let $R_0$ be the \ps obtained from $R_1'$ and $R_2'$ by adding a $\cutpn$~node~$n_0$~having~$a_1'$~and~$a_2'$~as~premises. Since $R_0$ has the shape in \Cref{subfig:sequential-cut} (in the picture, $R_i'$ replaces $R_i$ for each $i \in \{1, 2\}$,~and $n_0$ replaces $n$), and since $R_1', R_2'$ are \seq, $R_0$ is $n_0$-\seq. We now notice that $R'$ is obtained from $R_1''$ and $R_0$ by adding a $\cutpn$ node $n'$ having $a_1''$ and $a_2''$ as premises. Since $R'$ has the shape~in~\Cref{subfig:sequential-cut}~(in~the picture, $R_1''$ replaces $R_1$,~$R_0$~replaces~$R_2$,~and~$n'$~replaces $n$), and since $R_1'', R_0$ are \seq, we can conclude that $R'$ is $n'$-\seq. \qedhere
 	\end{itemize}
 \end{proof}

 \subsection{Complements of \texorpdfstring{\Cref{sec:untyped-sequentiality-theorem}}{Section 3}}
 \label{subsec:cw-forall}
 
 \begin{remark}
  \label{rmk:switching-nodes}
  Let $R$ be a \ps, let $\varphi$ be a switching of $R$ and let $n$ be a node of $R^\varphi$. Then $n$ is either a node of $R$, or a $\bullet$ node.
 \end{remark}
 
 \begin{lemma}
  \label{lemma:erasing-tree}
  Let $R$ be a \ps, let $n$ be an \er node of $R$, and let $T$ be the~subgraph~of~$R$~that has the set $\smash{\{n\} \cup \{m : m \prec n\}}$ as set of nodes. Then $T$ is a rooted~tree~with~root~$n$, such that:
  \begin{enumerate}[(i)]
   \item
    Every node of $T$ is \er;
   \item
    A node of $T$ is a $\bot$ node if and only if it is a leaf of $T$.
  \end{enumerate}
 \end{lemma}
 
 \begin{proof}
  By a straightforward induction on the order $\prec$ (see \Cref{def:erasing-order}).
 \end{proof}
 
 \begin{corollary}
  \label{cor:erasing-above}
  Let $R$ be a \ps, let $n_1$ and $n_2$ be nodes of $R$, and let $\gamma$~be~a~path~of~$R$~from~$n_1$~to $n_2$ having a premise of $n_1$ as its first arc. If $n_1$ is \er, then $n_2$ is \er.
 \end{corollary}
 
 \begin{proof}
  Immediate consequence of \Cref{lemma:erasing-tree}.
 \end{proof}
 
 \cwForall*
 
 \begin{proof}
  Let us assume that $R \models \Cwf$ and consider a \wsp $\gamma$ of $R$ starting from an \er node $n$ and having the conclusion of $n$ as its first arc. Then,~by~\Cref{rmk:w-switching-paths},~there~exists a \wcs $\psi$ of $R$ such that $\gamma$ is a path of $R^\psi$. Since $R \models \Cwf$,~we~have~$R^\psi \models \Cwf$, meaning that every connected component of $R^\psi$ either has no \er nodes of $R$, or is a \thr. Since $n$ is an \er node of $R$, the connected component of $R^\psi$ which contains~$n$~is~a \thr. We can then conclude that $\gamma$ crosses only \er nodes of $R$ by~\Cref{rmk:thread}.
  
  Conversely, assume that, for every \wsp $\gamma$ of $R$ starting from an \er node $n$ of $R$ and having the conclusion of $n$ as its first arc, $\gamma$ crosses only \er nodes of $R$. Let us consider a \wcs $\psi$ of $R$ and $C \in \cc(R^\psi)$. By \Cref{rmk:thread}, it suffices to prove that, if $C$ has an \er node of $R$, then $C$ contains only \er nodes of $R$ and an (\er) $\bullet$ node of $R^\psi$. We then suppose that $C$ has an \er node $n_1$~of~$R$~and~we~consider~a~node $n_2$ of $C$ different from $n_1$. By \Cref{rmk:switching-nodes}, we know that $n_2$ is either a node of $R$, or a $\bullet$ node. It is enough to prove that, if $n_2$ is a node of $R$, then $n_2$ is \er in $R$. Since~$C$~is~connected,~there exists a path $\gamma$ of $R^\psi$ from $n_1$ to $n_2$. And~since~$n_1$~and $n_2$ are both nodes of $R$, we have that $\gamma$ is a path of $R$. By~\Cref{rmk:w-switching-paths},~$\gamma$~is~a~\wsp of $R$. There~are~two~possibilities:~either $\gamma$ has a premise of $n_1$ as its first arc, and then $n_2$ is an \er node of $R$ by \Cref{cor:erasing-above}, or $\gamma$ has the conclusion of $n_1$ as its first arc, and then $n_2$ is \er in $R$ because our hypothesis ensures that $\gamma$ crosses only \er nodes of $R$.
 \end{proof}
 
 \cwForallTensor*
 
 \begin{proof}
  The implication~\ref{itm:cw-forall}~$\Rightarrow$~\ref{itm:cw-forall-existential} is trivial. We observe that~\ref{itm:cw-forall-existential}~$\Rightarrow$~\ref{itm:wten-graph}: if $R$ is not a $\wten$-\ps, then there is a $\cutpn$ or $\otimes$ node $n$ of $R$ such that a premise of $n$ is the conclusion of an \er node $m$ of $R$. Then, for every (\wc) switching $\psi$ of $R$, the connected component $C$ of $R^\psi$ containing $m$ (an \er node of $R$) also contains $n$, which is a node with exactly~two premises in $C$. By \Cref{def:thread}, $C$ is not a \thr. Hence, there~is~no~(\wc)~switching~$\psi$~of~$R$ such that $R^\psi \models \Cwf$.
  
  We now prove that \ref{itm:wten-graph}~$\Rightarrow$~\ref{itm:cw-forall}. By~\Cref{lemma:cw-forall}, it is enough to prove that, for every \wsp $\gamma$ of $R$ starting from an \er node $n$ of $R$ and having the conclusion $a$ of $n$ as its first arc, $\gamma$ crosses only \er nodes of $R$. We proceed by induction on the length of $\gamma$. Consider a \wsp $\gamma_0$ of $R$ starting from an \er node $n_0$ of $R$ and having~the~conclusion $a_0$ of $n_0$ as its first arc. Since $R$ is a $\wten$-\ps, and by \Cref{rmk:w-switching-erasing}, we know that $a_0$~is~a~premise of an \er $\parr$ or $\bullet$ node $n_1$ of $R$. If $n_1$ is the last node of $\gamma$, then the only nodes of $\gamma$ are $n_0$ and $n_1$, which are both \er nodes, and we are done. Otherwise, we have~$\gamma_0 = n_0 a_0 \gamma_1$~with $\gamma_1$ a \wsp of $R$ starting from $n_1$, and $n_1$ is an \er $\parr$ node, so we can consider its conclusion $a_1$. Since $\gamma_1$ is, in particular, a \spath of $R$, we have that $a_1$ is~the~first arc of $\gamma_1$. And since $\gamma_1$ is strictly~shorter~than~$\gamma_0$,~by~induction~hypothesis~$\gamma_1$~crosses~only~\er nodes of $R$. Thus, $\gamma_0$ crosses only \er nodes of $R$.
 \end{proof}
 
 \wtenConnected*
 
 \begin{proof}
  As for \Cref{itm:existence-non-erasing}, by \Cref{lemma:cw-forall-tensor} we have $R \models \Cwf$. Let $\psi$ be a \wcs of $R$. From $R \models \Cwf$ and $R \models \Cw$ we deduce that $R^\psi \models \Cwf$ and $R^\psi \models \Cw$. By \Cref{rmk:cw-one}, exactly one connected component $C$ of $R^\psi$ has no erasing node of $R$: any~node~of~$C$~is~non-\er~in~$R$.
 	
  As for \Cref{itm:erasing-component}, let $m$ be a node of $R_0$ and let $\psi$ be a \wcs of $R$. Then $n$ and $m$ belong to two different connected components of $R^\psi$. Since $\cutpn$ and $\otimes$ nodes are never \er, $n$ is not \er in $R$. Then,~by~\Cref{rmk:cw-one,lemma:cw-forall-tensor},~the~connected~component~of $R^\psi$ that contains $m$ is a \thr. Hence, by \Cref{rmk:thread}, $m$~is~an~\er~node~of~$R$.
 	
  As for \Cref{itm:connectivity}, either $R$ has no terminal $\otimes$ or $\cutpn$ nodes, in which case $R$ consists of a unique terminal $\axpn$ or $\onepn$ node, and we are done, or there exists a terminal $\cutpn$ or $\otimes$ node $n$. In this case, by \Cref{itm:erasing-component}, a connected component of $R$ which does not contain~$n$~would~only~have~\er nodes: this would imply that $R$ has a terminal \er node, which contradicts the hypothesis. Thus, $R$ has a unique connected component.
 \end{proof}
 
 \begin{lemma}
  \label{lemma:ac}
  Let $R$ be a \ps, let $n$ be a $\cutpn$, $\bot$, $\otimes$ or $\parr$ node of $R$, and let $R_1, R_2$ be \pss such that $R$ has the shape in the sub-figure of \Cref{fig:sequential} determined by the label of $n$. Let $k \coloneq 2$ if $n$ is a $\cutpn$ or $\otimes$ node, $k \coloneq 1$ otherwise. Then~$R \models \AC$~if~and~only~if,~for~every~$i \in \{1, \dots, k\}$,~we have $R_i \models \AC$.
 \end{lemma}
 
 \begin{proof}
  For every switching $\varphi$ of $R$, and for every $i \in \{1, \dots, k\}$, let $\varphi_i$ be the restriction of $\varphi$ to a switching of $R_i$. Then, for every cycle $\gamma$, $\gamma$ is in $R^\varphi$ if and only if there exists $i \in \{1, \dots, k\}$ such that $\gamma$ is in $R_i^{\varphi_i}$. Therefore, $R^\varphi \models \AC$ if and only if, for every $i \in \{1, \dots, k\}$, $R_i^{\varphi_i} \models \AC$. Since this holds for an arbitrary switching $\varphi$ of $R$, we obtain that $R \models \AC$ if and only if, for every $i \in \{1, \dots, k\}$, $R_i \models \AC$.
 \end{proof}
 
 \untypedSequentiality*
 
 \begin{proof}
  We reason by induction on $\sharp \ar(R)$. Consider the following cases:
  \begin{itemize}
   \item
 	\emph{$R$ has a terminal $\bot$ node $n$.} There exists a \ps $R_1$ such that $R$ has the shape in \Cref{subfig:sequential-bottom}. Observe that $R_1 \models \Cw$. Indeed, for every switching $\varphi$ of $R_1$ (or, equivalently, of $R$),~from $R^\varphi \models \Cw$ we deduce that $R_1^\varphi \models \Cw$:
 	\[
 	 \sharp \cc(R_1^\varphi) = \sharp \cc(R^\varphi) - 1 = (\sharp \w(R) + 1) - 1 = \sharp \w(R) = \sharp \w(R_1) + 1
 	\]
   \item
 	\emph{$R$ has a terminal $\parr$ node $n$.} There exists a \ps $R_1$ such that $R$ has the shape in \Cref{subfig:sequential-par}. Observe that $R_1 \models \Cw$. Indeed, for every switching $\varphi_1$ of $R_1$, and for every extension~$\varphi$~of $\varphi_1$ to a switching of $R$, from $R^\varphi \models \Cw$ we deduce that $R_1^{\varphi_1} \models \Cw$:
 	\[
 	 \sharp \cc(R_1^{\varphi_1}) = \sharp \cc(R^\varphi) = \sharp \w(R) + 1 = \sharp \w(R_1) + 1
 	\]
  \end{itemize}
  In both cases, by \Cref{lemma:ac}, $R_1 \models \AC$, so $R_1 \models \ACCw$. Moreover, since $R$ is a $\wten$-\ps, so is $R_1$. And since $\smash{\sharp \ar(R_1) < \sharp \ar(R_1) + 1 = \sharp \ar(R)}$, we can apply the induction hypothesis~to~obtain that $R_1$ is \seq. It follows immediately that $R$ is $n$-\seq.
  
  Now suppose that $R$ has no terminal $\bot$ or $\parr$ node. Since $R$ is a $\wten$-\ps and $R \models \Cw$, from \Cref{itm:connectivity} of \Cref{lemma:wten-connected} we deduce that $R$ is a connected graph. There~are~two~possibilities:
  \begin{itemize}
   \item
 	\emph{$R$ has no terminal $\cutpn$ or $\otimes$ node.} Then every terminal node of $R$ is an $\axpn$ or $\onepn$ node. Since $R$ is connected, $R$ is made up of a unique terminal $\axpn$ or $\onepn$ node~$n$,~and~thus~we~can immediately conclude that $R$ is $n$-\seq;
   \item
    \emph{$R$ has a terminal $\cutpn$ or $\otimes$ node.} By \Cref{lemma:splitting}, there exists a splitting $\cutpn$ or $\otimes$ node $n$ of $R$. By \Cref{def:splitting}, there exist two \pss $R_1, R_2$ such that $n$ splits $R$ into $R_1$ and $R_2$ and, since $R$ is connected, $R_1, R_2$ are uniquely determined. It suffices to prove that $R_1, R_2$ are \seq: then $R$ is $n$-\seq. We only prove that $R_1$ is \seq: $R_2$ is \seq by an analogous argument. Since $\smash{\sharp \ar(R_1) < \sharp \ar(R_1) + 1 \leq \sharp \ar(R)}$, we can conclude by applying the induction hypothesis on $R_1$. By~\Cref{lemma:ac},~$R_1 \models \AC$,~and~$R_1$~is~a~$\wten$-\ps~because $R$ is.
    
    The only thing left to prove is that $R_1 \models \Cw$. By \Cref{lemma:cw-forall-tensor}, $R \models \Cwf$ and $R_1 \models \Cwf$. Let $\psi_1$ be a \wcs of $R_1$, and let $\psi$ be any extension of $\psi_1$ to a \wcs of $R$. Then we have $R^\psi \models \Cwf$ and $R_1^{\psi_1} \models \Cwf$. Since $R^\psi \models \Cwf$ and $R^\psi \models \Cw$, by \Cref{rmk:cw-one} exactly one connected component $C_{\neg \w}$ of $R^\psi$ has no \er node of $R$. And since $\cutpn$ and $\otimes$ nodes are never \er, $n$ is a node of $C_{\neg \w}$. Then the connected component $C_{\neg \w}'$ of $R_1^{\psi_1}$ which contains a premise of $n$ has no \er node of $R$. Moreover, for every $C_1 \in \cc(R_1^{\psi_1})$ with $C_1 \neq C_{\neg \w}'$, $C_1$ is actually a connected component of $R^\psi$, and $C_1 \neq C_{\neg \w}$. Since $R^\psi \models \Cwf$, $C_1$ is a \thr, so $C_1$ has a $\bot$ node. Therefore, $C_{\neg \w}'$ is the only connected component of $R_1^{\psi_1}$ with no \er node of $R$. Since~$R_1^{\psi_1} \models \Cwf$,~by~\Cref{rmk:cw-one}~we have~$R_1^{\psi_1} \models \Cw$. We~can~then~conclude,~by~\Cref{itm:acyclicity-cc-existential}~of~\Cref{prop:acyclicity-cc},~that~$R_1 \models \Cw$. \qedhere
  \end{itemize}
 \end{proof}

 \subsection{Complements of \texorpdfstring{\Cref{sec:btenll}}{Section 4}}
 \label{subsec:equivalence-canonical-jumps}
 
% \begin{example}
%  \begin{figure}
%   \centering
%%   \scalebox{\figscale}{
%   	\input{Figures/canonical-jump.tikz}
%%   }
%   \caption{\label{fig:canonical-jump}A \ps of $\btenll$.}
%  \end{figure}
%  
%  Let $R$ be the \ps of $\btenll$ in \Cref{fig:canonical-jump}, and let $n$ be the~unique~$\bot$~node~of $R$. Then $R_n$ is obtained from $R$ by defining $J_{R_n}(n)$ as the unique $\otimes$ node of $R$, and $\can{R} = R_n$.
% \end{example}
 
 \sequentializationCanonicalJump*
 
 \begin{proof}
  We reason by induction on $\sharp \ar(R)$.
 	
  Suppose $R$ has a terminal \er node $n$. Then $n$ is either a $\bot$ or $\parr$ node. Therefore,~there exists a \ps $R_1$ such that $R$ has the shape in \Cref{subfig:desequentialization-bottom} or~\ref{subfig:desequentialization-par}. Since $n$ is \er and $m$ is non-\er, we have $n \neq m$, thus $m$ is a node of $R_1$. By following the same argument seen in the proof of \Cref{thm:untyped-sequentiality}, we obtain that $R_1$ satisfies the induction hypothesis. Hence,~there exists a \seqc proof $\pi_1$ in $\btenll$ such that $\pi_1 \deseqjump \canext{R_1}{m}$. Let $\pi$ be the \seqc proof obtained from $\pi_1$ by applying a $\bot$ or~$\parr$~rule~introducing~the~type~of~the~conclusion~of~$n$. Then $\pi \deseqjump \canext{R}{m}$, because:
  \begin{itemize}
   \item
 	If $n$ is a $\bot$ node, by \Cref{def:canonical-jump}, we have $J_{\canext{R}{m}}(n) = m$;
   \item
 	If $n$ is a $\parr$ node, we have $\wparr(R) = \wparr(R_1)$, and thus $J_{\canext{R}{m}} = J_{\canext{R_1}{m}}$.
  \end{itemize}
 	
  We can now assume that $R$ has no terminal \er node. By \Cref{itm:canonical-jump-terminal} of \Cref{rmk:canonical-jump},~we~have $\w(R) = \wparr(R)$. By \Cref{itm:canonical-jump} of \Cref{rmk:canonical-jump}, $\canext{R}{m} = \can{R}$. We suppose that $R$ has a terminal non-\er $\parr$ node $n$. Then, by \Cref{def:erasing}, we can consider a non-\er node $m_1$ having a premise of $n$ among its conclusions. As before, there exists a \ps $R_1$ such that~$R$~has~the~shape in \Cref{subfig:desequentialization-par} and, since $m_1 \in \ve(R_1)$, we can apply the induction hypothesis on $R_1, m_1$. Hence, there exists a \seqc proof $\pi_1$ in $\btenll$ such that $\pi_1 \deseqjump \canext{R_1}{{m_1}}$. Let $\pi$ be the \seqc proof obtained from $\pi_1$ by applying a $\parr$ rule which introduces the type~of~the conclusion of $n$. Then $\pi \deseqjump \can{R}$. Indeed, let $n_1 \in \w(R_1) = \w(R)$:
  \begin{itemize}
   \item
    If $n_1 \in \wparr(R_1)$, then we immediately get $J_{\can{R}}(n_1) = J_{\canext{R_1}{{m_1}}}(n_1)$;
   \item
    If $n_1 \in \we(R_1)$, then $n$ is the least non-\er node of $R$ such that $n_1 \prec n$, and therefore:
    \[
     J_{\can{R}}(n_1) = m_1 = J_{\canext{R_1}{{m_1}}}(n_1)
    \]
  \end{itemize}
 	
  We can now suppose that $R$ has no terminal $\bot$ or $\parr$ node. As in the proof of \Cref{thm:untyped-sequentiality}, we can apply \Cref{itm:connectivity} of \Cref{lemma:wten-connected} to deduce that $R$ is a connected graph. There~are~two~cases:
  \begin{itemize}
   \item
    \emph{$R$ has no terminal $\otimes$ node.} Then every terminal node of $R$ is an $\axpn$ or $\onepn$ node. Therefore, we trivially have $\can{R} = R$. Since $R$ is connected, $R$ is made up of~a~unique terminal $\axpn$ or $\onepn$ node $n$. Let $\pi$ be the \seqc proof consisting of~a~unique~\axsc~or~$\one$~rule~introducing the types of the conclusions of $n$. Then we immediately get $\pi \deseqjump R$;
   \item
  	\emph{$R$ has a terminal $\otimes$ node.} By \Cref{lemma:splitting}, there exists a splitting $\otimes$ node $n$ of $R$. By \Cref{def:splitting}, there exist two \pss $R_1, R_2$ such that $n$ splits $R$ into $R_1$ and $R_2$ and,~since $R$ is connected, $R_1, R_2$ are uniquely determined. We prove that there exists~a~\seqc proof $\pi_1$ in $\btenll$ such that $\pi_1 \deseqjump \can{R_1}$. An analogous argument determines a \seqc proof $\pi_2$ in $\btenll$ such that $\pi_2 \deseqjump \can{R_2}$. We will then define $\pi$ as the \seqc proof obtained from $\pi_1$ and $\pi_2$ by applying a $\otimes$ rule~that~introduces~the~type~of~the conclusion of $n$, and deduce $\pi \deseqjump \can{R}$ because, for every $n_0 \in \w(R)$:
 	\begin{equation}
     \label{eqn:canonical-jumps-tensor}
     J_{\can{R}}(n_0) =
     \begin{cases}
      J_{\can{R_1}}(n_0) & \text{if $n_0 \in \w(R_1)$,} \\
      J_{\can{R_2}}(n_0) & \text{otherwise (i.e.~$n_0 \in \w(R_2)$).}
     \end{cases}
    \end{equation}
    First of all, by following the same argument we made in the proof of~\Cref{thm:untyped-sequentiality},~we~obtain that $R_1$ satisfies the induction hypothesis, and in particular $R_1 \models \Cw$. We~observe~that,~by \Cref{def:wten}, the premise of $n$ in $R_1$ is a conclusion of a non-\er node $m_1$. Thus,~$R_1$~has no terminal \er node, and then $\canext{R_1}{{m_1}} = \can{R_1}$ by \Cref{itm:canonical-jump-terminal} of \Cref{rmk:canonical-jump}. We~conclude by applying the induction hypothesis on $R_1$, $m_1$. \qedhere
  \end{itemize}
 \end{proof}
 
 \begin{lemma}
  \label{lemma:erasing-btenll}
  Let $R$ be a \ps of $\btenll$ and let $n$ be an \er node of $R$.
  \begin{enumerate}[(i)]
   \item
    Either $n$ is terminal, or its conclusion is a premise of a $\parr$ node of $R$;
   \item \label{itm:erasing-btenll-descent}
    Either every node $m$ such that $n \prec m$ is an \er $\parr$ or $\bullet$ node, or~there~is~a~non-\er $\parr$ node $n_0$ of $R$ such that $n \prec n_0$.
  \end{enumerate}
 \end{lemma}
 
 \begin{proof}
  Immediate consequence of \Cref{def:btenll}.
 \end{proof}
 
 \begin{lemma}
  \label{lemma:adding-jump-ac}
  Let $R$ be a \ps of $\btenll$, let $\smash{n \in \w(R) \setminus \dom(J_R)}$, and let $m \in \ve(R)$. If $R \models \ACCw$ and $R_n^m \models \AC$, then $R_n^m \models \ACCw$.
 \end{lemma}
 
 \begin{proof}
  We set $R_0 \coloneq R_n^m$. It suffices to observe that, if $R_0 \models \AC$, then, by the hypothesis that $R \models \ACCw$, and by \Cref{itm:acyclicity-cc-number} of \Cref{prop:acyclicity-cc}, for any~switching~$\varphi$~of~$R_0$~(or,~equivalently,~of~$R$):
  \begin{align*}
   \sharp \cc(R_0^\varphi) & = \sharp \ve(R_0^\varphi) - \sharp \ar(R_0^\varphi) = \sharp \ve(R^\varphi) - (\sharp \ar(R^\varphi) + 1) = (\sharp \ve(R^\varphi) - \sharp \ar(R^\varphi)) - 1 \\
   & = \sharp \cc(R^\varphi) - 1 = (\sharp \w(R) - \sharp \dom(J_R) + 1) - 1 = \sharp \w(R) - \sharp \dom(J_R) \\
   & = \sharp \w(R_0) - (\sharp \dom(J_{R_0}) - 1) = \sharp \w(R_0) - \sharp \dom(J_{R_0}) + 1 \qedhere
  \end{align*}
 \end{proof}
 
 \begin{definition}
  The jump function $J_R$ of a \ps $R$ of $\mllu$ is \emph{non-\er} (resp.~\emph{initial}) when, for every $\bot$ node $n$ of $R$, we have that $J_R(n)$ is a non-\er (resp.~an~$\axpn$~or~$\onepn$)~node~of~$R$.
 \end{definition}
 
 \begin{remark}
  \label{rmk:initial-non-erasing}
  If the jump function $J_R$ of a \ps $R$ of $\mllu$ is initial, then $J_R$ is non-\er.
 \end{remark}
 
 \begin{lemma}
  \label{lemma:adding-jump}
  Let $R$ be a \ps of $\btenll$ such that $J_R$ is non-\er and $R \models \ACCw$.
  \begin{enumerate}[(i)]
   \item
    For every $n \in \wparr(R) \setminus \dom(J_R)$, we have $R_n \models \ACCw$;
   \item
 	For every $\smash{n \in \we(R) \setminus \dom(J_R)}$ and $m$ non-\er node of $R$, we have $R_n^m \models \ACCw$.
  \end{enumerate}
 \end{lemma}
 
 \begin{proof}
  Let $\smash{n \in \w(R) \setminus \dom(J_R)}$. We define $R_0 \coloneq R_n$ if $n \in \wparr(R)$, otherwise $R_0 \coloneq R_n^m$. By \Cref{lemma:adding-jump-ac}, it is sufficient to prove that $R_0 \models \AC$. Let $\varphi$ be a switching of $R_0$ (or, equivalently, a switching of $R$), let $C$ be the connected component of $R^\varphi$~containing~$n$,~and~let~$j$~be~the~only arc of $R_0^\varphi$ which is not in $R^\varphi$.
  \begin{enumerate}[(i)]
   \item
    We set $n' \coloneq J_{R_n}(n)$ and consider the least non-\er $\parr$ node $n_0$ of $R$ such that~$n \prec n_0$. Then $j$ is an arc from $n$ to $n'$, and we distinguish two cases:
    \begin{itemize}
     \item
      \emph{$n_0$ does not belong to $C$.} By \Cref{lemma:erasing-btenll} and by minimality of $n_0$, we know that $C$ is a \thr (\Cref{def:thread}) and that every node of $C$ is \er. Hence, $n'$ is not a node of $C$. Since $R^\varphi \models \AC$ and $j$ is an arc between two~distinct~connected~components~of~$R^\varphi$,~we obtain that $R_n^\varphi \models \AC$;
     \item
      \emph{$n_0$ belongs to $C$.} Suppose, for the sake of contradiction, that $R_n^\varphi$ has a cycle $\gamma$. Since $R^\varphi \models \AC$, and since $R^\varphi$ is obtained from $R_n^\varphi$ by erasing $j$, we have $j \in \gamma$. Moreover, by \Cref{lemma:erasing-btenll}, by minimality of $n_0$, by the assumption that $n_0$~belongs~to~$C$,~and~by~the hypothesis that $J_R$ is non-\er, we can write $\gamma = \delta \, \gamma' j \, n$ with $\delta$ a \dpath of $R^\varphi$ from $n$ to $n_0$ and $\gamma'$ a path of $R^\varphi$ from $n_0$ to $n'$. Let $a \coloneq \varphi(n_0)$,~let~$\varphi'$~be~the~switching of $R$ such that $\varphi'(n_0) \neq a$~and,~for~every~$\parr$~node~$n'' \neq n_0$, $\varphi'(n'') = \varphi(n'')$, and~finally let $a' \coloneq \varphi'(n_0)$. We observe that $a'$ is the conclusion of $n'$ and that $\gamma'$ does~not~contain $a'$, because $\gamma'$ is a path of $R^\varphi$. Moreover, $\gamma'$ does not contain $a$, because $\delta$ contains $a$. Hence, $\gamma'$ is also a path of $R^{\varphi'}$. But then $\gamma_0 \coloneq \gamma' a' n_0$ is a cycle of $R^{\varphi'}$, contradicting $R \models \AC$. We have then proven that $R_n^\varphi \models \AC$.
    \end{itemize}
   \item
    We have that $j$ is an arc from $n$ to $m$. Since $n \in \we(R)$, by \Cref{lemma:erasing-btenll} every node $m'$ such that $n \prec m'$ is an erasing $\parr$ or $\bullet$ node. Therefore, $C$ is a thread and contains only \er nodes. Since $m$ is not \er, $m$ does not belong to $C$. Since $R^\varphi \models \AC$, and since $j$ is an arc between~two~distinct~connected~components~of~$R^\varphi$,~we~can~conclude~that~$R_0^\varphi \models \AC$.
  \end{enumerate}
  In every case, since our choice of the switching $\varphi$ was arbitrary, we obtain that $R_0 \models \AC$.
 \end{proof}
 
% \begin{corollary}
%  \label{cor:canonical-jump-accw}
%  Let $R$ be a \ps of $\btenll$. If $R \models \ACCw$, then $\can{R} \models \ACCw$.
% \end{corollary}
 
 \begin{corollary}
  Let $R$ be a \ps of $\btenll$ such that $J_R$ is non-\er and $R \models \ACCw$,~and let $m$ be a non-\er node of $R$. Then $\canext{R}{m} \models \ACC$.
 \end{corollary}
 
 \begin{proof}
  Immediate consequence of \Cref{itm:acc-jump} of \Cref{rmk:jump}, \Cref{itm:canonical-jump-total} of \Cref{rmk:canonical-jump}, \Cref{lemma:adding-jump}.
 \end{proof}
 
 \begin{lemma}
  \label{lemma:equivalence-canonical-jump}
  Let $R$ be a \jc \ps of $\btenll$ such that $J_R$ is non-\er,~and~let~$m$ be a non-\er node of $R$. Then $R \equivpns \canext{R}{m}$.
 \end{lemma}
 
 \begin{proof}
  By induction on $k_R + h_R$, where:
  \begin{align*}
   k_R & \coloneq \sharp \{n : \text{$n \in \wparr(R)$ and $R \neq (R_{\check{n}})_n$}\} \\
   h_R & \coloneq \sharp \{n : \text{$n \in \we(R)$ and $R \neq (R_{\check{n}})_n^m$}\}
  \end{align*}
  If $h_R = k_R = 0$, then we are done, because $R = \canext{R}{m}$. If $h_R > 0$~(resp.~$k_R > 0$),~let~$n \in \wparr(R)$ (resp.~$n \in \we(R)$) such that $R \neq R'$, where $R' \coloneq (R_{\check{n}})_n$ (resp.~$R' \coloneq (R_{\check{n}})_n^m$). Then~$R'$~is~\jt, $\ujf{R'} = \ujf{R}$, and, by \Cref{itm:acc-jump} of \Cref{rmk:jump} and \Cref{lemma:adding-jump}, we obtain $R' \models \ACC$. Hence, $R'$ is \jc, and then $R \rewpns R'$. Since $\smash{h_{R'} = h_R - 1 < h_R}$ (resp.~$\smash{k_{R'} = k_R - 1 < k_R}$), we can apply the induction hypothesis on $R'$, from which we deduce that $R' \equivpns \canext{R}{m}$. We can then conclude that $R \equivpns \canext{R}{m}$.
 \end{proof}
 
 \begin{proposition}[3.7 in \cite{heijltjes2014no}]
  \label{prop:equivalence-initial}
  For every \jc \ps $R$ of $\mllu$, there is a \jc \ps $R'$ of $\mllu$ such that $R \equivpns R'$ and $J_{R'}$ is initial.
 \end{proposition}
 
 \equivalenceCanonicalJump*
 
 \begin{proof}
  Straightforward consequence of \Cref{rmk:initial-non-erasing}, \Cref{lemma:equivalence-canonical-jump} and \Cref{prop:equivalence-initial}.
 \end{proof}

 \equivalence*
 
 \begin{proof}
  By \Cref{itm:existence-non-erasing} of \Cref{lemma:wten-connected} and \Cref{itm:wten-proof-structure} of \Cref{rmk:erasing}, there exists a non-\er node $m$ of $R_1$. Since $\ujf{R_1} = \ujf{R_2}$, $m$ is also a node of $R_2$, and $\canext{R_1}{m} = \canext{R_2}{m}$. Hence,~by~\Cref{prop:equivalence-canonical-jump}:
  \[
   R_1 \equivpns \canext{R_1}{m} = \canext{R_2}{m} \equivpns R_2 \qedhere
  \]
 \end{proof}
 
 \rulePermutationsRewiring*
 
 \begin{proof}
  (1) is Proposition~6 in~\cite{heijltjes2014no}. By \Cref{rmk:underlying-jump-free}, we have $\pi_1^\deseq = \pi_2^\deseq$ if and only if $\langle R_1 \rangle = \langle R_2 \rangle$. It is sufficient to observe that, if $R_1 \approx R_2$, then $\langle R_1 \rangle = \langle R_2 \rangle$ by \Cref{def:rewiring}, and conversely,~if $\langle R_1 \rangle = \langle R_2 \rangle$, then $R_1 \approx R_2$ by \Cref{cor:equivalence}.
 \end{proof}
 
 \subsection{Complements of \texorpdfstring{\Cref{sec:imell}}{Section 6}}
 \label{subsec:connectivity-polarities}
 
 We recall the following result:
 
 \accwOutputConclusion*
 
 In the spirit of the previous sections, we give the proof in the simplified framework~of~$\imll$ (\Cref{cor:accw-output-conclusion}). Every result we present up to \Cref{cor:accw-output-conclusion} holds in presence of $\cutpn$ nodes. We consider the connected components of \sgs induced by particular switchings, called \emph{\intu}, introduced by the following definition.
 
 \begin{definition}
  Let $R$ be a \ps of $\imll$. A switching $\varphi$ of $R$ is \emph{\intu} when, for every (output) $\parr$ node $n$ of $R$ with output premise $a$, we have that $\varphi(n) = a$. A~path~$\gamma$~of~$R^\varphi$~is~\emph{output} (resp.~\emph{input}) if, for every $a \in \gamma$, $a$ is output (resp.~input).
 \end{definition}
 
 \begin{notation}
  \label{not:unique-path-ac}
  Let $R$ be a \ps. We denote by $\nof{a}$ the node of $R$ having among its conclusions the arc $a$. If $\varphi$ is a switching of $R$ such that $R^\varphi \models \AC$, and $n_1, n_2$ are nodes of~the~same~connected component of $R^\varphi$, the unique path of $R^\varphi$ between $n_1$ and $n_2$ is denoted by $\up{n_1}{n_2}{\varphi}$. Finally,~if~$a, b$ are arcs of the same connected component of $R^\varphi$, we write $\up{a}{b}{\varphi} \coloneq \up{\nof{a}}{\nof{b}}{\varphi}$.
 \end{notation}
 
 \begin{remark}
  \label{rmk:output-input-arc}
  Let $R$ be a \ps of $\imll$, and let $a$ be an arc of $R$.
  \begin{enumerate}[(i)]
   \item \label{itm:output-premise}
    If $a$ is output, then $a$ is a premise of a $\cutpn$, $\otimes$, output $\parr$ or $\bullet$ node of $R$;
   \item \label{itm:input-conclusion}
    If $a$ is input, then $\nof{a}$ is an $\axpn$, $\bot$, input $\parr$ or input $\otimes$ node of $R$.
  \end{enumerate}
 \end{remark}
 
 \begin{definition}
  Let $G$ be a graph. The set of paths of $G$ is denoted by $\pa(G)$. The binary relation $\subseteq$ on $\pa(G)$ is defined in the following way: we have $\gamma_1 \subseteq \gamma_2$ if and only if, for every $a \in \gamma_1$, we have $a \in \gamma_2$.
 \end{definition}
 
 \begin{remark}
  The binary relation $\subseteq$ is a partial order on $\pa(G)$. Moreover, since $\subseteq$ is finite, $\subseteq$ is well-founded and converse well-founded,~meaning~that~every~non-empty~subset~of~$\pa(G)$~has~a minimal and a maximal element.
 \end{remark}
 
 We now identify notable polarized paths in the \pss of $\imll$.
 
 \begin{lemma}
  \label{lemma:maximal-polarized-descent}
  Let $R$ be a \ps of $\imll$, let $a$ be an arc of $R$, let $\varphi$~be~a~switching~of~$R$~such~that $R^\varphi \models \AC$, and let $C$ be the connected component of $R^\varphi$ which contains $a$.
  \begin{enumerate}[(i)]
   \item \label{itm:maximal-output-descent}
    If $a$ is output, and if $\varphi$ is \intu, then there exists~a~$\cutpn$,~input~$\otimes$~or~$\bullet$~node~$n$~of~$C$ such that $\up{\nof{a}}{n}{\varphi}$ is the maximal output \dpath of $R^\varphi$ issued from $\nof{a}$;
   \item \label{itm:maximal-input-descent}
    If $a$ is input, then there is an $\axpn$ or $\bot$ node $n$ of $C$ such that $\up{n}{\nof{a}}{\varphi}$~is~the~maximal input \dpath of $R^\varphi$ reaching $\nof{a}$.
  \end{enumerate}
 \end{lemma}
 
 \begin{proof}
  Let us denote by $\succ$ the dual order of $\prec$ (see \Cref{def:erasing-order} and \Cref{rmk:prec}).
  \begin{enumerate}[(i)]
   \item
    By induction on $\succ$. By \Cref{itm:output-premise} of \Cref{rmk:output-input-arc}, $a$ is a premise of a $\cutpn$, $\otimes$, output $\parr$ or $\bullet$ node $m$ of $R$. We observe that $m$ is a node of $C$: this is obvious if $m$ is a $\cutpn$, $\otimes$ or $\bullet$ node, and uses the hypothesis that $\varphi$ is \intu if $m$ is an output~$\parr$~node~because,~in that case, we have $\varphi(m) = a$. Observe that $\up{\nof{a}}{m}{\varphi}$ only contains the~arc~$a$,~which~is~output, so $\up{\nof{a}}{m}{\varphi}$ is an output \dpath. Now, if $m$ is a $\cutpn$, input $\otimes$ or $\bullet$ node, then either $m$ has no conclusion, or $m$ has exactly one input conclusion. Thus, for every output \dpath $\gamma$ of $R^\varphi$ starting from $\nof{a}$, we have $\gamma \subseteq \up{\nof{a}}{m}{\varphi}$. Hence,~$\up{\nof{a}}{m}{\varphi}$~is~the~maximal~output \dpath of $R^\varphi$ issued from $\nof{a}$.
    
    Now suppose that $m$ is not a $\cutpn$, input $\otimes$ or $\bullet$ node. Then $m$ is either an output $\otimes$ or an output $\parr$ node. In particular, the conclusion $b$ of $m$ is output. Since~$m \succ \nof{a}$,~and~since $b$ is output, we can apply the induction hypothesis on $b$, from~which~it~follows~that~there~is~a $\cutpn$, input $\otimes$ or $\bullet$ node $n$ of $C$ such that $\up{m}{n}{\varphi}$ is the maximal output \dpath of $R^\varphi$ starting from $m$. And since $\up{\nof{a}}{n}{\varphi} = \up{\nof{a}}{m}{\varphi} \up{m}{n}{\varphi}$, $\up{\nof{a}}{n}{\varphi}$ is an output \dpath. Also, for every output \dpath $\gamma$ of $R^\varphi$ starting from $\nof{a}$, we have $\gamma = \up{\nof{a}}{m}{\varphi} \gamma'$,~where~$\gamma'$~is~an output \dpath issued from $m$. Therefore, by~maximality~of~$\up{m}{n}{\varphi}$,~we~obtain~that:
    \[
     \gamma = \up{\nof{a}}{m}{\varphi} \gamma' \subseteq \up{\nof{a}}{m}{\varphi} \up{m}{n}{\varphi} = \up{\nof{a}}{n}{\varphi}
    \]
    We then conclude that $\up{\nof{a}}{n}{\varphi}$ is the maximal output \dpath~of~$R^\varphi$~starting~from~$\nof{a}$.
   \item
    By induction on $\prec$. By \Cref{itm:input-conclusion} of \Cref{rmk:output-input-arc}, $\nof{a}$ is an $\axpn$, $\bot$, input $\parr$ or input $\otimes$ node of $R$. If $\nof{a}$ is an $\axpn$ or $\bot$ node, then the empty path $\up{\nof{a}}{\nof{a}}{\varphi}$~is~trivially~an input \dpath. Moreover, $\nof{a}$ has no premise. Thus, every input \dpath~$\gamma$~of~$R$~reaching~$\nof{a}$~is~empty. Then $\up{\nof{a}}{\nof{a}}{\varphi}$ is the maximal input \dpath of $R^\varphi$ reaching $\nof{a}$.
    
    Now suppose that $\nof{a}$ is not an $\axpn$ or $\bot$ node. Then $\nof{a}$ is an input $\parr$ or input $\otimes$ node of $R$. Let $b$ be the only input premise of $\nof{a}$ if $\nof{a}$ is an input $\otimes$ node, $b \coloneq \varphi(\nof{a})$ otherwise. Then $n_b$ is a node of $C$. Since $\nof{b} \prec \nof{a}$, and since $b$ is input, we can apply the induction hypothesis on $b$, from which we deduce that there exists an $\axpn$ or $\bot$ node $n$ of $C$ such that $\up{n}{\nof{b}}{\varphi}$ is the maximal input \dpath of $R^\varphi$ reaching $\nof{b}$. From the fact that $b$ is input, and from  $\up{n}{\nof{a}}{\varphi} = \up{n}{\nof{b}}{\varphi} b$, we deduce that $\up{n}{\nof{a}}{\varphi}$ is an input \dpath. Moreover, for every input \dpath $\gamma$ of $R^\varphi$ reaching $\nof{a}$, we~have~that~$\gamma = \gamma' b$,~where~$\gamma'$~is~an~input \dpath reaching $\nof{b}$. Therefore,~by~maximality~of~$\up{n}{\nof{b}}{\varphi}$,~we~obtain:
    \[
     \gamma = \gamma' b \subseteq \up{n}{\nof{b}}{\varphi} b = \up{n}{\nof{a}}{\varphi}
    \]
    We then conclude that $\up{n}{\nof{a}}{\varphi}$ is the maximal input \dpath of $R^\varphi$ reaching $n_a$. \qedhere
  \end{enumerate}
 \end{proof}
 
 \begin{definition}
  Let $R$ be a \ps of $\imll$, and let $\varphi$ be an \is of $R$ such that $R^\varphi \models \AC$. We define a binary relation $\less{\varphi}$ on $\{n : \text{$n$ $\axpn$ or $\onepn$ node of $R^\varphi$}\}$ as follows:~we have $n_1 \less{\varphi} n_2$ if and only if $n_1$ and $n_2$ belong to the same connected component of $R^\varphi$, and $\up{n_1}{n_2}{\varphi}$ contains an input conclusion of $n_1$ and the output conclusion of $n_2$.
 \end{definition}
 
 \begin{remark}
  Let $R$ be a \ps of $\imll$, and let $\varphi$ be an \is of $R$ such that $R^\varphi \models \AC$. Then $\less{\varphi}$ is a strict order relation. Therefore, since $\less{\varphi}$ is~finite,~$\less{\varphi}$~is~well-founded.
 \end{remark}
 
 \begin{lemma}
  \label{lemma:bottom-weakening-output}
  If $R$ is a \ps of $\imll$, if $\varphi$ is an \is of $R$, and if $R^\varphi \models \AC$, then every connected component of $R^\varphi$ contains at least a $\bot$~node~or~an~output~conclusion.
 \end{lemma}
 
 \begin{proof}
  Let $C$ be a connected component of $R^\varphi$, and suppose that $C$ contains no $\bot$ node. Then $C$ contains an $\axpn$ or $\onepn$ node, so we can consider an $\axpn$ or $\onepn$ node $n$ which is minimal with respect to $\less{\varphi}$. Let $a$ be the output conclusion of $n$. By \Cref{itm:maximal-output-descent} of \Cref{lemma:maximal-polarized-descent}, there exists a $\cutpn$, input $\otimes$ or $\bullet$ node $n_0$ of $C$ such that $\up{n}{n_0}{\varphi}$ is the maximal output \dpath of $R^\varphi$ issued from $n_a$. It is sufficient to prove that $n_0$ is a $\bullet$ node: the premise of $n_0$ is then an output conclusion of $C$. Suppose that $n_0$ is a $\cutpn$ or input $\otimes$ node, and consider the input premise $b$ of $n_0$. By \Cref{itm:maximal-input-descent} of \Cref{lemma:maximal-polarized-descent}, there exists an $\axpn$, or $\bot$ node $m$ of $C$ such that $\up{m}{n_b}{\varphi}$ is the maximal input \dpath of $R^\varphi$ reaching $n_b$. By the assumption that $C$ contains no $\bot$ node, $m$ is an $\axpn$ node. And since $\up{m}{n}{\varphi} = \up{m}{n_b}{\varphi} b \up{n_0}{n}{\varphi}$,~the~path~$\up{m}{n}{\varphi}$~contains~both~$a$,~which~is~the~output~conclusion~of~$n$, and an input conclusion of $m$. Therefore,~$m \less{\varphi} n$,~which~contradicts~the~minimality~of~$n$.
 \end{proof}
 
 The previous result does not exclude the possibility that a connected component of a \sg induced by an \is contains both a $\bot$ node \emph{and} an output conclusion. In the sequel, we prove that the disjunction expressed by \Cref{lemma:bottom-weakening-output} is exclusive (\Cref{cor:connected-components-intuitionistic}). The following results, relating connectivity and polarities,~hold~for~\emph{any}~\ps~$R$~of $\imll$: it is not required that $R \models \AC$.
 
 \begin{lemma}
  \label{lemma:connectivity-polarities}
  Let $R$ be a \ps of $\imll$, let $a, b$ be distinct arcs of $R$ such that $a$ is input, let $\varphi$ be an \is of $R$, and let $\gamma$ be a path of $R^\varphi$ from $\nof{a}$ to $\nof{b}$~having~$a$~as~its~first~arc. Then $b \in \gamma$ if and only if $b$ is output.
 \end{lemma}
 
 %False if a = b: ax-cut cycle with a input, then b in gamma and b input
 
 \begin{proof}
  We need to prove that:
  \begin{enumerate}[(i)]
   \item \label{itm:inside-output}
    If $b \in \gamma$, then $b$ is output;
   \item \label{itm:outside-input}
    If $b \notin \gamma$, then $b$ is input.
  \end{enumerate}
  We proceed by mutual induction on the length $k$ of $\gamma$. We observe that $k \geq 1$ because $a \in \gamma$.
  \begin{enumerate}[(i)]
   \item
    Suppose $b \in \gamma$. We have $k \geq 2$. Let $n$ be the node of $R^\varphi$ having $b$ as a premise. Cases:
    \begin{itemize}
     \item
      \emph{$\gamma$ contains a premise $b' \neq b$ of $n$.} Then we have $\gamma = \gamma' b' n \, b \, \nof{b}$~and~$n$~is~a~$\cutpn$~or~$\otimes$~node. At most one of the arcs $b, b'$ is input, so it is enough to prove that $b'$ is input. If~$a = b'$, we are done. Otherwise, $a \in \gamma'$. Since the length of $\gamma'$ is $\smash{k - 2 < k}$~and~$b' \notin \gamma'$,~we~can apply \Cref{itm:outside-input} of the induction hypothesis on $b', \gamma'$, from which~it~follows~that~$b'$~is~input;
     \item
      \emph{$\gamma$ contains the conclusion $c$ of $n$.} Then $n = \nof{c}$, and $\gamma = \gamma' b \, \nof{b}$ with $\gamma' = \gamma'' c \, n$. Hence, $a \in \gamma'$ and $a \neq c$. Since the length of $\gamma'$ is $\smash{k - 1 < k}$ and $c \in \gamma'$,~we~can apply \Cref{itm:inside-output} of the induction hypothesis on $c, \gamma'$ to obtain that $c$ is output. Then~$n$~is~an~output $\otimes$ or output $\parr$ node. And since $\gamma$ is a path of $R^\varphi$ with $\varphi$ \is,~we~can conclude that $b$ is output;
    \end{itemize}
   \item
    Suppose $b \notin \gamma$. Cases:
    \begin{itemize}
     \item
      \emph{$\gamma$ contains a conclusion $b' \neq b$ of $\nof{b}$.} Then $\nof{b} = \nof{b'}$ is an $\axpn$ node and, since~$b \notin \gamma$,~we have $b' \in \gamma$ and $\nof{a} \neq \nof{b}$. In particular, $a \neq b'$. We observe that exactly~one~of~the~arcs $b, b'$ is input, so it is enough to prove that $b'$ is output. It suffices to apply~\Cref{itm:inside-output}~(which is already established for paths of length $k$) on $b', \gamma$;
     \item
      \emph{$\gamma$ contains a premise $c$ of $\nof{b}$.} Then $\gamma = \gamma' c \, \nof{b}$. We claim that $c$ is input. If~$a = c$,~then this is obvious. Otherwise, since the length of $\gamma'$ is $\smash{k - 1 < k}$ and $c \notin \gamma'$,~we~can~apply \Cref{itm:outside-input} of the induction hypothesis on $c, \gamma'$, from which we deduce that $c$ is input. Since $\varphi$ is \intu, we have that $\nof{b}$ is an input $\otimes$ or input~$\parr$~node. Thus,~$b$~is~input. \qedhere
    \end{itemize}
  \end{enumerate}
 \end{proof}
 
 \begin{lemma}
  \label{lemma:connectivity-polarities-up}
  Let $R$ be a \ps of $\imll$, let $a, b$ be distinct arcs of $R$, let $\varphi$ be an \is of $R$ and let $\gamma$ be a path of $R^\varphi$ from $\nof{a}$ to $\nof{b}$. If $a, b \notin \gamma$, then either $a$ or~$b$~is~input.
 \end{lemma}
 
 \begin{proof}
  By induction on the length $k$ of $\gamma$. If $k = 0$ or, equivalently, if $\gamma$ is empty, then $\nof{a} = \nof{b}$. Since $a \neq b$, we have that $\nof{a} = \nof{b}$ is an $\axpn$ node. Then exactly one of the~arcs~$a, b$~is~input. If $k \geq 1$, we consider the two following cases:
  \begin{itemize}
   \item
    \emph{$\gamma$ contains a conclusion $a' \neq a$ of $\nof{a}$.} Then $\nof{a} = \nof{a'}$ is an $\axpn$ node and, since $a \notin \gamma$, we have that $a' \in \gamma$ and $\nof{a} \neq \nof{b}$. In particular, $a' \neq b$. We observe that exactly one~of~the~arcs $a, a'$ is input. If $a$ is input, then we are done. If $a'$ is input, then~$b$~is~input~by~\Cref{lemma:connectivity-polarities};
   \item
    \emph{$\gamma$ contains a premise $c$ of $\nof{a}$.} Then $\gamma = \nof{a} c \, \gamma'$. In particular, $c \neq b$. Since the length of $\gamma'$ is $\smash{k - 1 < k}$ and $c, b \notin \gamma'$, we can apply the induction hypothesis on~$c, b, \gamma'$~to~obtain~that either $c$ or $b$ is input. If $b$ is input, then we are done. If $c$ is input, then $\nof{a}$ is an input~$\otimes$~or input $\parr$ node, because $\varphi$ is \intu. Therefore, $a$ is input. \qedhere
  \end{itemize}
 \end{proof}
 
 \begin{corollary}
  \label{cor:connectivity-polarities}
  Let $R$ be a \ps of $\imll$, let $a, b$ be distinct arcs of $R$, let $\varphi$ be an \is of $R$. If there exists a path $\gamma$ of $R^\varphi$ from $\nof{a}$ to $\nof{b}$ such that:
  \begin{enumerate}[(i)]
   \item \label{itm:down}
    $a, b \in \gamma$, then either $a$ or $b$ is output;
   \item \label{itm:mixed}
  	$a \in \gamma$ and $b \notin \gamma$, then either $a$ is output, or $b$ is input;
   \item \label{itm:up}
 	$a, b \notin \gamma$, then either $a$ or $b$ is input.
  \end{enumerate}
 \end{corollary}
 
 \begin{proof}
  Straightforward consequence of \Cref{lemma:connectivity-polarities,lemma:connectivity-polarities-up}.
 \end{proof}

 \begin{notation}
  Let $R$ be a \ps of $\imll$, let $\varphi$ be a switching of $R$, and let $C$ be a connected component of $R^\varphi$. We denote by $\out(C)$ the set of output conclusions of $C$.
 \end{notation}
 
 \begin{corollary}
  \label{cor:connected-components-intuitionistic}
  Let $R$ be a \ps of $\imll$, let $\varphi$ be an \is of $R$ such that $R^\varphi \models \AC$. For any connected component $C$ of $R^\varphi$, we have $\smash{\sharp \w(C) + \sharp \out(C) = 1}$. In~particular:%
  \begin{equation}
   \label{eqn:connected-components-imell}
   \sharp \cc(R^\varphi) = \sharp \w(R) + \sharp \out(R)
  \end{equation}
 \end{corollary}
 
 \begin{proof}
  Let $C$ be a connected component of $R^\varphi$. By \Cref{lemma:bottom-weakening-output}, $\smash{\sharp \w(C) + \sharp \out(C) \geq 1}$ (here~the hypothesis $R^\varphi \models \AC$ is crucial). We prove $\smash{\sharp \w(C) + \sharp \out(C) \leq 1}$. This follows straightforwardly from \Cref{cor:connectivity-polarities}: we have $\sharp \w(C) \leq 1$ by \Cref{itm:down}, $\sharp \out(C) \leq 1$ by \Cref{itm:up}, and either $\sharp \w(C) = 0$ or $\sharp \out(C) = 0$, by \Cref{itm:mixed}. Hence, $\smash{\sharp \w(C) + \sharp \out(C) = 1}$. \Cref{eqn:connected-components-imell} is obtained by summing on all the connected components of $R^\varphi$ and by observing that $\sharp \out(R^\varphi) = \sharp \out(R)$ by the hypothesis that $\varphi$ is an \is.
 \end{proof}
 
 By applying \Cref{itm:acyclicity-cc-same} of \Cref{prop:acyclicity-cc}, we immediately obtain the following result.
 
 \begin{corollary}
  \label{cor:connected-components-imell}
  Let $R$ be a \ps of $\imll$ such that $R \models \AC$, and let $\varphi$ be \emph{any} switching~of~$R$. Then \Cref{eqn:connected-components-imell} holds.
 \end{corollary}
 
 \begin{corollary}
  \label{cor:accw-output-conclusion}
  Let $R$ be a \ps of $\imll$ such that $R \models \AC$. Then $R \models \ACCw$ $\Leftrightarrow$ $\sharp \out(R) = 1$.
 \end{corollary}

 \begin{proof}
  Straightforward consequence of \Cref{cor:connected-components-imell}.
  %The implication $\Rightarrow$ follows immediately from \Cref{cor:connected-components-imell}. The converse $\Leftarrow$ is proven by induction on $\dpt(R)$, using \Cref{rmk:output-conclusion-box} and \Cref{cor:connected-components-imell}.
 \end{proof}
 
% \begin{lemma}
%  A sequent of $\imll$ is provable if and only if it has exactly one~output~formula.
% \end{lemma}
% 
% \begin{definition}
%  A \seqc proof $\pi$ in $\imll$ is called \emph{output context light} (\emph{\ol} for short) if, for every input $\otimes$ in a sequent of $\pi$, its output premise is derived~with~no~context.
% \end{definition}
% 
% \begin{lemma}
%  Let $\Gamma$ be a sequent of $\imll$. Then $\Gamma$ is provable if and only if there exists~an \ol \seqc proof in $\imll$ with conclusion $\Gamma$.
% \end{lemma}
 
 The rest of the appendix is consecrated to the proof of \Cref{thm:sequentialization-icomll}. Therefore, from now on, we only consider \seqc proofs with no \axsc or \cutsc rule, and \pss~with~no~$\axpn$~or~$\cutpn$~node.

 \begin{lemma}
  \label{lemma:wparr-intuitionistic}
  Let $R$ be a \ps of $\icomll$, and let $n \in \w(R)$. If there exists an output~node~$n_0$ of $R$ such that $n \prec n_0$, then $n \in \wparr(R)$.
 \end{lemma}
 
 \begin{proof}
  By \Cref{def:erasing-order}, there exists a non-empty \dpath $\delta$ of $R$ from $n$ to $n_0$. Let $n'$ be the first output node of $\delta$. Of course, a prefix of $\delta$ is a non-empty \dpath of $R$ from $n$ to $n'$, thus $n \prec n'$. Moreover, by minimality of $n'$, we know that $n'$ has an input premise and an output conclusion: $n'$ is an output $\parr$ node of $R$. Therefore,~by~\Cref{def:wparr-intuitionistic},~$n \in \wparr(R)$.
 \end{proof}

 \begin{remark}
  \label{rmk:icomll}
  Let $R$ be a \ps of $\icomll$. Then:
  \begin{enumerate}[(i)]
   \item \label{itm:icomll-one-conclusion}
    Every node of $R$ has exactly one conclusion;
   \item \label{itm:icomll-components-conclusions}
    Any two conclusions of $R$ belong to different connected components of $R$ (by \Cref{itm:icomll-one-conclusion});
   \item \label{itm:icomll-acyclicity}
    $R$ is an acyclic graph (by \Cref{itm:icomll-one-conclusion});
   \item \label{itm:icomll-ac}
    $R \models \AC$ (immediate consequence of \Cref{itm:icomll-acyclicity}).
  \end{enumerate}
 \end{remark}
 
 \begin{lemma}
  \label{lemma:sequentialization-icomll}
  Let $R$ be a \jf \ps of $\icomll$ such that $\sharp \out(R) = 1$. Then there exists~a \seqc proof $\pi$ in $\icomll$ such that $\pi \deseqjump \can{R}$.
 \end{lemma}
 
 \begin{proof}
  We reason by induction on $\sharp \ar(R)$.
  
  Suppose $R$ has a terminal input node $n$. By definition of $\icomll$, $n$ is either a $\bot$,~input~$\parr$ or input $\otimes$ node. We thus study these cases:
  \begin{itemize}
   \item
    \emph{$\bot$ (resp.~input $\parr$) node.} There exists a \ps $R_1$ such that $R$ is as in \Cref{subfig:desequentialization-bottom} (resp.~\ref{subfig:desequentialization-par}). Since $\sharp \out(R_1) = 1$ and $\smash{\sharp \ar(R_1) < \sharp \ar(R_1) + 1 = \sharp \ar(R)}$, we can apply the induction hypothesis on $R_1$ to obtain that there exists a \seqc proof $\pi_1$ in $\icomll$ such that $\pi_1 \deseqjump \can{R_1}$. We define $\pi$ as the \seqc proof obtained from $\pi_1$ by applying a $\bot$ (resp.~$\parr$)~rule introducing the type of the conclusion of $n$. We now consider the~two~sub-cases~separately:
    \begin{itemize}
     \item
      \emph{$\bot$ node.} Let $R_0$ be the \ps such that $\ujf{R_0} = \ujf{R}$, the restriction of $J_{R_0}$~to~$\w(R_1)$~is~$J_{\can{R_1}}$, and $n \notin \dom(J_R)$. From $\pi_1 \deseqjump \can{R_1}$ we deduce that $\pi \deseqjump (R_0)_n^m$ for every $m \in \ve(R_1)$. In particular, we have $\pi \deseqjump (R_0)_n$. Since $\can{R_1}$ is obtained from $\can{R}$ by removing~$n$,~we know that $(R_0)_n = \can{R}$. Therefore, we obtain $\pi \deseqjump \can{R}$;
     \item
      \emph{Input $\parr$ node.} We have $\wparr(R) = \wparr(R_1)$, and thus $J_{\can{R}} = J_{\can{R_1}}$. Since $\pi_1 \deseqjump \can{R_1}$, we immediately deduce that $\pi \deseqjump \can{R}$;
    \end{itemize}
   \item
    \emph{Input $\otimes$ node.}  By \Cref{def:imell}, $n$ has exactly one output premise $a$. We consider the \ps $R'$ obtained by removing $n$ from $R$, the connected component $R_1$ of $R'$ containing $a$, and the \ps $R_2$ obtained from $R'$ by removing $R_1$. By \Cref{itm:icomll-acyclicity} of \Cref{rmk:icomll}, $n$ splits $R$ into $R_1$ and $R_2$ (\Cref{def:splitting}). By \Cref{itm:icomll-components-conclusions} of \Cref{rmk:icomll}, and since $R_1$ is a connected graph, $a$ is the only conclusion of $R_1$. Now, since $a$ is output, $\sharp \out(R_1) = 1$. On the~other~hand,~since $n$ is input and $\sharp \out(R) = 1$, we also have $\sharp \out(R_2) = 1$. And since, for each $i \in \{1, 2\}$, we have $\smash{\sharp \ar(R_i) < \sharp \ar(R_i) + 2 \leq \sharp \ar(R)}$, we can apply the induction hypothesis on $R_i$, from which we deduce that there exists a \seqc proof $\pi_i$ in $\icomll$ such that $\pi_i \deseqjump \can{R_i}$. We prove that, for every $n_0 \in \w(R)$, \Cref{eqn:canonical-jumps-tensor} holds. This suffices to conclude: by taking $\pi$ as the \seqc~proof~obtained~from~$\pi_1$~and~$\pi_2$~by~applying~a~$\otimes$~rule~introducing~the type of the conclusion of $n$, we get $\pi \deseqjump \can{R}$. We thus consider any $n_0 \in \w(R)$. There~are two possibilities:
    \begin{itemize}
     \item
      \emph{$n_0 \in \w(R_1)$.} Since $a$ is the only conclusion of $R_1$, we know that $n_0 \prec n_a$ (\Cref{not:unique-path-ac}). And since $a$ is output, by \Cref{lemma:wparr-intuitionistic} we have $n_0 \in \wparr(R_1) \subseteq \wparr(R)$. We then obtain that $J_{\can{R}}(n_0) = J_{\can{R_1}}(n_0)$;
     \item
      \emph{$n_0 \in \w(R_2)$.} It is sufficient to observe that $\wparr(R_2) \subseteq \wparr(R)$ and $\wi(R_2) \subseteq \wi(R)$,~from which it follows immediately that $J_{\can{R}}(n_0) = J_{\can{R_2}}(n_0)$.
    \end{itemize}
  \end{itemize}
  
  We can now suppose $R$ has no terminal input node. In other words, every terminal node of $R$ is output. Then, by \Cref{lemma:wparr-intuitionistic}, we have $\w(R) = \wparr(R)$ (or, equivalently, the~set~$\wi(R)$~is empty). Moreover, by the hypothesis that $\sharp \out(R) = 1$, $R$ has exactly one terminal~output~node $n$. By definition of $\icomll$, $n$ is either a $\onepn$, output $\parr$ or output $\otimes$ node. Cases:
  \begin{itemize}
   \item
    \emph{$\onepn$ node.} The set $\w(R)$ is empty, thus $R = \can{R}$. By defining $\pi$ as the \seqc~proof consisting of a unique $\one$ rule, we immediately get $\pi \deseqjump R$;
   \item
    \emph{Output $\parr$ node.} There exists a \ps $R_1$ such that $R$ has the shape in \Cref{subfig:desequentialization-par}. It~is~immediate to check that $\sharp \out(R_1) = 1$ and that $\smash{\sharp \ar(R_1) < \sharp \ar(R_1) + 1 = \sharp \ar(R)}$. Therefore, we can apply the induction hypothesis on $R_1$ to obtain a \seqc proof $\pi_1$ in~$\icomll$~such~that $\pi_1 \deseqjump \can{R_1}$. We define $\pi$ as the \seqc proof obtained from $\pi_1$ by applying a $\parr$ rule introducing the type of the conclusion of $n$. It is enough to prove that $J_{\can{R}} = J_{\can{R_1}}$: from $\pi_1 \deseqjump \can{R_1}$ we will then immediately deduce that $\pi \deseqjump \can{R}$. Let $m$ be the node of $R_1$ having the output premise of $n$ as its conclusion. Then $\ta{m} = \ta{n}$ (notation introduced in \Cref{itm:output-leaf} of \Cref{lemma:output}), by uniqueness of $\ta{n}$. To~conclude~that~$J_{\can{R}} = J_{\can{R_1}}$,~consider~an arbitrary node $n_1 \in \w(R_1) = \w(R)$. There are two possibilities:
    \begin{itemize}
     \item
      \emph{$n_1 \in \wparr(R_1)$.} Then we immediately get $J_{\can{R}}(n_1) = J_{\can{R_1}}(n_1)$;
     \item
      \emph{$n_1 \in \wi(R_1)$.} Then $n$ is the least output $\parr$ node of $R$ such that $n_1 \prec n$. Since~$\ta{m} = \ta{n}$, we obtain $J_{\can{R}}(n_1) = J_{\can{R_1}}(n_1)$.
    \end{itemize}
   \item
    \emph{Output $\otimes$ node.} By \Cref{itm:icomll-acyclicity} of \Cref{rmk:icomll}, there exist \pss $R_1$ and $R_2$ such that $n$ splits $R$ into $R_1$ and $R_2$. Since $R$ has exactly one output conclusion, we know that $R$ is connected, so $R_1$ and $R_2$ are uniquely determined and connected. From \Cref{itm:icomll-components-conclusions}~of~\Cref{rmk:icomll}~we~then deduce that, for each $i \in \{1, 2\}$, $R_i$ has exactly one conclusion $a_i$. And since $n$ is~an~output $\otimes$ node, $a_i$ is output. Hence,~$\sharp \out(R_i) = 1$~and,~by \Cref{lemma:wparr-intuitionistic}, $\w(R_i) = \wparr(R_i)$. Notice that \Cref{eqn:canonical-jumps-tensor} trivially holds. Since~we~have~$\smash{\sharp \ar(R_i) < \sharp \ar(R_i) + 2 \leq \sharp \ar(R)}$,~we~can~apply~the induction hypothesis on $R_i$, from which we deduce that there exists a \seqc~proof $\pi_i$ in $\icomll$ such that $\pi_i \deseqjump \can{R_i}$. By taking $\pi$ as the \seqc proof obtained from $\pi_1$ and $\pi_2$ by applying a $\otimes$ rule introducing the type of the conclusion of~$n$,~we~can~now conclude that $\pi \deseqjump \can{R}$. \qedhere
  \end{itemize}
 \end{proof}
 
 By \Cref{cor:accw-output-conclusion} and \Cref{itm:icomll-ac} of \Cref{rmk:icomll}, we can reformulate \Cref{lemma:sequentialization-icomll} as follows.
 
 \sequentializationICOMLL*

\end{document}